\documentclass[sigconf,10pt]{acmart}
\setcopyright{rightsretained}

\usepackage{comment}
\usepackage{amsmath}
\usepackage{amssymb}
\usepackage{amsthm}
\usepackage{mathtools}
\usepackage{etoolbox}

\usepackage{stmaryrd}
\usepackage[normalem]{ulem}
\usepackage{subcaption}
\usepackage{booktabs}
\usepackage{graphicx}
\usepackage{listings}
\usepackage{fancyvrb}
\usepackage{caption}
\usepackage{subcaption}
\usepackage{braket}
\usepackage[inline]{enumitem}
\usepackage{xspace}
\usepackage{colortbl}
\usepackage{url}
\usepackage{hyperref}
\hypersetup{breaklinks=true}
\usepackage{hyphenat}

\usepackage{url}

\usepackage{breakurl}
\usepackage{multirow}
\usepackage{makecell}
\usepackage{cleveref}

\graphicspath{ {figures/} }

\usepackage[disable]{todonotes}

\usepackage{upquote}

\definecolor{black}{rgb}{0,0,0}
\definecolor{grey}{rgb}{0.8,0.8,0.8}
\definecolor{red}{rgb}{1,0,0}
\definecolor{green}{rgb}{0,1,0}
\definecolor{darkgreen}{rgb}{0,0.5,0}
\definecolor{darkpurple}{rgb}{0.5,0,0.5}
\definecolor{darkdarkpurple}{rgb}{0.3,0,0.3}
\definecolor{blue}{rgb}{0,0,1}
\definecolor{shadegreen}{rgb}{0.95,1,0.95}
\definecolor{shadeblue}{rgb}{0.95,0.95,1}
\definecolor{shadered}{rgb}{1,0.85,0.85}
\definecolor{shadegrey}{rgb}{0.85,0.85,0.85}
\definecolor{oddRowGrey}{rgb}{0.80,0.80,0.80}
\definecolor{evenRowGrey}{rgb}{0.85,0.85,0.85}

\newcommand{\BG}[1]{\todo[inline]{\textbf{Boris says:$\,$} #1}}

\newbool{ShowDetailedProofs}
\newcommand{\detailedproof}[2]{\ifbool{ShowDetailedProofs}{
\begin{proof}
#2
\end{proof}
}{\noindent\textit{Proof Sketch:}#1\qed}\smallskip}
\newcommand{\ifnotdetailedproof}[1]{\ifbool{ShowDetailedProofs}{}{#1}}
\newcommand{\ifnottechreport}[1]{\ifbool{ShowDetailedProofs}{}{#1}}
\newrobustcmd{\iftechreport}[1]{\ifbool{ShowDetailedProofs}{#1}{}}
\newcommand{\ProofPointerToAppendix}{
  \vspace*{-3.5mm}
  \begin{proof} See \Cref{sec:proofs} \end{proof}
}
\newcommand{\ProofPointerToAppendixWithSketch}[1]{
  \vspace*{-1mm}
  \noindent \textsc{Proof Sketch.} #1 \hfill \textit{(See \Cref{sec:proofs} for full proof)}
}

\newcommand{\mypar}[1]{\smallskip\noindent\textbf{{#1}.}}
\newcommand{\bfcaption}[1]{\caption{#1}}
\newcommand{\trimfigurespacing}{\vspace*{-5mm}}

\DeclareMathAlphabet{\mathbbold}{U}{bbold}{m}{n}

\newtheorem{Theorem}{Theorem}
\newtheorem{Definition}{Definition}
\newtheorem{Lemma}{Lemma}
\newtheorem{Proposition}{Proposition}
\newtheorem{Corollary}{Corollary}
\newtheorem{Example}{Example}

\newcommand{\proofpara}[1]{\medskip\noindent\underline{{#1}:}}

\newcommand{\mathtext}[1]{\ensuremath{\,\text{#1}\,}}
\newcommand{\card}[1]{\vert{#1}\vert}
\newcommand{\defas}{\stackrel{def}{=}}

\DeclareMathAlphabet{\mathbbold}{U}{bbold}{m}{n}

\newcommand{\schema}{\mathbf{D}}
\newcommand{\relschema}{\mathbf{R}}
\newcommand{\schemaOf}{\textsc{Sch}}
\newcommand{\arity}[1]{arity({#1})}
\newcommand{\db}{D}
\newcommand{\rel}{R}
\newcommand{\query}{Q}
\newcommand{\qClass}{\mathcal{C}}
\newcommand{\tup}{t}

\newcommand{\aDom}{\mathbb{D}}

\newcommand{\tupDom}{\textsc{TupDom}}
\newcommand{\tuple}[1]{\left<\;\textsf{\upshape #1}\;\right>}
\newcommand{\simpletuple}[1]{\left<\;#1\;\right>}
\newcommand{\comprehension}[2]{\left\{\;#1\;|\;#2\;\right\}}

\newcommand{\selection}{\sigma}
\newcommand{\projection}{\pi}
\newcommand{\join}{\Join}
\newcommand{\union}{\cup}

\newcommand{\raPlus}{\mathcal{RA}^{+}}

\newcommand{\abbrUADB}{UA-DB\xspace}
\newcommand{\abbrUADBs}{UA-DBs\xspace}
\newcommand{\abbrTI}{TI-DB\xspace}
\newcommand{\abbrTIs}{TI-DBs\xspace}
\newcommand{\termTIs}{tuple-independent databases\xspace}

\newcommand{\abbrCoddTable}{Codd-table\xspace}
\newcommand{\abbrCoddTables}{Codd-tables\xspace}

\newcommand{\abbrVtables}{V-tables\xspace}
\newcommand{\abbrCtable}{C-table\xspace}
\newcommand{\abbrCtables}{C-tables\xspace}
\newcommand{\abbrPCtable}{PC-table\xspace}
\newcommand{\abbrPCtables}{PC-tables\xspace}

\newcommand{\abbrVCtables}{Virtual C-tables\xspace}
\newcommand{\abbrXDB}{x-DB\xspace}
\newcommand{\abbrXDBs}{x-DBs\xspace}

\newcommand{\termBG}{best-guess\xspace}
\newcommand{\termBGQP}{best-guess query processing\xspace}
\newcommand{\termBGW}{best-guess world\xspace}
\newcommand{\termBGWs}{best-guess worlds\xspace}
\newcommand{\termBGA}{best-guess answers\xspace}

\newcommand{\abbrBGW}{BGW\xspace}
\newcommand{\abbrBGQP}{BGQP\xspace}

\newcommand{\ptime}{\texttt{PTIME}\xspace}

\newcommand{\conpcomplete}{coNP-complete\xspace}
\newcommand{\conphard}{coNP-hard\xspace}
\newcommand{\sharpP}{\texttt{\#P}\xspace}

\newcommand{\cSet}{certain}
\newcommand{\pSet}{possible}

\newcommand{\pdb}{\mathcal{D}}

\newcommand{\uadb}{\db_{UA}}

\newcommand{\bgdb}{D_{bg}}
\newcommand{\prob}{P}

\newcommand{\TUL}{\mathcal{L}}

\newcommand{\doTUL}{\textsf{label}}
\newcommand{\doTULTI}{\doTUL_{\textsc{\abbrTI}}}
\newcommand{\doTULC}{\doTUL_{\textsc{\abbrCtable}}}
\newcommand{\doTULX}{\doTUL_{\textsc{\abbrXDB}}}

\newcommand{\xTup}{\tau}

\newcommand{\xAttr}{X_{id}}
\newcommand{\altAttr}{Alt_{id}}
\newcommand{\pAttr}{P}
\newcommand{\lcAttr}{LC}

\newcommand{\semN}{\mathbb{N}}
\newcommand{\semB}{\mathbb{B}}

\newcommand{\semAC}{\mathbb{A}}
\newcommand{\semK}{\mathcal{K}}

\newcommand{\onesymbol}{\mathbbold{1}}
\newcommand{\zerosymbol}{\mathbbold{0}}
\newcommand{\multsymb}{\otimes}
\newcommand{\addsymbol}{\oplus}
\newcommand{\monsymbol}{\ominus}

\newcommand{\kDom}{K}
\newcommand{\addK}{\addsymbol_{\semK}}
\newcommand{\multK}{\multsymb_{\semK}}
\newcommand{\monK}{\monsymbol_{\semK}}
\newcommand{\oneK}{\onesymbol_{\semK}}
\newcommand{\zeroK}{\zerosymbol_{\semK}}
\newcommand{\addOf}[1]{\addsymbol_{#1}}
\newcommand{\multOf}[1]{\multsymb_{#1}}

\newcommand{\oneOf}[1]{\onesymbol_{#1}}
\newcommand{\zeroOf}[1]{\zerosymbol_{#1}}

\newcommand{\dbDomK}[1]{\mathcal{DB}_{#1}}

\newcommand{\ordersymbol}{\preceq}
\newcommand{\geqsymbol}{\succeq}
\newcommand{\ordK}{\ordersymbol_{\semK}}
\newcommand{\geqK}{\geqsymbol_{\semK}}
\newcommand{\ordOf}[1]{\ordersymbol_{#1}}

\newcommand{\lub}{\sqcup}
\newcommand{\glb}{\sqcap}

\newcommand{\glbKof}[1]{\sqcap_{#1}}
\newcommand{\glbK}{\glbKof{\semK}}
\newcommand{\lubKof}[1]{\sqcup_{#1}}
\newcommand{\lubK}{\lubKof{\semK}}
\newcommand{\LubK}{\sqcup_{\semK}}
\newcommand{\GlbK}{\sqcap_{\semK}}

\newcommand{\setK}{\vec k}

\newcommand{\dpK}[2]{{#1}^{#2}}
\newcommand{\dpDom}[2]{{#1}^{#2}}
\newcommand{\doubleK}[1]{\dpK{#1}{2}}
\newcommand{\aDoubleK}{\doubleK{\semK}}
\newcommand{\doubleDom}[1]{\dpDom{#1}{2}}

\newcommand{\uaName}{UA}
\newcommand{\uaK}[1]{{#1}_{\uaName}}

\newcommand{\uaAdd}[1]{\addOf{\uaK{#1}}}
\newcommand{\uaMult}[1]{\multOf{\uaK{#1}}}

\newcommand{\uae}[2]{[#2,#1]}

\newcommand{\pwIndicator}{W}
\newcommand{\pwDom}{W}
\newcommand{\pwkDom}{\kDom^{\pwDom}}
\newcommand{\pwK}{\semK_{\pwIndicator}}
\newcommand{\pwKof}[1]{{#1}_{\pwIndicator}}
\newcommand{\pwAdd}{{\addOf{\pwK}}}
\newcommand{\pwMult}{{\multOf{\pwK}}}
\newcommand{\pwZero}{{\zeroOf{\pwK}}}
\newcommand{\pwOne}{{\oneOf{\pwK}}}
\newcommand{\pwWorld}[1]{\textsc{pw}_{{#1}}}
\newcommand{\pwe}[1]{\vec{{#1}}}

\newcommand{\certainName}{\textsc{cert}}
\newcommand{\possibleName}{\textsc{poss}}

\newcommand{\pwCertain}{{\certainName}_{\semK}}

\newcommand{\pwPossible}{{\possibleName}_{\semK}}
\newcommand{\pwCertOf}[1]{{\certainName}_{#1}}

\newcommand{\pwCard}{n}

\newcommand{\cOfName}{cert}
\newcommand{\cOf}[1]{h_{\cOfName}({#1})}
\newcommand{\dOfName}{det}
\newcommand{\dOf}[1]{h_{\dOfName}({#1})}

\newcommand{\bagEnc}{\textsc{Enc}}
\newcommand{\rewrN}[1]{\llbracket\;#1\; \rrbracket_{UA}}
\newcommand{\cAttr}{C}
\newcommand{\uAttr}{U}

\DeclareMathSymbol{\mlq}{\mathord}{operators}{``}
\DeclareMathSymbol{\mrq}{\mathord}{operators}{`'}
\DeclareMathOperator*{\argmax}{argmax}

\booltrue{ShowDetailedProofs}

\def\BibTeX{{\rm B\kern-.05em{\sc i\kern-.025em b}\kern-.08em
    T\kern-.1667em\lower.7ex\hbox{E}\kern-.125emX}}

\usepackage{etoolbox}

\title{Uncertainty Annotated Databases - A Lightweight Approach for Approximating Certain Answers}

\definecolor{revgreen}{rgb}{0,0.5,0}
\newrobustcmd{\reva}[1]{#1}
\newrobustcmd{\revb}[1]{#1}
\newrobustcmd{\revc}[1]{#1}
\newrobustcmd{\revm}[1]{#1}

\author{}

\begin{abstract}
Certain answers are a principled method for coping with uncertainty that arises in many practical data management tasks.
Unfortunately, this method is expensive and may exclude useful (if uncertain) answers. 
Thus, users frequently resort to less principled approaches 
to resolve the uncertainty. 
In this paper, we propose \emph{Uncertainty Annotated Databases} (\abbrUADBs), which combine an under- and over-approximation of certain answers to achieve the reliability of certain answers, with the performance of a classical database system.
Furthermore, in contrast to prior work on certain answers, \abbrUADBs achieve a higher utility by including some (explicitly marked) answers that are not certain.
\abbrUADBs are based on incomplete K-relations, which we introduce to  generalize the classical set-based notions of incomplete databases and certain answers to a much larger class of data models. 
Using an implementation of our approach, we
demonstrate experimentally  that it efficiently produces
tight approximations of certain answers that are of high utility.
\end{abstract}

\definecolor{lstpurple}{rgb}{0.5,0,0.5}
\definecolor{lstred}{rgb}{1,0,0}
\definecolor{lstreddark}{rgb}{0.7,0,0}
\definecolor{lstredl}{rgb}{0.64,0.08,0.08}
\definecolor{lstmildblue}{rgb}{0.66,0.72,0.78}
\definecolor{lstblue}{rgb}{0,0,1}
\definecolor{lstmildgreen}{rgb}{0.42,0.53,0.39}
\definecolor{lstgreen}{rgb}{0,0.5,0}
\definecolor{lstorangedark}{rgb}{0.6,0.3,0}
\definecolor{lstorange}{rgb}{0.75,0.52,0.005}
\definecolor{lstorangelight}{rgb}{0.89,0.81,0.67}
\definecolor{lstbeige}{rgb}{0.90,0.86,0.45}

\DeclareFontShape{OT1}{cmtt}{bx}{n}{<5><6><7><8><9><10><10.95><12><14.4><17.28><20.74><24.88>cmttb10}{}

\lstdefinestyle{psql}
{
tabsize=2,
basicstyle=\small\upshape\ttfamily,
language=SQL,
morekeywords={PROVENANCE,BASERELATION,INFLUENCE,COPY,ON,TRANSPROV,TRANSSQL,TRANSXML,CONTRIBUTION,COMPLETE,TRANSITIVE,NONTRANSITIVE,EXPLAIN,SQLTEXT,GRAPH,IS,ANNOT,THIS,XSLT,MAPPROV,cxpath,OF,TRANSACTION,SERIALIZABLE,COMMITTED,INSERT,INTO,WITH,SCN,UPDATED,LENS,SCHEMA_MATCHING,string,WINDOW,max,OVER,PARTITION,FIRST_VALUE,WITH},
extendedchars=false,
keywordstyle=\bfseries,
mathescape=true,
escapechar=@,
sensitive=true
}

\lstdefinestyle{psqlcolor}
{
tabsize=2,
basicstyle=\small\upshape\ttfamily,
language=SQL,
morekeywords={PROVENANCE,BASERELATION,INFLUENCE,COPY,ON,TRANSPROV,TRANSSQL,TRANSXML,CONTRIBUTION,COMPLETE,TRANSITIVE,NONTRANSITIVE,EXPLAIN,SQLTEXT,GRAPH,IS,ANNOT,THIS,XSLT,MAPPROV,cxpath,OF,TRANSACTION,SERIALIZABLE,COMMITTED,INSERT,INTO,WITH,SCN,UPDATED},
extendedchars=false,
keywordstyle=\bfseries\color{lstpurple},
deletekeywords={count,min,max,avg,sum},
keywords=[2]{count,min,max,avg,sum},
keywordstyle=[2]\color{lstblue},
stringstyle=\color{lstreddark},
commentstyle=\color{lstgreen},
mathescape=true,
escapechar=@,
sensitive=true
}

\lstdefinestyle{datalog}
{
basicstyle=\footnotesize\upshape\ttfamily,
language=prolog
}

\lstdefinestyle{pseudocode}
{
  tabsize=3,
  basicstyle=\small,
  language=c,
  morekeywords={if,else,foreach,case,return,in,or},
  extendedchars=true,
  mathescape=true,
  literate={:=}{{$\gets$}}1 {<=}{{$\leq$}}1 {!=}{{$\neq$}}1 {append}{{$\listconcat$}}1 {calP}{{$\cal P$}}{2},
  keywordstyle=\color{lstpurple},
  escapechar=&,
  numbers=left,
  numberstyle=\color{lstgreen}\small\bfseries,
  stepnumber=1,
  numbersep=5pt,
}

\lstdefinestyle{xmlstyle}
{
  tabsize=3,
  basicstyle=\small\upshape\ttfamily,
  language=xml,
  extendedchars=true,
  mathescape=true,
  escapechar=£,
  tagstyle=\bfseries,
  usekeywordsintag=true,
  morekeywords={alias,name,id},
  keywordstyle=\color{lstred}
}

\lstdefinestyle{xmlstyle-color}
{
  tabsize=3,
  basicstyle=\small\upshape\ttfamily,
  language=xml,
  extendedchars=true,
  mathescape=true,
  escapechar=£,
  tagstyle=\color{keywordpurple},
  usekeywordsintag=true,
  morekeywords={alias,name,id},
  keywordstyle=\color{lstred}
}


\begin{document}

\title{Uncertainty Annotated Databases - A Lightweight Approach for Approximating Certain Answers\\
(extended version)}

\author{Su Feng}
\affiliation{%
  \institution{Illinois Institute of Technology}
}
\email{sfeng14@hawk.iit.edu}

\author{Aaron Huber}
\affiliation{%
  \institution{University at Buffalo}
}
\email{ahuber@buffalo.edu}

\author{Boris Glavic}
\affiliation{%
  \institution{Illinois Institute of Technology}
}
\email{bglavic@iit.edu}

\author{Oliver Kennedy}
\affiliation{%
  \institution{University at Buffalo}
}
\email{okennedy@buffalo.edu}

\lstset{style=psql}


\maketitle

\section{Introduction}
\label{sec:intro}

Data uncertainty arises naturally in applications like sensing~\cite{DBLP:conf/icde/LetchnerRBP09},
data exchange~\cite{Fagin:2011:PDE:1989727.1989729}, distributed computing~\cite{DBLP:conf/sigmod/LangNRN14}, 
data cleaning~\cite{Chu:2015:KDC:2723372.2749431},
and many others.
Incomplete~\cite{DBLP:journals/jacm/ImielinskiL84} and probabilistic databases~\cite{suciu2011probabilistic} have emerged as a principled way to deal with uncertainty.
Both types of databases consist of a set of deterministic instances called \textit{possible worlds} that represent possible interpretations of  data available about the real world.
An often cited, conservative approach to uncertainty is to consider only \emph{certain answers}~\cite{AK91,DBLP:journals/jacm/ImielinskiL84} (answers in all possible worlds). 
However, this approach has two problems.
First, computing certain answers is expensive\footnote{\conpcomplete~\cite{AK91,DBLP:journals/jacm/ImielinskiL84} (data complexity) for first-order queries over V-tables~\cite{DBLP:journals/jacm/ImielinskiL84}, as well as for conjunctive queries for, 
e.g., OR-databases~\cite{DBLP:journals/jcss/ImielinskiMV95}.}.
Furthermore, requiring answers to be certain
may unnecessarily exclude useful, possible answers.
Thus, users instead resort to what we term \textit{\termBGQP} (\textit{\abbrBGQP}): making an educated guess about which possible world to use (i.e., how to interpret available data) and then working exclusively with this world.
\abbrBGQP is more efficient than certain answers, and generally includes more useful results.
However, information about uncertainty in the data is lost, and all query results produced by \abbrBGQP are consequently suspect.

Previous work has also explored approximations of certain answers~\cite{L16a,GP17,R86}.
Under the premise that missing a certain answer is better than incorrectly reporting an answer as certain, such work focuses 
on under\hyp{}approximating certain answers.
\revm{This addresses the performance problem, but under\hyp{}approximations only exacerbate the problem of excluded results.
Worse, these techniques are limited to specific uncertain data models such as V-tables, 
and with the 
exception of a brief discussion in~\cite{GL17}, only support set semantics.}


\begin{figure}[t]
  \centering
  \includegraphics[width=0.7\linewidth]{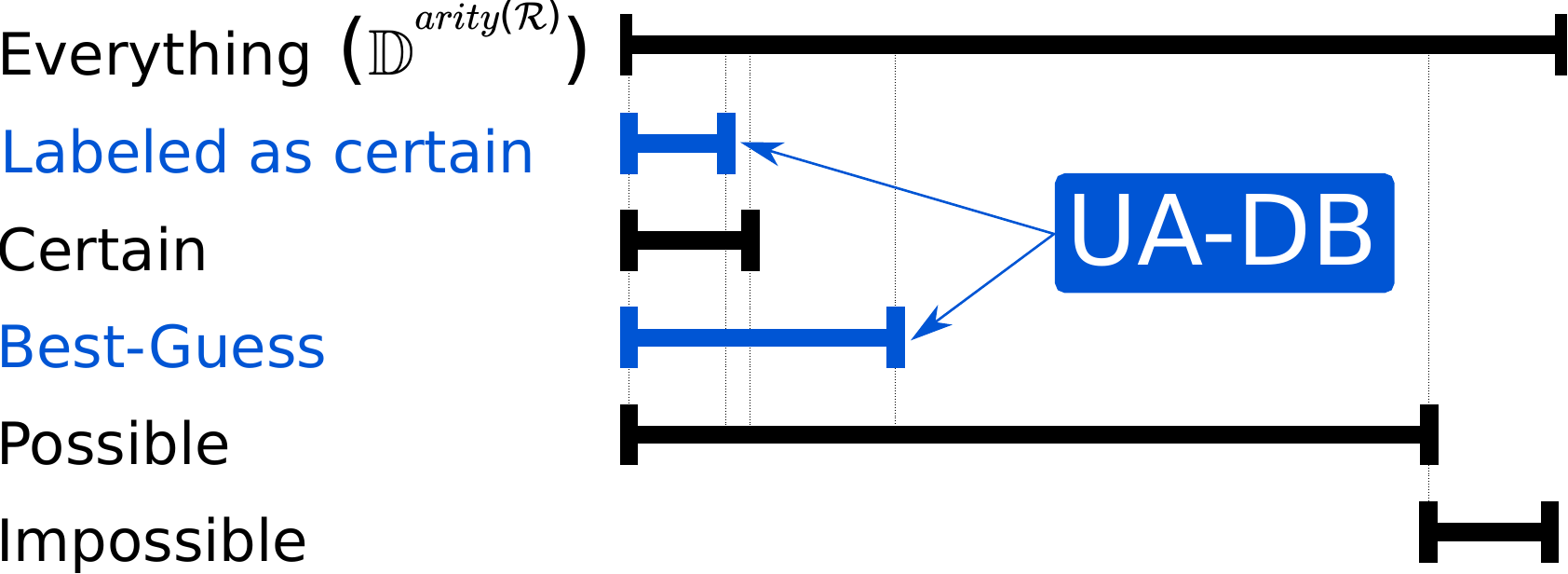}
  \caption{\abbrUADBs provide both an under- and over-approximation of certain answers.}
  \label{fig:relationship-of-UADB}
  \trimfigurespacing
\end{figure}

\begin{figure*}[t]
    \centering\small
    \begin{subfigure}{0.49\textwidth}
      \centering
      \textbf{\underline{ADDR}}\\
      \begin{tabular}{ccl}
        \textbf{id} & \textbf{address} & \textbf{geocoded} \\ \hline \\[-2mm]
        1 & 51 Comstock & (42.93, -78.81)\\
        2 & Grant at Ferguson & (42.91, -78.89) \textbf{or} (32.25, -110.87)\\
        3 & 499 Woodlawn & (42.91, -78.84) \textbf{or} (42.90, -78.85)\\
        4 & 192 Davidson & (42.93, -78.80) \\
      \end{tabular}
    \end{subfigure}
    \hspace{10mm}
    \begin{subfigure}{0.4\textwidth}
      \centering
      \textbf{\underline{LOC}}\\[-2mm]
      \begin{tabular}{ccc}
        \textbf{locale} & \textbf{state} & \textbf{rect} \\ \hline
        Lasalle & NY     & {\footnotesize ((42.93, -78.83), (42.95, -78.81))   }\\
        Tucson & AZ      & {\footnotesize ((31.99, -111.045), (32.32, -110.71))}\\
        Grant Ferry & NY & {\footnotesize ((42.91, -78.91), (42.92, -78.88))   }\\
        Kingsley & NY    & {\footnotesize ((42.90, -78.85), (42.91, -78.84))   }\\
        Kensington & NY  & {\footnotesize ((42.93, -78.81), (42.96, -78.78))   }\\
      \end{tabular}
    \end{subfigure}
    \vspace{-4mm}
    \caption{\revc{Input data for \Cref{ex:geocoding}. Tuples 2 and 3 of Table ADDR have uncertain geocoded values.}}
    \label{fig:geocoding}
  \trimfigurespacing
\end{figure*}

\begin{figure*}
  \begin{subfigure}{0.22\textwidth}
    \centering\small
    \begin{tabular}{ccc}
      \textbf{id} & \textbf{locale} & \textbf{state} \\ \hline
      1 & Lasalle & NY \\
      2 & Tucson & AZ \\
      3 & Kingsley & NY \\
      4 & Kensington & NY \\
    \end{tabular}\\[-2mm]
    \bfcaption{One Possible World}
    \label{fig:example:possworld}
  \end{subfigure}
  \begin{subfigure}{0.2\textwidth}
    \centering\small
    \begin{tabular}{ccc}
    \textbf{id} & \textbf{locale} & \textbf{state} \\ \hline
    1 & Lasalle & NY \\
    3 & Kingsley & NY \\
    4 & Kensington & NY \\
    \end{tabular}\\[-2mm]
    \bfcaption{Certain Answers}
    \label{fig:example:certain}
  \end{subfigure}
      \begin{subfigure}{0.2\textwidth}
    \centering\small
    \begin{tabular}{ccc}
    \textbf{id} & \textbf{locale} & \textbf{state} \\ \hline
    1 & Lasalle & NY \\
  2 & Tucson & AZ \\
  2 & Grant Ferry & NY \\
      3 & Kingsley & NY \\
    4 & Kensington & NY \\
    \end{tabular}\\[-2mm]
    \bfcaption{Possible Answers}
    \label{fig:example:possible}
  \end{subfigure}
  \begin{subfigure}{0.32\textwidth}
    \centering\small
    \begin{tabular}{ccc|c}
      \textbf{id} & \textbf{locale} & \textbf{state} & \textbf{Certain?} \\ \hline 
\cellcolor{shadegreen}      1 & \cellcolor{shadegreen}Lasalle & \cellcolor{shadegreen}NY & $true$\\
\cellcolor{shadegreen}      2 & \cellcolor{shadegreen}Tucson & \cellcolor{shadegreen}AZ & $false$\\
\cellcolor{shadered}      3 & \cellcolor{shadered}Kingsley & \cellcolor{shadered}NY & $false$\\
\cellcolor{shadegreen}      4 & \cellcolor{shadegreen}Kensington & \cellcolor{shadegreen}NY & $true$ \\
    \end{tabular}\\[-2mm]
    \bfcaption{Uncertainty Annotated Database}
    \label{fig:example:uadb}
  \end{subfigure}
  \vspace*{-4mm}
  \bfcaption{\revc{Examples of query results under different evaluation semantics over uncertain data.}}
  \label{fig:example}
  \trimfigurespacing\vspace{1mm}
\end{figure*}

\begin{Example} \label{ex:geocoding}
  Geocoders translate natural language descriptions of locations into 
  coordinates (i.e., latitude and longitude).
  \revc{
    Consider the \textbf{ADDR} and \textbf{LOC} relations in \Cref{fig:geocoding}.
    Tuples 2 and 3 of \textbf{ADDR} each have an ambiguous geocoding.
  }
  This is an x-table~\cite{DBLP:conf/vldb/AgrawalBSHNSW06}, a type of incomplete data model where each tuple may have multiple alternatives. Each possible world is defined by some combination of alternatives (e.g., \textbf{ADDR} encodes 4 possible worlds).
  An analyst might use a spatial join with a lookup table (\textbf{LOC}) to map coordinates to geographic regions.
  \revc{Figure~\ref{fig:example:possworld} shows the result of the following query in one world.}
\begin{lstlisting}
SELECT a.id, l.locale, l.state
FROM ADDR a, LOC l
WHERE contains(l.rect, a.geocoded)
\end{lstlisting}
  \revc{
  The certain answers to this query (Figure~\ref{fig:example:certain}) are tuples that appear in the result, regardless of which world is queried.
  Figure~\ref{fig:example:possible} shows all possible answers that could be returned for some choice of geocodings.
  Note also that ambiguous answers (e.g., address 2) may not be certain, but may still be useful.}
\end{Example}

Ideally, we would like an approach that (1) generalizes to a wide range of data models, (2) is easy to use like \abbrBGQP, \revm{(3) is compatible with a wide of probabilistic and incomplete data representations (e.g., \termTIs{}~\cite{suciu2011probabilistic}, \abbrCtables{}~\cite{DBLP:journals/jacm/ImielinskiL84}, and \abbrXDBs{}~\cite{DBLP:conf/vldb/AgrawalBSHNSW06}) and sources of uncertainty (e.g., inconsistent databases~\cite{DBLP:conf/pods/KoutrisW18,DBLP:journals/pvldb/KolaitisPT13,DBLP:journals/ipl/KolaitisP12,DBLP:series/synthesis/2011Bertossi,FM05,DBLP:conf/pods/ArenasBC99}, imputation of missing values, and more)}, and \revm{(4)} is principled like certain answers.
We address the generality requirement (1) by rethinking incomplete data management in terms of Green et. al.'s $\semK$-database framework~\cite{Green:2007:PS:1265530.1265535}.
In this framework, each tuple is annotated with an value from a semiring $\semK$.
Choosing an appropriate  semiring, $\semK$-databases can encode a wide range of query processing semantics including classical set- and bag-semantics, as well as query processing with access control, provenance, and more.
Our primary contribution here is to identify a natural, backwards-compatible generalization of certain answers to a broad class of $\semK$-databases.


Our second major contribution is to combine an under\hyp{}approximation of certain answers with \termBGQP to create an \textit{Uncertainty-Annotated Database} (\textit{\abbrUADB}).
A \abbrUADB is built around one distinguished possible world of an incomplete $\semK$-database, \revb{for instance the ``\termBG{}'' world that would normally be used in practice.}
This world serves as an \textit{over\hyp{}approximation} of  certain answers.
Tuples from this world are labeled as either certain or uncertain to encode an \textit{under\hyp{}approximation} of  certain answers.
As illustrated in Figure~\ref{fig:relationship-of-UADB}, a UA-DB sandwiches the 
certain answers between under- and over-approximations.
A lightweight (extensional~\cite{suciu2011probabilistic}) query evaluation semantics then propagates labels while preserving these approximations.

\begin{figure}[t]
  \centering
  \includegraphics[width=1\linewidth]{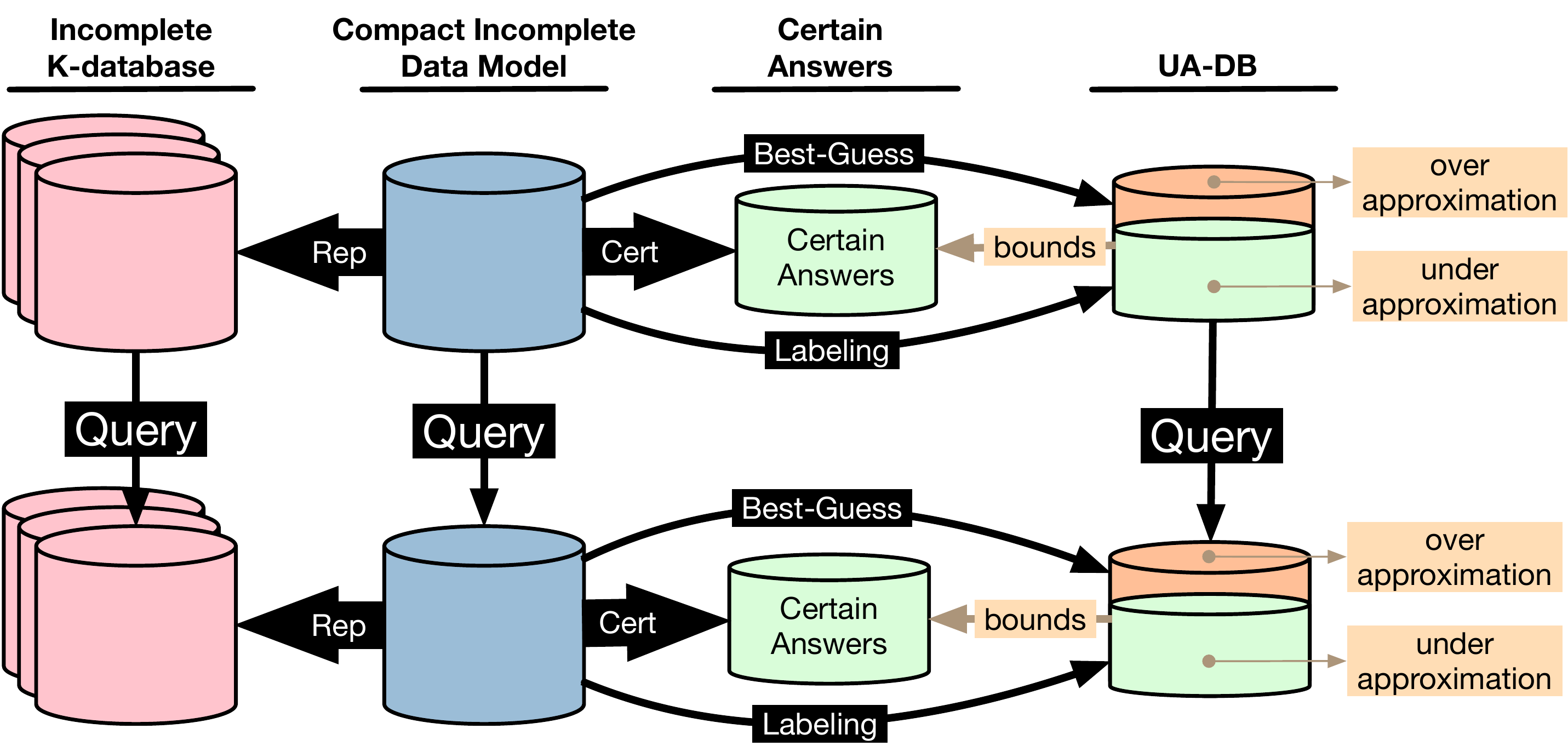}
  \vspace*{-6mm}
  \caption{The relationship between \abbrUADBs, certain answers, and other incomplete data models}
  \label{fig:overview-of-the-UADB}
  \trimfigurespacing
  \vspace{-1mm}
\end{figure}

\begin{Example}\label{ex:geocoding-uadb}
 Continuing with Example~\ref{ex:geocoding},
 Figure~\ref{fig:example:uadb} shows the result of the same query as a set \abbrUADB.
 When the \abbrUADB is built, one designated possible world of \textbf{ADDR} is selected, \revb{for example the highest ranked option provided by the geocoder.}
 For this example, we select the first option for each ambiguous tuple.
 The result is based on this one designated possible world, which serves as an over-approximation of the certain answers.
 A subset of these tuples (addresses 1 and 4) are explicitly labeled as certain.
 This is the under-approximation: A tuple might still be certain even if it is not labeled as such.
We consider the remaining tuples to be ``uncertain''.
In Figure~\ref{fig:example:uadb}, tuples 1 and 4 (resp., 2) are correctly marked as certain (resp., uncertain), while tuple 3 is mis-classified as uncertain even though it appears in all worlds.
We stress that even a mislabeled certain answer is still present: a \abbrUADB sandwiches the certain answers.
\end{Example}

Figure~\ref{fig:overview-of-the-UADB} overviews our approach.
\revm{We provide \emph{labeling schemes} that derive a \abbrUADB from common incomplete data models.  The resulting \abbrUADB bounds the certain tuples from above and below, a property preserved through queries.
}
\abbrUADBs are both efficient and precise.
We demonstrate efficiency by implementing a bag \abbrUADB as a query-rewriting front-end on top of a classical relational DBMS:
\abbrUADB queries have minimal performance overhead compared to the same queries 
on deterministic data.
We demonstrate precision both analytically and experimentally.
First, under specific conditions, some of which we identify in \Cref{sec:pres-c-compl}, exactly the certain answers will be marked as certain.
Second, we show experimentally that even when these conditions do not hold, the fraction of misclassified certain answers is 
low. 
Importantly, a wide range of uncertain data models can be translated into \abbrUADBs through simple and efficient transformations \revb{that (i) determine a \termBGW{} (\textit{\abbrBGW{}}) and (ii) obtain an under-approximation of the certain answers.
We define such transformations for three popular models of incomplete data in \Cref{sec:create-labelings}: \termTIs{}~\cite{suciu2011probabilistic},
\abbrXDBs~\cite{DBLP:conf/vldb/AgrawalBSHNSW06} 
and \abbrCtables~\cite{DBLP:journals/jacm/ImielinskiL84}.
In classical incomplete databases, where probabilities are not available, any possible world can serve as a \abbrBGW{}.
In probabilistic databases (or any incomplete data model that ranks possible worlds), we preferentially use the possible world with the highest probability (if computationally feasible), or an approximation thereof.
We emphasize that our approach does not require enumerating (or even knowing) the full set of possible worlds.
As long as some possible world can be obtained, our approach is applicable.
In worst case, if no certainty information is available, our approach labels all tuples as uncertain and degrades to classical \termBGQP{}.
}
\revc{Furthermore, our approach is also applicable in use cases like inconsistent query answering~\cite{DBLP:conf/pods/ArenasBC99} where possible worlds are defined declaratively (e.g., all repairs of an inconsistent database).} 

\revm{We significantly extend
 the state-of-the-art on under\--ap\-prox\-i\-ma\-ting certain answers~\cite{L16a,GP17,R86}:
(1) we combine an under-approximation with best-guess query processing bounding certain answers from above and below;
(2) we support sets, bags, and any other data model expressible as semiring annotations from a large class of semirings;  
(3) we support translation of a wide range of incomplete and probabilistic data models into our \abbrUADB model;
(4) in contrast to certain answers, \abbrUADBs are closed under queries.
}
The remainder of the paper is organized as follows.

\mypar{Incomplete $\semK$-Relations} (Section~\ref{sec:incompl-prob-k})
We introduce incomplete $\semK$-databases, generalizing incomplete databases to $\semK$-relations~\cite{Green:2007:PS:1265530.1265535}.
We then define certain annotations
as a natural extension of certain answers, based on the observation that certain answers are a lower bound on the content of a 
world.
\iftechreport{It is thus natural to define certainty based on a greatest-lower-bound operation (GLB) for semiring annotations based on so-called l-semirings~\cite{DBLP:conf/icdt/KostylevB12} where the GLB is well behaved.}
We show that certain annotations 
correspond to 
the classical notion of certain answers for 
set~\cite{L79a} and bag~\cite{GL17} semantics.

\mypar{\abbrUADBs} (Section~\ref{sec:UA-model})
We define \textit{\abbrUADBs} as databases that annotate tuples with pairs of annotations from a semiring $\semK$.
The annotation of a tuple in a \abbrUADB bounds the certain annotation of the tuple from above and below.
This is achieved by combining the annotations from one 
world (the over-approximation) with an under\hyp{}approximation that we call an \textit{uncertainty labeling}.
Relying on results for under-approximations that we develop in the following sections,
we prove that queries over \abbrUADBs preserve these bounds.

\mypar{Under\hyp{}approximating Certain Answers} (Section~\ref{sec:uncert-labell-1})
To better understand under\hyp{}approximations,
we 
define \textit{uncertainty labelings}, which are $\semK$-relations that
under\hyp{}approximate the set of certain tuples for an incomplete $\semK$-database.
\iftechreport{An uncertainty labeling is certain- or \textit{c-sound} (resp., \textit{c-complete}) if it is a lower (resp., upper) bound on the certain annotations of tuples in a $\pwK$-relation; and \textit{c-correct} if it is both.
We also extend these definitions to query 
semantics. A query semantics preserves c-soundness if the result of the query is a c-sound labeling for the result of evaluating the query over the input $\pwK$-database from which the labeling was derived. }


\mypar{Queries over Uncertainty Labelings} (Section~\ref{sec:c-soundess})
Since labelings are $\semK$-relations, we can evaluate queries over such labelings.
We demonstrate that evaluating queries in this fashion preserves under\hyp{}approximations of certain answers, generalizing a previous result for V-tables due to Reiter~\cite{R86}.  
\iftechreport{That is queries preserve c-soundness. }
\iftechreport{
Furthermore, 
under certain conditions 
this query semantics 
returns precisely the certain answers. 
That is, since all queries preserve c-soundness, under these conditions queries preserve c-correctness. }


\mypar{Implementation for Bag Semantics} (Section~\ref{sec:implementation})
We implement \abbrUADBs on top of a 
relational DBMS. We extend the schema of relations 
to label 
tuples as 
certain or uncertain (e.g., Figure~\ref{fig:example:uadb}). Queries with UA-relational semantics are compiled into standard relational queries over this encoding. \iftechreport{We prove this compilation process to be correct.}

\mypar{Performance} (Section~\ref{sec:experiments}) We demonstrate experimentally that \abbrUADBs outperform state-of-the-art incomplete \revm{and probabilistic} query processing schemes, and are competitive with  deterministic query evaluation and other certain answer under\hyp{}approximations.
Furthermore, for a wide range of real world datasets, comparatively few answers are misclassified by our approach. We also demonstrate that \termBGA{} and, hence, also \abbrUADBs{}, can have higher utility than certain answers. \revm{Finally, we demonstrate the use of \abbrUADBs{} for uncertain access control annotations and bag semantics.}


\section{Notation and Background}\label{sec:background}

A database schema $\schema = \{\relschema_1, \ldots, \relschema_n\}$ is a set of relation schemas.
A relational schema $\relschema(A_1, \ldots, A_n)$ consists of a relation name and a set of attribute names $A_1, \dots, A_n$.
The arity $\arity{\relschema}$ of a relation schema $\relschema$ is the number of attributes in $\relschema$.
An  instance $\db$ for database schema $\schema$ is a set of relation instances with one relation for each relation schema in $\schema$: $\db=\{\rel_1, \dots, \rel_n\}$.
Assume a universal domain of attribute values $\aDom$. A tuple with schema $\relschema$ is an element from $\aDom^{\arity{\relschema}}$.
In this paper, we consider both bag and set semantics. A set relation $\rel$ with schema $\relschema$ is a set of tuples with schema $\relschema$, i.e.,  $\rel \subseteq \aDom^{\arity{\relschema}}$. A bag relation $\rel$ with schema $\relschema$ is  a bag (multiset) of tuples with schema $\relschema$.
We use $\tupDom$ to denote the set of all tuples over domain $\aDom$.

\subsection{Possible Worlds Semantics}
\label{sec:prob-query-proc}

Incomplete and probabilistic databases model
uncertainty and its impact on query results. 
An \emph{incomplete database} 
$\pdb$  is a set of deterministic database instances $\db_1, \ldots, \db_n$ of schema $\schema$, called \emph{possible worlds}.
We write $\tup \in \db$ to denote that a tuple $\tup$ appears in a specific possible world $\db$.


\begin{Example} \label{example:pdb}
  Continuing Example~\ref{ex:geocoding}, Figure~\ref{table:pdb} shows the two possible worlds in the result of the spatial join.
  Observe that some tuples (e.g., $\tuple{1, Lasalle, NY}$) appear in all worlds.  
  Such tuples are called \textbf{certain}.
  Tuples that appear in at least one possible world (e.g., $\tuple{2, Tuscon, AZ}$) are called \textbf{possible}.
\end{Example}

\begin{figure}[t]
\centering
  \begin{subtable}{.49\linewidth}
    \centering
    \resizebox{!}{9mm}{
      \begin{tabular}{ccc}
        \textbf{id} & \textbf{locale} & \textbf{state} \\ \hline
        1 & Lasalle & NY\\
        \textbf{2} & \textbf{Tucson} & \textbf{AZ}\\
        3 & Kingsley & NY\\
        4 & Kensington & NY\\
      \end{tabular}
    }
    \vspace*{-1mm}
    \caption{$D_1$}
    \label{table:example-PDB-world-1}
  \end{subtable}
  \begin{subtable}{.49\linewidth}
    \label{table:example-PDB-world-2}
    \resizebox{!}{9mm}{
      \begin{tabular}{ccc}
        \textbf{id} & \textbf{locale} & \textbf{state} \\ \hline
        1 & Lasalle & NY\\
        \textbf{2} & \textbf{Grant Ferry} & \textbf{NY}\\
        3 & Kingsley & NY\\
        4 & Kensington & NY\\
      \end{tabular}
    }
    \vspace*{-1mm}
    \centering
    \caption{$D_2$}
  \end{subtable}\\[-4mm]
\bfcaption{Example incomplete database $\pdb=\{\db_1,\db_2\}$.}
\label{table:pdb}
\trimfigurespacing\vspace{2mm}
\end{figure}


Decades of research~\cite{suciu2011probabilistic,DBLP:journals/jacm/ImielinskiL84,DBLP:journals/debu/GreenT06,4221832,DBLP:conf/sigmod/BoulosDMMRS05,DBLP:conf/vldb/AgrawalBSHNSW06} has focused on 
query processing over incomplete databases. 
These techniques commonly adopt the  \textit{``possible worlds''} semantics:
The result of evaluating a query $\query$ over an incomplete database is the set of relations resulting from evaluating $\query$ over each possible world individually using deterministic 
semantics.
\begin{align}\label{eq:pqp}
\query(\pdb) = \comprehension{\query(\db)}{\db \in \pdb}
\end{align}
\begin{Example} \label{example:pqp}
Query $\query_{NY}=\selection_{state = 'NY'}(\pdb)$ returns locations in NY State from the database $\pdb$ shown in Figure~\ref{table:pdb}. 
  \noindent The result of $\query_{NY}(\pdb)$ is the set of worlds computed by evaluating $\query_{NY}$ over each world of $\pdb$ as shown in Figure \ref{table:pqp}.
  Observe that the location with id 2 appears in $\query_{NY}(\db_2)$, but not $\query_{NY}(\db_1)$.
\end{Example}

\begin{figure}[t]
\centering
  \begin{subtable}{.49\linewidth}
    \centering
    \resizebox{!}{9mm}{
  \begin{tabular}{ ccc }
    \textbf{id}  & \textbf{locale}  & \textbf{state}  \\
    \hline
    1 & Lasalle & NY\\
    3 & Kingsley & NY\\
    4 & Kensington & NY\\
    & & \\
  \end{tabular}
  }
  \vspace*{-1mm}
    \caption{$Q_{NY}(D_1)$}
  \label{table:example-PQP-result-world-1}
  \end{subtable}
  \begin{subtable}{.49\linewidth}
    \centering
    \resizebox{!}{9mm}{
  \begin{tabular}{ ccc }
    \textbf{id}  & \textbf{locale}  & \textbf{state}  \\
    \hline
    1 & Lasalle & NY\\
    2 & Grant Ferry & NY\\
    3 & Kingsley & NY\\
    4 & Kensington & NY\\
  \end{tabular}
  }
  \vspace*{-1mm}
    \caption{$Q_{NY}(D_2)$}
  \label{table:example-PQP-result-world-2}
  \end{subtable}\\[-2mm]
\bfcaption{The result of a query $\query$ over an incomplete database $\pdb$ is the set of results in all worlds $\db \in \pdb$.}
\label{table:pqp}
\trimfigurespacing
\end{figure}

\subsection{Certain and Best-Guess Answers} 
\label{sec:best-guess-query}
\label{sec:cert-poss-tupl}

An important goal of query processing over incomplete databases is to differentiate query results that are certain from ones that are merely possible.
Formally, a tuple is certain if it appears in every possible world. 
~\cite{DBLP:journals/jacm/ImielinskiL84,L79a}:
\ifnottechreport{
\begin{align}
  \cSet(\pdb) &= \{ \tup \mid \forall \db \in \pdb: \tup \in \db \} \label{eq:certain-tuple-set}
\end{align}
}
\iftechreport{
\begin{align}
  \cSet(\pdb) &= \{ \tup \mid \forall \db \in \pdb: \tup \in \db \} \label{eq:certain-tuple-set} \\
  \pSet(\pdb)&= \{ \tup \mid \exists \db \in \pdb: \tup \in \db \} \label{eq:possible-tuple-set}
\end{align}
}
In contrast to~\cite{DBLP:journals/jacm/ImielinskiL84}, which studies certain answers to queries, we define certainty at the instance level.
These approaches are equivalent since we can compute the certain answers of query $\query$ over incomplete instance $\pdb$ as $\cSet(\query(\pdb))$.
Although computing certain answers is \conphard~\cite{AK91} in general, there exist \ptime under\hyp{}approximations~\cite{L16a,GL16,R86}. 

\mypar{Best Guess Query Processing} As mentioned in the introduction, another approach commonly used in practice is to 
select one possible world.
Queries are evaluated solely in this world, and ambiguity is ignored or documented outside of the database.
We refer to this approach as \textit{best-guess query processing} (\abbrBGQP)~\cite{Yang:2015:LOA:2824032.2824055} since typically one would like to select the possible world that is deemed most likely.


\subsection{K-relations} \label{sec:data_provenance}
Our generalization of incomplete databases 
is based on the  \textbf{$\semK$-relation}~\cite{Green:2007:PS:1265530.1265535} framework.
In this framework, relations are annotated with elements from the domain $\kDom$ of a (commutative) semiring $\semK$.
   A commutative semiring is a structure $\semK=\simpletuple{\kDom,\addK,\multK,\oneK,\zeroK}$ with commutative and associative addition ($\addK$) and product ($\multK$) operations where $\addK$ distributes over $\multK$.
As before, $\aDom$ denotes a universal domain.
An $n$-nary $\semK$-relation is a function that maps tuples (elements from $\aDom^{n}$) to elements from $\kDom$.
Tuples that are not in the relation are annotated with $\zeroK$. 
Only finitely many tuples may be mapped to an element other than $\zeroK$ (i.e., relations must be finite).
Since $\semK$-relations are functions from tuples to annotations, it is customary to denote the annotation of a tuple $\tup$ in relation $\rel$ as $\rel(\tup)$. 
The specific information encoded by an annotation 
 depends on the choice of semiring.

\mypar{Encoding Sets and Bags}
Green et al.~\cite{Green:2007:PS:1265530.1265535} demonstrated that bag and set relations can be encoded as
commutative semirings: the natural numbers ($\semN$) with addition and multiplication, $\simpletuple{\semN, +, \times, 0, 1}$, annotates each tuple with its multiplicity; and boolean constants $\semB = \{T,F\}$ with disjunction and conjunction, $\simpletuple{\semB, \vee, \wedge, F, T}$, annotates each tuple with its set membership.
Abusing notation, we denote by $\semN$ and $\semB$ both the domain and the corresponding semiring.


\begin{figure}[t]
  \begin{subtable}{.49\linewidth}
    \centering
    \resizebox{\textwidth}{!}{
    \begin{tabular}{ ccc r r}
      \textbf{id} & \textbf{address} & \textbf{$\ell$} & \underline{$\semN$}  & \underline{$\semB$} \\
      \cline{1-3} \\[-3mm]
      1 & 51 Co\ldots & $L_1$ & $1$ & T\\
      2 & Grant \ldots & $L_2$ & $1$ &T \\
      3 & 499 W\ldots & $L_4$ & $1$ & T\\
    \end{tabular}
    }
    \bfcaption{Address}
  \end{subtable}
  \hspace{\fill}
  \begin{subtable}{.49\linewidth}
    \centering
    \resizebox{\textwidth}{!}{
    \begin{tabular}{ ccc r r}
      \textbf{$\ell$}  & \textbf{locale} & \textbf{state} & \underline{$\semN$}  & \underline{$\semB$} \\
      \cline{1-3} \\[-3mm]
          $L_1$ & L\ldots & NY & $1$ & T\\
          $L_2$ & T\ldots & AZ & $1$ & T\\
          $L_3$ & G\ldots & NY & $1$ & T\\
          $L_4$ & K\ldots & NY & $1$ & T\\
          $L_5$ & W\ldots & IL & $1$ & T\\
    \end{tabular}
    }
    \bfcaption{Neighborhood}
  \end{subtable}\\[-3mm]

  \begin{subtable}{\linewidth}
    \centering
    \begin{tabular}{ cll }
      \textbf{state} & \underline{$\semN$}  & \underline{$\semB$} \\
      \cline{1-1} \\[-3mm]
      NY & $2$ {\scriptsize $= (1 \cdot 1) + (1 \cdot 1)$} & $T$ {\scriptsize $= (T \wedge T) \vee (T \wedge T)$}\\
      AZ & $1$ {\scriptsize $= (1 \cdot 1)$} & $T$ {\scriptsize $= (T \wedge T)$}\\
      IL & $0$ {\scriptsize $= (0 \cdot 1)$} & $F$ {\scriptsize $= (F \wedge T)$}\\
    \end{tabular}
    \vspace*{-2mm}
    \bfcaption{Result of $Q_a$}
  \end{subtable}
  \vspace*{-3mm}
\bfcaption{$\semN$- and $\semB$-relation examples}
\label{table:2.1}
\trimfigurespacing
\end{figure}


\mypar{Query Semantics}
Operators of the positive relational algebra ($\raPlus$) over $\semK$-relations are defined by combining input annotations using operations $\addK$ and $\multK$. 
\begin{align*}
	\textbf{Union: } & [R_1 \union R_2](t) = R_1(t) \addK R_2(t) \\
	\textbf{Join: } & [R_1 \join R_2](t) = R_1(\pi_{R_1}[t]) \multK R_2(\pi_{R_2}[t]) \\
  \textbf{Projection: } & [\projection_U (R)](t) = \sum_{t=t'[U]}R(t')\\
  \textbf{Selection: } & [\selection_\theta(R)](t) = R(t) \multK \theta(t)
\end{align*}
\iftechreport{
For simplicity we assume in the definition above that tuples are of a compatible schema (e.g., $\relschema_1$ for a union $R_1 \union R_2$).
We use $\theta(t)$ to denote a function that returns $\oneK$ iff $\theta$ evaluates to true over tuple $t$ and $\zeroK$ otherwise.
}


\begin{Example}\label{ex:basic-k-rel}
  Figure~\ref{table:2.1} shows a bag semantics database encoded as an $\semN$-database, with each tuple $\tup$ annotated with its multiplicity (the copies of $\tup$ in the relation).
  Annotations appear beside each tuple. 
  Query $Q_a$, below, computes states. 
\vspace{-3mm}
  \begin{align*}
    Q_a = \projection_{state} \left(\textsf{\upshape Address} \join 
    \textsf{\upshape Neighborhood} \right)
  \end{align*}
  Every input tuple appears once (is annotated with $1$).
  The output tuple annotation is computed by multiplying annotations of joined tuples, and summing annotations projected onto the same result tuple.
  For instance, 2 NY addresses are returned. 
  %
\end{Example}

In the following,  
we will make use of homomorphisms. 
A mapping $h: \semK \to \semK'$ from a semiring $\semK$ to a semiring $\semK'$ is a called a \textbf{homomorphism} if it maps $\zeroOf{\semK}$ and $\oneOf{\semK}$ to their counterparts in $\semK'$ and distributes over sum and product (e.g., $h(k \addK k') = h(k) \addOf{\semK'} h(k')$).
As observed by Green et al.~\cite{Green:2007:PS:1265530.1265535}, any semiring homomorphism $h$ can be lifted to a homomorphism from $\semK$-relations to $\semK'$-relations by applying $h$ to the annotation of every tuple $\tup$: $h(R)(t) = h(R(t))$.
Importantly, queries commute with semiring homomorphisms. That is, given a homomorphism $h$, query $\query$, and $\semK$-database $\db$ we have $h(\query(\db)) = \query(h(\db))$.
We will abuse syntax and use the same function symbols (e.g., $h(\cdot)$) to denote mappings between semirings, $\semK$-relations, as well as $\semK$-databases.


\begin{Example}\label{ex:homo-k}
Continuing Example~\ref{ex:basic-k-rel}, 
we can derive a set instance  
through a mapping $h:\semN \rightarrow \semB$ defined as $h(k) = T$ if $k> 0$ and $h(k) = F$ otherwise.
$h$ is a semiring homomorphism, 
so evaluating $Q_a$ in $\semN$ first and then applying $h$ (i.e., $h(Q(\db))$) is equivalent to applying $h$ first, and then evaluating $Q_a$.
\end{Example}

When defining bounds for annotations in Section~\ref{sec:incompl-prob-k}, we make use of the so called \textit{natural order} $\ordK$ for a semiring $\semK$, defined as an element $k$ preceding $k'$ if it is possible to obtain $k'$ by adding to $k$. 
 Semirings for which the natural order is a partial order are called \textit{naturally ordered}~\cite{Geerts:2010bz}.
\begin{align}
\forall k, k' \in \kDom: k \ordK k' \Leftrightarrow \exists k'' \in \kDom: k \addK k'' = k'
\end{align}\\[-4mm]
\section{Incomplete K-relations}
\label{sec:incompl-prob-k}

Many incomplete data models either do not support bag semantics, or distinguish it from set semantics.
Our first contribution 
unifies both under a joint framework.
Recall that an incomplete database is a set of deterministic databases (possible worlds).
We now generalize this idea to $\semK$-databases.

\begin{Definition}[Incomplete $\semK$-database]
  Let $\semK$ be a semiring. An \emph{incomplete $\semK$-database} $\pdb$ is a set of possible worlds $\pdb = \{\db_1, \ldots, \db_n\}$ where each  $D_i$ is a $\semK$-database.
\end{Definition}

Like classical incomplete databases, queries over an incomplete $\semK$-database use
possible world semantics, i.e., the result of evaluating a query $\query$ over
an incomplete $\semK$-database $\pdb$ is the set of all possible worlds derived
by evaluating $\query$ over every possible world $\db \in \pdb$ (i.e., $\query(\pdb) = \{ \query(\db_1), \ldots, \query(\db_n) \}$).

\subsection{Certain Annotations} 
\label{sec:poss-worlds-cert}

While possible worlds semantics are directly compatible with incomplete
$\semK$-databases, the same does not hold for the concepts of certain\iftechreport{ and possible }
tuples, as we will show in the following. First off, we have to define
what precisely do we mean by certain answers over possible worlds that are
$\semK$-databases.

\begin{Example}\label{ex:incomplete-k-databases}
  Consider a $\semN$-database  $\pdb$ (bag semantics) containing a relation \textbf{LOC} with two attributes \texttt{locale} and \texttt{state}.
Assume that $\pdb$ consists of the two possible
  worlds below:
\vspace*{-1mm}
\begin{center}
  \begin{minipage}{0.45\linewidth}
    \centering
    \textbf{LOC} in $\db_1$\\
  \begin{tabular}{ c|c c }
    \textbf{locale} & \textbf{state} & \underline{$\semN$} \\
    \cline{1-2}
        Lasalle & NY & \textbf{3}\\
        Tucson & AZ & \textbf{2} \\
  \end{tabular}
  \end{minipage}
   \begin{minipage}{0.45\linewidth}
     \centering
    \textbf{LOC} in $\db_2$\\
       \begin{tabular}{ c|c c }
    \textbf{locale} & \textbf{state} & \underline{$\semN$} \\
         \cline{1-2}

        Lasalle & NY & \textbf{2}\\
         Tucson & AZ & \textbf{1} \\
         Greenville & IN & \textbf{5} \\
  \end{tabular}

  \end{minipage}
\end{center}
\vspace*{-1mm}

  Using semiring $\semN$ each tuple in
  a possible world is annotated with its multiplicity (the number of copies
  of the tuple that exist in the possible world). Arguably, tuples \texttt{(Lasalle, NY)} and \texttt{(Tucson, AZ)}
  are certain since they appear (multiplicity higher than $0$) in both possible
  worlds while \texttt{(Greenville, IN)} is not since it is not present (its
  multiplicity is zero) in possible world $\db_1$\footnote{
All tuples not shown in the tables are assumed to be annotated with zero.
  }.
  However, the boolean
  interpretation of certainty of incomplete databases is not suited to $\semN$-relations (or $\semK$-relations in general) because it
  ignores the annotations of tuples. In this particular example, tuple
  \texttt{(Lasalle, NY)} appears with multiplicity $3$ in possible world $\db_1$
  and multiplicity $2$ in possible world $\db_2$. We can state with certainty
  that in every possible world this tuple appears \emph{at least} twice. Thus,
  $2$ is a lower bound (the greatest lower bound) for the annotation of
  \texttt{(Lasalle, NY)}.
  Following this logic, we will define certainty through greatest lower bounds (GLBs) on tuple annotations.

\end{Example}

To further justify defining certain answers as lower bounds on annotations, consider classical incomplete databases which apply set semantics. Under set semantics, a tuple 
is \emph{certain} if it appears in all possible worlds\ifnottechreport{.}\iftechreport{ and \emph{possible} if it appears in at least one possible world. } 
Like the bag semantics example above,
certainty\iftechreport{ (possible) }
is a lower\iftechreport{ (upper) }
bound on a tuple's annotation across all 
worlds.
Consider the the order $false < true$.
If a tuple exists in every possible world (is always annotated \textit{true}), then intuitively, the GLB of its annotation across all worlds is \textit{true}.
Otherwise, the tuple is not certain (is annotated \textit{false} in at least one world), and the GLB is $false$. 

To define a sensible lower bound for annotations, we need an order relation for semiring elements.
We use the natural order $\ordK$ as introduced in Section~\ref{sec:data_provenance} to define the GLB\iftechreport{ and LUB }
of a set of $\semK$-elements. 
For a well-defined GLB, we require that $\ordK$ forms a lattice over $\kDom$, a property that makes $\semK$ an \textit{l-semiring}~\cite{DBLP:conf/icdt/KostylevB12}. 
\ifnottechreport{A lattice over a set $S$ according to a partial order $\leq_S$ over $S$ is a structure $(S, \lub, \glb)$ where $\glb$ (the greatest lower bound) is an operation over $S$ defined for all $a,b \in S$:}
\iftechreport{A lattice over a set $S$ and with a partial order $\leq_S$ is a structure $(S, \lub, \glb)$ where $\glb$ (the greatest lower bound) and $\lub$ (the lowest upper bound) are operations over $S$ defined for all $a,b \in S$ as:}

\ifnottechreport{
\vspace{-5mm}
  \begin{align*}
  a \glb b &= \max_{\leq_S}(\{ c \mid c \in S \wedge c \leq_{S} a \wedge c \leq_{S} b\})
\end{align*}
}
\iftechreport{
The least upper bound $\lub$ is defined symmetrically.
\begin{align*}
  a \lub b &= \min_{\leq_S}(\{ c \mid c \in S \wedge a \leq_{S} c \wedge b \leq_{S} c\})\\
  a \glb b &= \max_{\leq_S}(\{ c \mid c \in S \wedge c \leq_{S} a \wedge c \leq_{S} b\})
\end{align*}
In a lattice, $\lub$ and $\glb$ are  associative, commutative, and fulfill 
\begin{align*}
  a \lub (a \glb b) &= a &
  a \glb (a \lub b) &= a
\end{align*}
}

We will use $\glbK$ \iftechreport{and $\lubK$ }
to denote the $\glb$ \iftechreport{and $\lub$ }
operation of the lattice over 
$\ordK$ for a semiring $\semK$. 
Abusing notation, we will apply the $\GlbK$ \iftechreport{and $\LubK$ }
\ifnottechreport{operation}\iftechreport{operations } to sets of elements from $\semK$ with the understanding that \ifnottechreport{it}\iftechreport{they } will be applied iteratively to the elements in the set in some 
order, e.g., $\GlbK \{k_1, k_2, k_2\} = (k_1 \glbK k_2) \glbK k_3$.
\iftechreport{
This is well-defined for l-semirings, since in a lattice any set of elements has a unique greatest lower bound \iftechreport{and lowest upper bound } based on the associativity and commutativity laws of lattices. That is, no matter in which order  we apply $\glbK$ to the elements of a set, the result will be the same.
}
From here on, we will limit our discussion to l-semirings.
Many semirings, including the set semiring $\semB$ and the bag semiring $\semN$ are l-semirings. The natural order of $\semB$ is $F \ordOf{\semB} T$\iftechreport{, $k_1 \lubKof{\semB} k_2 = k_1 \vee k_2$, } 
and $k_1 \glbKof{\semB} k_2 = k_1 \wedge k_2$.
The natural order of $\semN$ is the standard order of natural numbers\iftechreport{, $k_1 \lubKof{\semN} k_2 = \max(k_1,k_2)$, }
and $k_1 \glbKof{\semN} k_2 = \min(k_1, k_2)$.


Based on $\glbK$\iftechreport{ and $\lubK$},
we define the certain \iftechreport{and possible }
annotation $\pwCertain(\pdb, \tup)$ of a tuple $\tup$ in an incomplete $\semK$-database $\pdb$ by gathering the annotations of tuple $\tup$ from all possible worlds of $\pdb$ and then applying  $\glbK$ to compute the greatest lower bound. 
\\[-5mm] 
\ifnottechreport{
\vspace{-5mm}
\begin{center}
  \begin{align*}
     \pwCertain(\pdb,\tup) &= \GlbK(\{\db(\tup) \mid \db \in \pdb \})
  \end{align*}\\[-5mm]

\end{center}
}
\iftechreport{
\begin{align*}
\pwCertain(\setK) &= \GlbK(\setK)   &\pwCertain(\pdb,\tup) &= \pwCertain(\pdb(\tup))\\
\pwPossible(\setK) &= \LubK(\setK) &\pwPossible(\pdb,\tup) &= \pwPossible(\pdb(\tup))
\end{align*}
}
%
Importantly, GLB coincides with the standard definition of certain answers for set semantics ($\semB$): $\pwCertOf{\semB}$ returns true only when the tuple is present in all worlds.
We also note that $\pwCertOf{\semN} = \min$, 
is analogous to 
the definition of certain answers for bag semantics from~\cite{GL16}. 
%
For instance, consider the certain annotation of the first tuple from Example~\ref{ex:incomplete-k-databases}.
The tuple's certain multiplicity is $\pwCertOf{\semN}(\{2,3\}) = min(2,3) = 2$. Similarly, for the third tuple, $\pwCertOf{\semN}(\{0,5\}) =0$. 
Reinterpreted under set semantics, all tuples that exist (multiplicity $>0$) are annotated $true$ ($T$) and all others $false$ ($F$). For the first tuple we get,
 $\glbKof{\semB}(\{T,T\}) = T \wedge T = T$ (certain).
For the third tuple we get $\glbKof{\semB}(\{F,T\}) = F \wedge T = F$ (not certain).

\subsection{$\pwK$-relations}
\label{sec:pwk-relations}

For the formal exposition in the remainder of this work it will be useful to define an alternative, but equivalent, encoding of an incomplete $\semK$-database as a single $\semK$-database using a special class of semirings whose elements encode the annotation of a tuple across a set of possible worlds.
This encoding is a technical device 
that allows us to adopt results from the theory of $\semK$-relations directly to our problem.
We assume a fixed set $\pwDom = \{ m \mid m \in \mathbb N \wedge 0 < m \leq n \}$ of possible world identifiers for some number of possible worlds $n \in \mathbb N$.
Given the domain $\kDom$ of a semiring $\semK$, we write $\pwkDom$ to denote the set of elements from the $\pwCard$-way cross-product of $\kDom$.
We annotate tuples $\tup$ with elements of $\pwkDom$ to store annotations of $\tup$ in each possible world.
We use  $\pwe{k}$, $\pwe{k_1}$, \ldots{} to denote elements from $\pwkDom$ to make explicit that they are vectors.

\begin{Definition}[Possible World Semiring]
  Let $\semK = (K,$ $\addK, \multK, \zeroK, \oneK)$ be an l-semiring. We define the possible world semiring $\pwK = (\pwkDom, \pwAdd, \pwMult, \pwZero, \pwOne)$. The operations of this semiring are defined as follows
\iftechreport{
  \begin{align*}
&\forall i \in \pwDom:    &\pwZero[i] &= \zeroK \\
&\forall i \in \pwDom:    &\pwOne[i] &= \oneK \\
&\forall i \in \pwDom:    &(\pwe{k_1} \pwAdd \pwe{k_2})[i] &= \pwe{k_1}[i] \addK \pwe{k_2}[i]\\
&\forall i \in \pwDom:    &(\pwe{k_1} \pwMult \pwe{k_2})[i] &= \pwe{k_1}[i] \multK \pwe{k_2}[i]
  \end{align*}
}
\ifnottechreport{
$\forall i \in \pwDom$:\\[-7mm]

\begin{align*}
\pwZero[i] = \zeroK
  && (\pwe{k_1} \pwAdd \pwe{k_2})[i] = \pwe{k_1}[i] \addK \pwe{k_2}[i]\\
\pwOne[i] = \oneK
  && (\pwe{k_1} \pwMult \pwe{k_2})[i] = \pwe{k_1}[i] \multK \pwe{k_2}[i]
\end{align*}
}
\end{Definition}

Thus, a $\pwK$-database is simply a pivoted representation of an incomplete $\semK$-database.

\begin{Example}\label{example:kwlabels}
  Reconsider the incomplete $\semN$-relation from Example~\ref{ex:incomplete-k-databases}. The encoding of this database as a $\semN^2$-relation is: 
  \begin{center}
  {\upshape \small
  \begin{tabular}{ c|c c }
        \textbf{locale} & \textbf{state} & \underline{$\semN^2$}  \\
    \cline{1-2}&&\\[-3.1mm]
          Lasalle & NY & \textbf{[3,2]}\\
         Tucson & AZ & \textbf{[2,1]} \\
         Greenville & IN & \textbf{[0,5]} \\
  \end{tabular}
  }
\end{center}
%
\end{Example}

\newcommand{\includeIsomorphismProp}{
\begin{Proposition}
  Incomplete $\semK$-databases and $\pwK$-databases are isomorphic wrt. possible worlds semantics for $\raPlus$ queries.
\end{Proposition}
}

\iftechreport{
Translating between incomplete $\semK$-databases and $\pwK$-databases is trivial. Given an incomplete $\semK$-database with $n$ possible worlds $\{ \db_i \}$,  we create the corresponding $\pwK$-database by annotating each tuple $t$ with the vector $[ \db_1(t), \ldots,$ $\db_n(t)]$. In the other direction, given a $\pwK$-database $\pdb$ with vectors of length $n$, we construct the corresponding incomplete $\semK$-database by annotating each tuple $t$ with $\pdb(t)[i]$ in possible world $\db_i$.
In addition, we will show below that queries over $\pwK$-databases encode possible world semantics.
Thus, the following result holds and we can use incomplete $\semK$- and $\pwK$-databases interchangeably.

\includeIsomorphismProp
}

Observe that $\pwK$ is a semiring, since we define $\pwK$ using
the $\card{W}$-way version of the \textit{product} operation of universal algebra, and 
 products of semirings are also semirings~\cite{burris2012algebra}.

\mypar{Possible Worlds}
We can extract the $\semK$-database for a possible world (e.g., the \textit{\termBGW}) from a $\pwK$-database by projecting on one dimension of its annotations. This can be modeled as a mapping $\pwWorld{i}: \pwkDom \to \kDom$ where $i \in \pwDom$:
\begin{align}
\pwWorld{i}(\pwe{k}) = \pwe{k}[i]
\end{align}
Recall that under possible world semantics, the result of a query $\query$ is the set of worlds computed by evaluating $\query$ over each world of the input.
As a sanity check, we would like to ensure that query processing over $\pwK$-relations matches this definition.
We can state possible world semantics equivalently as follows:
the content of a possible world in the query result ($\pwWorld{i}(\query(\pdb))$) is the result of evaluating query $\query$ over this possible world in the input ($\query(\pwWorld{i}(\pdb))$):
\iftechreport{That is, 
$\pwK$-relations have possible worlds semantics iff $\pwWorld{i}$ commutes with queries:
}
  \begin{align*}
&\forall i \in \pwDom: \pwWorld{i}(\query(\pdb))= \query(\pwWorld{i}(\pdb))
  \end{align*}


Recall from Section~\ref{sec:data_provenance} that a mapping between semirings commutes with queries iff it is a semiring homomorphism.
\iftechreport{Note that $\pwK$-relations admit a trivial extension to probabilistic data by defining a distribution $\prob: W \mapsto [0,1]$. See~\cite{FH18} for details.}

\begin{Lemma}\label{lem:poss-world-is-homomorphism}
  For any semiring $\semK$ and possible world $i \in \pwDom$, mapping $\pwWorld{i}$ is a semiring homomorphism.
\end{Lemma}
\ProofPointerToAppendix

\ifnottechreport{
\revm{A useful consequence of \Cref{lem:poss-world-is-homomorphism} is that $\semK$-databases and $\pwK$-databases are equivalent and interchangeable.}

\includeIsomorphismProp
}

\iftechreport{\mypar{Probabilistic Data}
$\pwK$-relations admit a trivial extension to probabilistic data by defining a distribution $\prob: W \mapsto [0,1]$ such that $\sum_{i \in W} \prob(i) = 1$.
In contrast to classical frameworks for possible worlds, where the collection of worlds is a \emph{set},  $\pwK$ queries preserve the same $\card{W}$ possible worlds\footnote{Although it has no impact on our results, it is worth noting that the worlds in a $\pwK$  query result may not be distinct.}.
Hence, the input distribution $\prob$ applies, unchanged, to the $\card{W}$ possible query outputs.
}

\iftechreport{\mypar{Certain and Possible Annotations} }
\ifnottechreport{\mypar{Certain Annotations} }
Since the annotation of a tuple $\tup$ in a $\pwK$-database is a vector recording $\tup$'s annotations in all worlds, certain annotations for incomplete $\semK$-databases are computed by applying $\GlbK$ to the set of annotations contained in the vector.
Thus, the certain annotation of a tuple $\tup$ from a $\pwK$-DB $\pdb$ is computed as:
\ifnottechreport{
$\pwCertain(\pdb,\tup) = \GlbK \left(\pdb(\tup)\right)$
}
\iftechreport{
\begin{align*}
\pwCertain(\setK) &= \GlbK(\setK) &   \pwCertain(\pdb,\tup) &= \pwCertain(\pdb(\tup))\\
\pwPossible(\setK) &= \LubK(\setK) &   \pwPossible(\pdb,\tup) &= \pwPossible(\pdb(\tup))
\end{align*}
}
\section{
  Labeling Schemes 
}\label{sec:create-labelings}

We define efficient (\ptime) labeling schemes for \ifnottechreport{three  incomplete data models \revm{and their probabilistic extensions}: Tuple\hyp{}Independent \revm{(probabilistic)} databases~\cite{suciu2011probabilistic}, \revm{(P)}C-Tables~\cite{DBLP:journals/jacm/ImielinskiL84,DBLP:journals/debu/GreenT06}, and \abbrXDBs~\cite{DBLP:conf/vldb/AgrawalBSHNSW06}. \reva{For further details and examples see~\cite{FH18}.} 
}
\iftechreport{three existing incomplete data models: Tuple-Independent databases~\cite{suciu2011probabilistic}, the disjoint-independent x-relation model from~\cite{DBLP:conf/vldb/AgrawalBSHNSW06}, and C-Tables~\cite{DBLP:journals/jacm/ImielinskiL84}.
We also show how to extract a \termBGW from an $\pwK$-database derived from these models.
Since computing certain answers is hard in general, our \ptime labeling schemes
cannot be c-correct for all models.
}

\iftechreport{
\subsection{Labeling Schemes}
\label{sec:labeling-schemes}
}

\mypar{Tuple-Independent Databases}
\reva{A tuple\hyp{}independent  database (\abbrTI) $\pdb$ is a database where each tuple $\tup$ is marked as optional or not. 
The incomplete database represented by a \abbrTI $\pdb$ is the set of instances that include all non-optional tuples and some subset of the optional tuples. That is, the existence of a tuple $t$ is independent of the existence of any other tuple $t'$. 
In the probabilistic version of \abbrTIs each tuple is associated with its marginal probability. 
The probability of a possible world is then the product of the probability of all tuples included in the world multiplied by the product of $1 - \prob(t)$ for all tuples from $\pdb$ that are not part of the possible world.
}
We define a labeling function $\doTULTI$ for \abbrTIs that returns a $\semB$-labeling $\TUL$ that annotates a tuple with $T$ (certain) iff it is not optional.
\reva{For probabilistic \abbrTIs{} we label tuples as certain if their marginal probability is $1$.}
%
\begin{align*}\label{eq:tip-labeling}
\TUL(\tup) = T &\Leftrightarrow  t\,\text{is not marked as optional} 
\end{align*}
\ifnottechreport{\reva{To create a \abbrBGW from an incomplete \abbrTI{} $\pdb$, we can choose any subset of $\pdb$ that is a superset of the tuples we labeled as certain. For probabilistic \abbrTIs, we choose the world with the highest probability, which can be computed efficiently as follows:\hfill $\bgdb = \{ t \mid t \in \pdb \wedge \prob(t) \geq 0.5 \}$
}}


\begin{Theorem}[$\doTULTI$ is c-correct]\label{theo:h-homo-UA}
Given a \abbrTI  $\pdb$, $\doTULTI(\pdb)$ is a c-correct labeling.
\end{Theorem}
\begin{proof}
 Trivially holds. An incomplete (probabilistic) database tuple is certain iff it is not optional (if $\prob(\tup) = 1$).
\end{proof}

\ifnottechreport{
 \mypar{\revm{(P)}\abbrCtables}
Tuples in a \textit{\abbrCtable}~\cite{DBLP:journals/jacm/ImielinskiL84} consist of constants and variables from a set $\Sigma$.
Tuples are annotated by a boolean expression $\phi_\pdb(\tup)$ over comparisons of values from $\Sigma \cup \aDom$, called the local condition.
Each variable assignment $v : \Sigma \rightarrow \aDom$  satisfying a boolean expression called the global condition
defines a possible world derived by retaining only tuples with local conditions satisfied under $v$.
Determining certainty of tuples based on a \abbrCtable is expensive.
Thus, we cannot hope for an efficient c-correct labeling scheme for \abbrCtables.
Instead, consider the following sufficient  condition for certainty.
If (1) a tuple $\tup$ in a \abbrCtable contains only constants and (2) its local condition $\phi_\pdb(\tup)$ is a tautology, then the tuple is certain.
Our labeling scheme for \abbrCtables applies this sufficient condition and, thus, is c-sound.
Formally, $\TUL = \doTULC(\pdb)$, where for a C-table $\pdb$ and any tuple $\tup \in \tupDom$:
\begin{align*}
 \TUL(\tup) = T \Leftrightarrow \phi_\pdb(\tup) \mathtext{is in CNF} \wedge (\models \phi_\pdb(\tup))
\end{align*}
\reva{Green et. al.~\cite{DBLP:journals/debu/GreenT06} extended \abbrCtables by assigning each variable to a probability distribution over its possible values.
In the resulting ``\abbrPCtables,'' variables are treated as independent, i.e., the probability of a possible world is the product of probabilities of its defining variable assignments.
Our labeling scheme works for both \abbrCtables{} and \abbrPCtables{}.
}
\reva{To compute a \abbrBGW for a \abbrCtable, we can randomly choose an assignment for each variable in $\Sigma$. 
For a \abbrPCtable, 
  computing the most likely possible world reduces to answering a 
  query over a TI-DB (\texttt{\#P} in general~\cite{suciu2011probabilistic}).
  Thus, we heuristically select the highest probability value for each variable independently.
  Any such assignment must contain all certain tuples, since their local conditions are tautologies.
  Alternatively, a wide range of algorithms~\cite{DBLP:journals/vldb/GatterbauerS17,DBLP:conf/sigmod/FinkHOR11,DBLP:journals/vldb/FinkHO13,DBLP:journals/vldb/LiSD11} can compute an arbitrarily close approximation of the most likely  world.
}

\begin{Theorem}[$\doTULC$ is c-sound]\label{theo:h-homo-UA-c-table}
Given an incomplete/probabilistic database  $\pdb$ encoded as a \abbrCtable or \abbrPCtable, $\doTULC(\pdb)$ is c-sound.
\end{Theorem}
}

\iftechreport{
 \mypar{\abbrCtables}
\textit{C-Tables}~\cite{DBLP:journals/jacm/ImielinskiL84} use a set $\Sigma$ of variable symbols to define possible worlds.
Tuples are annotated by a boolean expression over comparisons of values from $\Sigma \cup \aDom$, called the local condition.
Each variable assignment $v : \Sigma \rightarrow \aDom$ satisfying a boolean expression called the global condition defines a possible world, derived by retaining only tuples with local conditions satisfied under $v$.
Computing certain answers for first order queries is \conpcomplete~\cite{V86,AK91} even for \abbrCoddTables.
Since the result of any first order query over a \abbrCoddTable can be represented as a \abbrCtable and evaluating a query in this fashion is efficient, it follows that determining whether a tuple is certain in a \abbrCtable cannot be in \ptime.
Instead, consider the following sufficient, but not necessary  condition for a tuple $\tup$ to be certain.
If (1) a tuple $\tup$ in a \abbrCtable contains only constants and (2) its local condition $\phi_\pdb(\tup)$ is a tautology, then the tuple is certain.
To see why this is the case, recall that under the closed-world assumption, a \abbrCtable represents a set of possible worlds, one for each valuation of the variables appearing in the \abbrCtable (to constants from $\aDom$).
A tuple is part of a possible world corresponding to such a valuation if the tuple's local condition is satisfied under the valuation.
Thus, a tuple consisting of constants only, with a local condition that is a tautology is part of every possible world represented by the \abbrCtable.
If the local condition of a tuple is in conjunctive normal form (CNF) then checking whether it is a tautology is efficient (\texttt{PTIME}).
Our labeling scheme for \abbrCtables applies this sufficient condition and, thus, is c-sound.
Formally, $\TUL = \doTULC(\pdb)$, where for a C-table $\pdb$ and any tuple $\tup \in \tupDom$:
\begin{align*}
 \TUL(\tup) = T \Leftrightarrow \phi_\pdb(\tup) \mathtext{is in CNF} \wedge (\models \phi_\pdb(\tup))
\end{align*}

\reva{Green et. al.~\cite{DBLP:journals/debu/GreenT06} introduced \abbrPCtables a probabilistic version of \abbrCtables where each variable is associated with a probability distribution over its possible values. Variables are considered independent of each other, i.e., the probability of a possible world is computed as the product of the probabilities of the individual variable assignments based on which the world was created. Our labeling scheme works for both the incomplete and probabilistic version of \abbrCtables{}.
}

\begin{Theorem}[$\doTULC$ is c-sound]\label{theo:h-homo-UA-c-table}
Given an incomplete database  $\pdb$ encoded as \abbrCtables, $\doTULC(\pdb)$ is c-sound.
\end{Theorem}

Note that $\TUL$ is not guaranteed to be c-correct. For instance, a tuple $\tup$ consisting only of constants and for which $\phi_{\pdb}(\tup)$ is a tautology is guaranteed to be certain, but $\TUL(\tup) = F$ if  $\phi_{\pdb}(\tup)$ is not in CNF.

\begin{Example}
  Consider a \abbrCtable consisting of two tuples $t_1 = (1,X)$ with $\phi_\pdb(t_1) \defas (X = 1)$ and $t_2 = (1,1)$ with $\phi_{\pdb}(t_2) \defas (X \neq 1)$. $\doTULC$ would mark $(1,1)$ as uncertain, because even though this tuple exists in the \abbrCtable and it's local condition is in CNF, the local condition is not a tautology. However, tuple $(1,1)$ is certain since either $X=1$ and then first tuple evaluates to $(1,1)$ or $X \neq 1$ and the second tuple is included in the possible world.

\end{Example}
}

\mypar{\abbrXDBs}
\reva{An \abbrXDB~\cite{DBLP:conf/vldb/AgrawalBSHNSW06} is a set of x-relations, which are sets of x-tuples. An
x-tuple $\xTup$ is a set of tuples $\{t_1, \ldots, t_n\}$ with a label indicating whether the x-tuple is optional. 
Each x-tuple is assumed to be independent of the others, and its alternatives are assumed to be disjoint.
Thus, a possible world of an x-relation $\rel$ is constructed by selecting at most one alternative $t \in \xTup$ for every x-tuple $\xTup$ from $\rel$ if $\xTup$ is optional, or exactly one if it is not optional.
The probabilistic version of \abbrXDBs (also called a Block-Independent or BI-DB) as introduced in~\cite{DBLP:conf/vldb/AgrawalBSHNSW06} assigns each alternative a probability and we require that $\prob(\xTup) = \sum_{t \in \xTup} \prob(t) \leq 1$. Thus, a tuple is optional if $\prob(\xTup) < 1$ and there is no need to use labels to mark optional tuples.}
%
We use $\card{\xTup}$ to denote the number of alternatives of x-tuple $\xTup$.
We define a labeling scheme $\doTULX$ for x-relations that annotates a tuple $t$ from an \abbrXDB $\pdb$ with $T$ if $t$ is the single, non-optional alternative of an x-tuple, and $F$ otherwise. In probabilistic \abbrXDBs we check $\prob(\xTup) = 1$.
$$
\TUL(\tup) = T \Leftrightarrow \exists \xTup \in \pdb: \card{\xTup} = 1 \wedge \xTup\,\text{is not optional} 
$$


\begin{Theorem}[$\doTULX$ is c-correct]\label{theo:h-homo-UA-TUL-correct}
Given a  database $\pdb$, $\doTULX(\pdb)$ is a c-correct labeling.
\end{Theorem}

\ifnottechreport{
\reva{For probabilistic \abbrXDBs, the possible world with highest probability can be efficiently computed and used as the \abbrBGW.
  Since the x-tuples in an \abbrXDB are independent, the probability of a possible world from an \abbrXDB $\pdb$ is maximized by
  including the highest probability alternative for each x-tuple $\xTup$ (i.e., $\argmax_{t \in \xTup} \prob(t)$), or no tuple if that is the highest probability option (i.e., $\max_{t \in \xTup} \prob(t) <  (1 - \prob(\xTup))$).
  For incomplete \abbrXDBs, we choose a random alternative for each x-tuple.}
}

\iftechreport{
  \subsection{Extracting \termBGWs{}}
  \label{sec:deriv-poss-worlds}

  Computing some possible world is trivial for most incomplete and probabilistic data models.
  However, for the case of probabilistic data models we are particularly interested in the highest-probability  world (the best guess world). \reva{We now discuss in more detail how we choose the \abbrBGW $\bgdb$ for the data models for which we have introduced labeling schemes above.}

  \mypar{\abbrTI}
  For a \abbrTI $\pdb$, the best guess world consists of all tuples $\tup$ such that $\prob(\tup) \geq 0.5$. 
  To understand why this is the case recall that the probability of a world from a  \abbrTI is the product of the probabilities of included tuples with one minus the probability of excluded tuples.
  This probability is maximized by including only tuples where $\prob(\tup) \geq 0.5$. \reva{For the incomplete version of \abbrTIs we have to include all non-optional tuples and can choose arbitrarily which optional tuples to include in $\bgdb$.}

  \mypar{\abbrPCtables}
  For a \abbrPCtable, 
  computing the most likely possible world reduces to answering a 
  query over the database, which is known to be \texttt{\#P} in general~\cite{suciu2011probabilistic}.
  Specific tables (e.g., those generated by ``safe'' queries~\cite{suciu2011probabilistic}) admit \ptime solutions.
  Alternatively, there exist a wide range of algorithms~\cite{DBLP:journals/vldb/GatterbauerS17,DBLP:conf/sigmod/FinkHOR11,DBLP:journals/vldb/FinkHO13,DBLP:journals/vldb/LiSD11} that can be used to compute an arbitrarily close approximation of the most likely  world.

  \mypar{Disjoint-independent databases}
  Since the x-tuples in an \abbrXDB are independent of each other, the probability of a possible world from an \abbrXDB $\pdb$ is maximized by
  including for every x-tuple $\xTup$ its alternative with the highest probability $\argmax_{t \in \xTup} \prob(t)$ or no alternative if $\max_{t \in \xTup} \prob(t) <  (1 - \prob(\xTup))$, i.e., if the probability of not including any alternative for the x-tuple is higher than the highest probability of an alternative for the x-tuple.

}





\section{UA-Databases} \label{sec:UA-model}

We now introduce \textit{\abbrUADBs} (uncertainty-annotated databases) which encode both under\hyp{} and over\hyp{}approximations of the certain annotations of an incomplete $\semK$-database $\pdb$.
%
This is achieved by annotating every tuple with a pair $\uae{d}{c} \in \doubleDom{\kDom}$ where $d$ records the tuple's annotation in the \abbrBGW ($\db(t)$, for some $\db \in \pdb$) and $c$ stores the under-approximation of the tuple's certain annotation (i.e., $c \ordK \pwCertain(\pdb,t) \ordK d$). 
Both under\hyp{} and over\hyp{}approximations of certain annotations assign tuples annotations from $\semK$, making them $\semK$-databases.
This will be important for proving that these bounds are preserved under queries.
%
Every possible world is by definition a superset of the certain tuples, so a \abbrUADB contains all certain answers, even though the certainty of some answers may be underestimated.
We start by formally defining the annotation domains of \abbrUADBs and mappings that extract the two components of an annotation. Afterwards, we state the main result of this section: queries over \abbrUADBs preserve the under\hyp{} and over\hyp{}approximation of certain annotations.


\subsection{UA-semirings}\label{sec:ua-semirings}
We define a UA-semiring as a $\semK^2$-semiring, i.e., the direct product of a semiring $\semK$ with itself (see Section~\ref{sec:ua-semirings}).
\BG{Afterwards, we prove that the result of a query over a \abbrUADB encodes the result of the query over the input possible world and uncertainty labeling.}
\iftechreport{In the following we will write $k k'$ instead of $k \multK k'$ if the semiring $\semK$ is clear from the context.  } Recall that operations in $\aDoubleK = (\doubleDom{\kDom}, \addOf{\aDoubleK}, \multOf{\aDoubleK}, \zeroOf{\aDoubleK}, \oneOf{\aDoubleK})$ are defined pointwise, e.g., $[k_1, {k_1}'] \multOf{\aDoubleK} [k_2,{k_2}'] = [k_1 \multK k_2, {k_1}' \multK {k_2}']$.

\begin{Definition}[UA-semiring]
  Let $\semK$ be a semiring. We define the corresponding UA-semiring
  $\uaK{\semK}=\semK^2$ 
\end{Definition}

Note that for any $\semK$, $\uaK{\semK}$ is a 
semiring, because, as mentioned earlier, products of semirings are semirings.

\subsection{Creating UA-DBs}
\label{sec:creat-ua-relat}

We now discuss how to derive UA-relations from a $\pwK$-database or a compact encoding of a $\pwK$-database using some uncertain data model like c-tables. 
Consider a $\pwK$-database $\pdb$, let $\db$ be one of its  worlds and  $\TUL$ a $\semK$-database   under\hyp{}approximating the certain annotations of $\pdb$. We refer to $\TUL$ as a \textit{labeling} and will study such labelings in depth in Section~\ref{sec:uncert-labell-1} and~\ref{sec:c-soundess}.
We cover in Section~\ref{sec:create-labelings} how to generate a \abbrUADB from common uncertain data models by extracting a (best-guess) world $\db$ and a labeling $\TUL$. 
We construct a \abbrUADB $\uadb$  as an \textit{encoding} of  $\db$ and $\TUL$  by setting for every tuple $\tup$:
$$
\uadb(\tup) = \uae{\db(\tup)}{\TUL(\tup)}
$$
For a \abbrUADB $\uadb$ constructed in this fashion we say that $\uadb$ \textit{approximates} $\pdb$ by encoding $(\TUL, \db)$.
%
Given a \abbrUADB $\uadb$,  
we would like to be able to restore $\TUL$ and $\db$ from $\uadb$.
For that we define two morphisms $\doubleK{\semK} \to \semK$:
{
$$ 
  \cOf{\uae{d}{c}} = c
    \hspace{5mm}
  \dOf{\uae{d}{c}} = d$$
}
%
Note that by construction, if an \abbrUADB $\uadb$  is an \textit{encoding} of a possible world $\db$ and a labeling $\TUL$ of a $\pwK$-database $\pdb$  then:
%
$  \dOf{\uadb} = \db
    \hspace{2mm} \textbf{and} \hspace{2mm}
\cOf{\uadb} = \TUL$.


\subsection{Querying \abbrUADBs{}}
\label{sec:creat-ua-query}


We now state the main result of this section: query evaluation over \abbrUADBs preserves the under\hyp{}approximation and over\hyp{}approximation
of certain annotations.
\iftechreport{
To prove the main result, we first show that $h_{\cOfName}$ and $h_{\dOfName}$ are homomorphisms, because this implies that queries over \abbrUADBs are evaluated over the $c$ and the $d$ component of an annotation independently. Thus, we can prove the result for under\hyp{} and over\hyp{}approximations separately.
For over\hyp{}approximation we can trivially show an even better result: By definition (Section~\ref{sec:pwk-relations}) the possible world used as an over-approximation is preserved exactly.
Hence, the over-approximation property is preserved and \abbrUADBs are also backwards compatible with \abbrBGQP.
For under\hyp{}approximations we have to show that query evaluation preserves
under\hyp{}approximations. This part is more involved and we will prove this result in Section~\ref{sec:c-soundess}.
}


\begin{Theorem}[Queries Preserve Bounds]\label{theo:UA-preserve-approx}
  Let $\pdb$ be a $\pwK$-database, $\TUL$ a labeling for $\pwK$, $\db$ one of its possible worlds, and $\uadb$ be the \abbrUADB encoding the pair $(\TUL, \db)$. Clearly $\uadb$ approximates $\pdb$. Then  $\query(\uadb)$ is an approximation for $\query(\pdb)$  encoding the pair $(\query(\db), \query(\TUL))$. 
\end{Theorem}
\iftechreport{\ProofPointerToAppendix}
\ifnottechreport{\ProofPointerToAppendixWithSketch{
The proof decouples $c$ and $d$ by showing that $h_{\cOfName}$ and $h_{\dOfName}$ are homomorphisms.
Computing $d$ is exactly \abbrBGQP.
Then, we show in \Cref{sec:c-soundess} that queries over $h_{\cOfName}$ preserve the lower bound.
}}

\section{Uncertainty Labelings}
\label{sec:uncert-labell-1}

We now 
define uncertainty labelings, which are $\semK$-databases whose annotations over\hyp{} or under\hyp{}approximate certain annotations of tuples in a $\pwK$-database with respect to the natural order of semiring $\semK$.
A labeling scheme is a mapping  
from an incomplete databases to labelings.   

\begin{Definition}[Uncertainty Labeling Scheme]
Let $\dbDomK{\semK}$ be the set of all $\semK$-databases, $\mathcal{M}$ an incomplete/probabilistic data model, and $\mathcal{DB}_{\mathcal{M}}$ the set of all possible instances of this model.
  An uncertainty labeling scheme is a function $\doTUL: \mathcal{DB}_{\mathcal{M}} \to \dbDomK{\semK}$ such that the labeling $\TUL = \doTUL(\pdb)$  has the schema $\schema$.
\end{Definition}

Ideally, we would like the label (annotation) $\TUL(\tup)$ of a tuple $\tup$ from an uncertainty labeling $\TUL$ to be exactly $\pwCertain(\pdb, \tup)$.\revb{
Observe that an exact labeling can always be computed in $O(|W|)$ time if all worlds of the incomplete database can be enumerated.
However, the number of possible worlds is frequently exponential in the data size.
Thus, most incomplete data models rely on factorized encodings, with size typically logarithmic in $|W|$.
Ideally, we would like labeling schemes to be \ptime in the size of the encoding (rather than in $|W|$).
As mentioned in the introduction, computing certain answers is \conpcomplete, so for tractable query semantics we must accept that
$\TUL(\tup)$ may either over\hyp{}  or under\hyp{}approximate  $\pwCertain(\pdb, \tup)$ (with respect to $\ordK$).
}
For instance, under bag semantics (semiring $\semN$), a label $n$ may be smaller or larger than the certain multiplicity of a tuple.
We call a labeling \textit{c-sound} (no false positives) if it consistently under\hyp{}approximates the certain annotation of tuples, \textit{c-complete} (no false negatives) if it consistently over\hyp{}approximates certainty, and \textit{c-correct} if it annotates every tuple with its certain annotation. We also apply this terminology to labeling schemes, e.g., a c-sound labeling scheme only produces c-sound labelings. 
For \abbrUADBs{} we are mainly interested in c-sound labeling schemes to provide an under\hyp{}approximation of certain annotations.

\begin{Definition} 
  If $\TUL$ is an uncertainty labeling for 
  $\pdb$.
  \begin{center}
  \begin{tabular}{ r@{\hspace{10mm}}l }
  \textbf{We call $\TUL$\ldots} & \textbf{\ldots iff for all tuples $t \in \pdb$\ldots} \\[1mm]
  \textbf{c-sound}    & $\TUL(\tup) \ordK \pwCertain(\pdb,\tup)$\\
  \textbf{c-complete} & $\pwCertain(\pdb,\tup) \ordK \TUL(\tup)$\\
  \textbf{c-correct}  & $\pwCertain(\pdb,\tup) = \TUL(\tup)$\\
  \end{tabular}
  \end{center}
\end{Definition}
A labeling is both c-sound and c-complete iff it is c-correct. 
Ideally, queries over labelings would preserve these bounds.


\begin{Definition}[Preservation of Bounds]
A query semantics for uncertainty labelings preserves a property $X$ (c-soundness, c-completeness, or c-\-cor\-rect\-ness) wrt. a class of queries $\qClass$, if for any incomplete database $\pdb$, labeling  $\TUL$ for $\pdb$ that has property $X$, and query $\query \in \qClass$ we have: $\query(\TUL)$ is an uncertainty labeling for $\query(\pdb)$ with property $X$.
\end{Definition}

\BG{NOT MEANINGFUL HERE ANYMORE: As mentioned in the introduction (and observed elsewhere~\cite{L16a,GL16}), false negatives are less of a concern than false positives. Thus, for cases where we cannot preserve c-correctness, we would at least opt for c-soundness.
That is, we prefer labelings that underestimate the ``real'' $\pwCertain(\pdb, \tup)$ annotation wrt. the natural order $\ordK$.}


\section{Querying Labelings}\label{sec:c-soundess}

We now study whether queries over labelings produced by \textit{labeling schemes} such as the ones described in \Cref{sec:create-labelings} preserve c-soundness. 
Specifically, we demonstrate that standard $\semK$-relational query evaluation preserves \emph{c-soundness} for any c-sound labeling scheme.
Recall that a query semantics for labelings preserves c-soundness if a query $\query(\TUL)$ evaluated on a c-sound labeling $\TUL$ of \revb{incomplete} database $\pdb$ is a c-sound labeling for $\query(\pdb)$.
Our result generalizes a previous result of Reiter~\cite{R86} to any type of incomplete $\semK$-database for which we can define an efficient c-sound labeling scheme.
We need the following lemma, to show that the natural order of a semiring factors through addition and multiplication.
This is a known result that we only state for completeness.

\begin{Lemma}\label{lem:order-factors-through-operations}
  Let $\semK$ be a naturally ordered semiring. For all $k_1, k_2, k_3, k_4 \in \semK$ we have:
  \begin{align*}
    k_1 \ordK k_3 \wedge k_2 \ordK k_4 \Rightarrow k_1 \addK k_2 \ordK k_3 \addK k_4\\
    k_1 \ordK k_3 \wedge k_2 \ordK k_4 \Rightarrow k_1 \multK k_2 \ordK k_3 \multK k_4
  \end{align*}
\end{Lemma}
\ProofPointerToAppendix

\subsection{Preservation of C-Soundness}\label{sec:pres-c-soundn}

We now prove that $\raPlus$ over labelings preserves c\hyp{}soundness.
Since queries over both $\pwK$-databases and labelings have $\semK$-relational query semantics,
we can make use of the fact that $\raPlus$ over $\semK$-relations is defined using $\addK$ and $\multK$.
At a high level, the argument is as follows:
(1) we show that $\pwCertain$ applied to the result of an addition (or multiplication) of two $\pwK$-elements $\pwe{k_1}$ and $\pwe{k_2}$ yields a larger (wrt. $\ordK$) result than adding (or multiplying) the result of applying $\pwCertain$ to $\pwe{k_1}$ and $\pwe{k_2}$;
(2) Since c-sound labelings for an input provide a lower bound on $\pwCertain$, we can apply Lemma~\ref{lem:order-factors-through-operations} to show that the query result over a c-sound (or c-correct) labeling  is a lower bound for $\pwCertain$ of the result of the query. Combining arguments, we get preservation of c-soundness.

Functions that have the property mentioned in (1) are called superadditive and supermultiplicative. Formally, a function $f: A \to B$ where $A$ and $B$ are closed under addition and multiplication, and $B$ is ordered (order $\leq_{B}$) is superadditive (supermultiplicative) iff for all $a_1, a_2 \in A$:
\begin{align*}
   f(a_1 + a_2)  &\geq_{B} f(a_1) + f(a_2) \tag{superadditive}\\
f(a_1 \times a_2) &\geq_{B} f(a_1) \times f(a_2) \tag{supermultiplicative}
\end{align*}
In a nutshell, if we are given a c-sound $\semK$-labeling, then evaluating any $\raPlus$-query over the labeling using $\semK$-relational query semantics preserves c-soundness if we can prove that $\pwCertain$ is superadditive and supermultiplicative.

\begin{Lemma}\label{lem:super-all}
Let $\semK$ be a semiring. $\pwCertain$ is superadditive and supermultiplicative wrt. the natural order $\ordK$.
\end{Lemma}
\ProofPointerToAppendix

Using the superadditivity and -multiplicativity of $\pwCertain$, we now prove preservation of c-soundness. 
\ifbool{ShowDetailedProofs}{
We first prove a restricted version of this result.
\begin{Lemma}\label{lem:k-rel-preserves-soundness-over-c-correct}
Let $\pdb$ be a $\pwK$-database and $\TUL$ be a c-correct $\semK$-labeling for $\pdb$. $\raPlus$ queries over $\TUL$ preserve c-soundness.
\end{Lemma}

The major drawback of Lemma~\ref{lem:k-rel-preserves-soundness-over-c-correct} is that it is limited to c-correct input labelings.
Next, we show that c-soundness is still preserved even if the input labeling is only c-sound.

}{} 

\begin{Theorem}\label{theo:UA-are-c-sound} 
Let $\pdb$ be a $\pwK$-database and $\TUL$ a c-sound labeling for $\pdb$. $\raPlus$ queries over $\TUL$ preserve c-soundness.   
\end{Theorem}
\ProofPointerToAppendix
\revm{In Appendix~~\ref{sec:pres-c-compl} we demonstrate that under certain circumstances, queries also preserve c-completeness.}

\section{Preservation of C-Completeness}
\label{sec:pres-c-compl}


\mypar{\abbrTIs}
We now demonstrate that positive queries preserve c-\-com\-plete\-ness if the input is a labeling produced by the c-complete labeling scheme $\doTULTI$ (\Cref{sec:create-labelings}). 
To show this, we observe that the existence of a world for in which two $\pwK$-elements $\pwe{k_1}$ and $\pwe{k_2}$ are both minimal
then $\GlbK$ commutes with addition and multiplication, and 
standard $\semK$-relational semantics preserve c-completeness.

\begin{Lemma}\label{lem:min-factors-when-same-min}
  Let $\pwe{k_1}, \pwe{k_2} \in \pwkDom$ for some possible world semiring $\pwK$. If there exists $i \in \pwDom$ such that $\GlbK(\pwe{k_1}) = \pwe{k_1}[i]$ and $\GlbK(\pwe{k_2}) = \pwe{k_2}[i]$, then the following holds:
  \begin{gather*}
    \GlbK(\pwe{k_1} \pwAdd \pwe{k_2}) = \GlbK(\pwe{k_1}) \addK \GlbK(\pwe{k_2}) = (\pwe{k_1} \pwAdd \pwe{k_2})[i]\\
    \GlbK(\pwe{k_1} \pwMult \pwe{k_2}) = \GlbK(\pwe{k_1}) \multK \GlbK(\pwe{k_2}) = (\pwe{k_1} \pwMult \pwe{k_2})[i]
  \end{gather*}
\end{Lemma}
\begin{proof}
 Recall that $\pwWorld{i}$ is a homomorphism (Lemma~\ref{lem:poss-world-is-homomorphism}), so $(\pwe{k_1} \pwAdd \pwe{k_2})[i] = \pwe{k_1}[i] \addK \pwe{k_2}[i]$ and $\GlbK(\pwe{k_j}) = \pwe{k_j}[i]$ for $j \in \{1,2\}$.
 Thus, $(\pwe{k_1} \pwAdd \pwe{k_2})[i]= \GlbK(\pwe{k_1}) \addK \GlbK(\pwe{k_2})$. Next, $\GlbK(\pwe{k_1} \pwAdd \pwe{k_2}) = (\pwe{k_1} \pwAdd \pwe{k_2})[i]$ which holds if for any $j \neq i \in W$ we have $(\pwe{k_1} \pwAdd \pwe{k_2})[i] \ordK (\pwe{k_1} \pwAdd \pwe{k_2})[j]$. Since $\pwe{k_1}[i] = \GlbK(\pwe{k_1})$ and $\GlbK$ is defined based on the natural order, we
 know that $\pwe{k_1}[i] \ordK \pwe{k_1}[j]$ and analog for $\pwe{k_2}$ we have $\pwe{k_2}[i] \ordK \pwe{k_2}[j]$. Lemma~\ref{lem:order-factors-through-operations} then implies $(\pwe{k_1} \pwAdd \pwe{k_2})[i] \ordK (\pwe{k_1} \pwAdd \pwe{k_2})[j]$. 
\ifnottechreport{
 The proof for multiplication is analogous.
}
\iftechreport{
 The proof for multiplication is analog using Lemma~\ref{lem:order-factors-through-operations} to show that $(\pwe{k_1} \pwMult \pwe{k_2})[i] \ordK (\pwe{k_1} \pwMult \pwe{k_2})[j]$ for any $j \in \pwDom$.
}
\end{proof}

To demonstrate c-completeness preservation for \abbrTIs we have to demonstrate that the encoding of a \abbrTI as a $\pwK$-database fulfills the precondition of Lemma~\ref{lem:min-factors-when-same-min}.

\begin{Lemma}\label{lem:tip-has-same-position-min}
  Let $\pdb$ be a $\pwK$-database that represents a \abbrTI. Then there exists $i \in \pwDom$ such that for any tuple $\tup$: 
  \begin{center}
    $\GlbK(\pdb(\tup)) = \pdb(\tup)[i]$.
  \end{center}

\end{Lemma}
\begin{proof}
  Consider the possible world $\db$ defined as follows:
  \begin{align*}
    \db(\tup) =
    \begin{cases}
      \GlbK(\pdb(\tup)) &\mathtext{if} \prob(\tup) = 1\\
      \zeroK &\mathtext{otherwise}\\
    \end{cases}
  \end{align*}
  This world exists, because in a \abbrTI all tuples with probability $p=1$ have annotation $\oneOf{\mathbb B}$ in all worlds.
  Furthermore, since the tuples are independent events, there must exist one world containing no tuples with probability $p < 1$. 
  Let $i$ denote the identifier of this world and denote by $\db = \pwWorld{i}(\pdb)$.
  \textbf{(Case 1)} $\prob(\tup) = 1$ and so $\forall j \in \pwDom: \pdb(\tup)[i] = \pdb(\tup)[j]$.
  \textbf{(Case 2)} $\prob(\tup) < 1$ and $\db(\tup) = \pdb(\tup)[i] = \zeroK$.
  Because $\forall k \in \kDom: \zeroK \ordK k$, it follows that $\GlbK(\pdb(\tup)) = \zeroK = \pdb(\tup)[i]$.  As a result, $\forall t \in \tupDom: \GlbK(\pdb(\tup)) = \db(\tup) = \pdb(\tup)[i]$
\end{proof}


Lemmas~\ref{lem:min-factors-when-same-min} and~\ref{lem:tip-has-same-position-min} together imply that our labeling approach preserves c-completeness if the input is a \abbrTI.

\begin{Corollary}
Let $\TUL$ be a labeling for a \abbrTI $\pdb$ computed as $\doTUL_{TI}(\pdb)$. Then $\raPlus$  
over $\TUL$ preserves c-completeness.
\end{Corollary}

\mypar{\abbrXDBs}
In general, $\raPlus$ queries over labelings derived from \abbrXDBs using our labeling scheme $\doTULX$ from Section~\ref{sec:create-labelings} do not preserve c-completeness. 
We present a sufficient condition for a query to preserve c-completeness over such a labeling. To this end, we define \textit{x-keys}, constraints that ensure that alternatives within the scope of an x-tuple are not all identical if projected on a set of attributes $A$. Since our labeling scheme for \abbrXDBs is c-complete, queries preserve c-completeness unless a result tuple that is certain is derived from multiple \emph{correlated} uncertain input tuples. Since x-tuples from an \abbrXDB are independent of each other, this can only be the case if a result tuple is derived from alternatives of an x-tuple $\xTup$ from every possible world (i.e., where $\xTup$ is not optional).  
Such a situation can be avoided if it is guaranteed that it is impossible for a result tuple to be derived from all alternatives of an x-tuple.


\begin{Definition}[x-key]\label{def:x-key}
  Let $\rel$ be an x-relation with schema $\relschema$. A set of attributes $A \subseteq \relschema$ is called an x-key for $\rel$ iff 
  \[
    \forall \xTup \in R: (\xTup\,\text{is optional}) \vee \revc{\card{\xTup} = 1} \vee (\exists t_1, t_2  \in \xTup: t_1[A] \neq t_2[A])
\]
\end{Definition}

An x-key is a set of attributes $A$  such that for any x-tuple $\xTup$ that is not optional 
and has more than one alternative, there exists at least two alternatives that differ in $A$.
\iftechreport{The following lemma states that a superset of an x-key is also an x-key.

\begin{Lemma}\label{lem:superset-of-xkey-is-xkey}
Let $A \subseteq B \subseteq \relschema$ where $\relschema$ is the schema of an x-relation $\rel$. If $A$ is an x-key for $R$, then so is $B$.
\end{Lemma}
\begin{proof}
Whether the first condition or first subcondition of the second condition of Definition~\ref{def:x-key} hold for an x-tuple is independent of the particular choice of x-key. If the second subcondition is true (which trivially implies that the first subcondition is true), then two alternatives of the x-tuple differ on $A$ which trivially implies that they differ on a superset of $A$.
\end{proof}
}
We prove that for any \abbrXDB  $\pdb$, if a conjunctive, self-join free query $\query$ (a query using selection, projection, and join that accesses no relation more than once) returns at least one x-key per accessed relation, then the query preserves c-completeness. 

\begin{Theorem}\label{theo:x-db-c-completeness}
  Let $\TUL$ be a labeling for an \abbrXDB $\pdb$ computed using $\doTULX$. Consider a conjunctive query $\query$ in canonical form $\projection_{\mathcal{A}}(\selection_\theta(R_1 \times \ldots \times R_n))$ with $R_i \neq R_j$ for all $i \neq j \in \{1, \ldots, n\}$. Query $\query$ preserves c-completeness if $\mathcal{A}$ contains an x-key for every relation $R_i$ accessed by $\query$.
\end{Theorem}
\begin{proof}
  Let $\pdb = \{R_1, \ldots, R_n\}$ be an x\hyp{}database, $\pdb' = \{R_1', \ldots, R_n'\}$ its encoding as a $\pwKof{\semB}$-database, $\TUL$ a c-complete labeling for $\pdb'$ derived using $\doTULX$, and $\query$ be a selfjoin-free query of the form $\projection_{\mathcal{A}}(\selection_\theta(R_1 \times \ldots \times R_n))$ such that $\mathcal{A}$ contains an x-key for every relation $R_i$ for $i \in \{1, \ldots, n\}$. Any selfjoin-free $\raPlus$ query without union can be brought into this form.
We have to show that $\query(\TUL)$ is a c-complete labeling for $\query(\pdb')$. We prove this claim by contradiction. For sake of the contradiction assume that $\query(\TUL)$ is not a c-complete labeling. Then there has to exist a tuple $t \in \query(\pdb')$ such that $\query(\TUL)(t) = F$ and $\pwCertOf{\semB}(\query(\pdb')(t)) = T$.
Recall that $\addOf{\semB} = \vee$ and $\multOf{\semB} = \wedge$. Unfolding definitions of relational algebra operators over $\semK$-relations we get:
\begin{align*}
\query(\TUL)(\tup) &= \bigvee_{u: u[A] = \tup \wedge \forall i \in \{1,\ldots,n\}: u[R_i] = t_i} \left(\bigwedge_{i=1}^{n} \TUL(t_i)\right) \wedge \theta(u)\\
\query(\pdb')(\tup) &= \hspace{-1cm} \sum_{u: u[A] = \tup \wedge \forall i \in \{1,\ldots,n\}: u[R_i] = t_i} \hspace{-1.5cm}R'_1(t_1) \multOf{\pwKof{\semB}} \ldots \multOf{\pwKof{\semB}} R'_n(t_n) \multOf{\pwKof{\semB}} \theta(u)
\end{align*}
Note that for result tuples $u$ of the crossproduct for which $u \not\models \theta$ we have $\theta(u) = F$ (respective $\theta(u) = \zeroOf{\pwKof{\semB}}$). Thus, any monomial (product) corresponding to such a $u$ will evaluate to $F$ ($\zeroOf{\pwKof{\semB}}$). Thus, we can equivalently write the above expressions as shown below where the $j$ values identify monomials for which $u  \models \theta$ WLOG assuming that there are $m \in \semN$ such monomials. 
\ifnottechreport{ We get:
$\query(\TUL)(\tup) = \bigvee_{\forall j \in \{1,\ldots,m\}} \left(\bigwedge_{i=1}^{n} \TUL(t_{j_i}) \right)$ and $\query(\pdb')(\tup) = \sum_{\forall j \in \{1,\ldots,m\}} \prod_{i=1}^{n} R'_i(t_{j_i})$.
}
\iftechreport{
\begin{align*}
\query(\TUL)(\tup) &= \bigvee_{\forall j \in \{1,\ldots,m\}} \left(\bigwedge_{i=1}^{n} \TUL(t_{j_i}) \right)\\
  \query(\pdb')(\tup) &= \sum_{\forall j \in \{1,\ldots,m\}} \prod_{i=1}^{n} R'_i(t_{j_i})
\end{align*}
}
We use $b_{j_i}$ to denote $\TUL(t_{j_i})$ and $\pwe{k_{j_i}}$ to denote $R'_i(t_{j_i})$.
Based on our assumption we know:
\ifnottechreport{
$
\bigvee_{\forall j \in \{1,\ldots,m\}} \left(\bigwedge_{i=1}^{n} b_{j_i} \right) = F
$.
}
\iftechreport{
$$
\bigvee_{\forall j \in \{1,\ldots,m\}} \left(\bigwedge_{i=1}^{n} b_{j_i} \right) = F
$$
}
So this can only be the case if for every $j \in \{1, \ldots, m\}$ there exists $f \in \{1, \ldots, n\}$ such that $b_{j_f} = F$.
For any $j \in \{1, \ldots, m\}$ let $min_j$ denote the smallest such $f$, i.e., the first element in the $j^{th}$ conjunct that is false and let $t_{min_j}$ denote the corresponding tuple. Based on the fact that $\TUL = \doTULX(\pdb')$ and that $\doTULX$ is c-complete, we know that if $\TUL(t_{min_j}) = F$ then $t_{min_j}$ is not certain. 
We will use this fact to derive a contradiction with the assumption $\pwCertOf{\semB}(\query(\pdb')(\tup)) = T$. For that, we partition the set of monomials from $\query(\pdb')(\tup)$ into two subsets $M_1$ and ${M_1}^{\mathcal{C}}$ where $M_1$ contains the identifiers $j$ of all monomials  such that $min_j = 1$ and ${M_1}^{\mathcal{C}}$ contains all remaining monomials.  
\ifnottechreport{
We will show that $\pwCertOf{\semB}(\sum_{j \in M_1} \prod_{i=1}^{n} k_{j_i}) = F$, then $\pwCertOf{\semB}(\sum_{j \in {M_1}^{\mathcal{C}}} \prod_{i=1}^{n} k_{j_i}) = F$, and finally $\pwCertOf{\semB}(\query(\pdb')(\tup)) = \pwCertOf{\semB}(\sum_{j \in M_1} \prod_{i=1}^{n} k_{j_i} \addOf{\pwKof{\semB}} \sum_{j \in {M_1}^{\mathcal{C}}} \prod_{i=1}^{n} k_{j_i}) = F$ which is the contradiction we wanted to derive.
}
\iftechreport{
We will show that $$\pwCertOf{\semB}(\sum_{j \in M_1} \prod_{i=1}^{n} k_{j_i}) = F$$, then $$\pwCertOf{\semB}(\sum_{j \in {M_1}^{\mathcal{C}}} \prod_{i=1}^{n} k_{j_i}) = F$$, and finally $$\pwCertOf{\semB}(\query(\pdb')(\tup)) = \pwCertOf{\semB}(\sum_{j \in M_1} \prod_{i=1}^{n} k_{j_i} \addOf{\pwKof{\semB}} \sum_{j \in {M_1}^{\mathcal{C}}} \prod_{i=1}^{n} k_{j_i}) = F$$ which is the contradiction we wanted to derive.
}

\ifnottechreport{
First, consider $\sum_{j \in M_1} \prod_{i=1}^{n} k_{j_i}$. Since $\multOf{\semB} = \wedge$ and $\multOf{\pwKof{\semB}}$ is defined as point-wise application of $\wedge$ to a vector $\pwe{k} \in \pwKof{\semB}$ we have $\pwe{k} \multOf{\pwKof{\semB}} \pwe{k'} \ordOf{\pwKof{\semB}} \pwe{k}$ for any $\pwe{k}$ and $\pwe{k'}$. Thus, $\sum_{j \in M_1} \prod_{i=1}^{n} k_{j_i} \ordOf{\pwKof{\semB}} \sum_{j \in M_1} k_{j_1}$. We  show $\pwCertOf{\semB}(\sum_{j \in M_1} k_{j_1}) = F$ from which follows $\pwCertOf{\semB}(\sum_{j \in M_1} \prod_{i=1}^{n} k_{j_i}) = F$.
}
\iftechreport{

First, consider $$\sum_{j \in M_1} \prod_{i=1}^{n} k_{j_i}$$ Since $\multOf{\semB} = \wedge$ and $\multOf{\pwKof{\semB}}$ is defined as point-wise application of $\wedge$ to a vector $\pwe{k} \in \pwKof{\semB}$ we have $\pwe{k} \multOf{\pwKof{\semB}} \pwe{k'} \ordOf{\pwKof{\semB}} \pwe{k}$ for any $\pwe{k}$ and $\pwe{k'}$. Thus, $$\sum_{j \in M_1} \prod_{i=1}^{n} k_{j_i} \ordOf{\pwKof{\semB}} \sum_{j \in M_1} k_{j_1}$$. We  show $\pwCertOf{\semB}(\sum_{j \in M_1} k_{j_1}) = F$, from which follows $$\pwCertOf{\semB}(\sum_{j \in M_1} \prod_{i=1}^{n} k_{j_i}) = F$$.
}

By construction we have that $t_{j_1}$ is not certain
for all $j$ in $M_1$. Now consider the set of x-tuples from $R_1$ for which the tuples $t_{j_1}$ are alternatives. WLOG let $\xTup_1, \ldots, \xTup_l$ be these x-tuples. Now consider an arbitrary x-tuple $\xTup$ from this set and let $s_1, \ldots, s_o$ be its alternatives that are present in $M_1$. We know that none of the  $s_i$ are certain based on the fact that alternatives are disjoint events and x-tuples are independent of each other.
We distinguish 2 cases: either $\xTup$ is optional or $\xTup$ is not optional.
In the latter case based on the fact that the query result 
contains an x-key for $R_1$ we know that there exists at least one alterative $s$ of $\xTup$ that is neither in $M_1$ nor in ${M_1}^{\mathcal{C}}$. To see why this is the case observe that its presence in $M_1$ would violate the x-key while by construction ${M_1}^{\mathcal{C}}$ only contains tuples $t$ from $R_1'$  which are certain.  
Next we construct a possible world $w \in \pwDom$ from $\pdb$ which does not contain any of the $t_{j_1}$ which means that $k_{j_1}[w] = F$. In turn, this implies that $\sum_{j \in M_1} k_{j_1}  = F$. We construct $w$ as follows: for every x-tuple $\xTup$ from $\xTup_1, \ldots, \xTup_l$ we either include no alternative of $\xTup$ if the x-tuple is optional 
or an alternative that is not present in $M_1$.
Now further partition ${M_1}^{\mathcal{C}}$ into two subsets: $M_2$ which contains all monomials for which $min_j = 2$ and ${M_2}^{\mathcal{C}}$ for all remaining monomials. Then using an argument symmetric to the one given for $M_1$ above we can construct a possible world for which $\sum_{j \in M_2} \prod_{i=1}^{n} t_{j_i}[w] = F$ and, thus, $\pwCertOf{\semB}(\sum_{j \in M_2} \prod_{i=1}^{n} k_{j_i}) = F$. Because the x-tuples from $M_1$ and $M_2$ are from different relations there is no overlap between these sets of x-tuples. Based on the independence of x-tuples in \abbrXDBs this implies that we can also construct a possible world $w$ where $\sum_{j \in M_1} \prod_{i=1}^{n} k_{j_i}[w] \addOf{\pwKof{\semB}} \sum_{j \in M_2} \prod_{i=1}^{n} k_{j_i}[w] = F$ and, thus, $\pwCertOf{\semB}(\sum_{j \in M_1} \prod_{i=1}^{n} k_{j_i} \addOf{\pwKof{\semB}} \sum_{j \in {M_2}} \prod_{i=1}^{n} k_{j_i}) = F$. We can now continue this construction to include $M_3$, $M_4$, and so on. Note that we are guaranteed that $M_n$ contains all monomials that will be left over at this point, because we started from the observation that at least one $k$ in every monomial corresponds to a tuple $t$ which is not certain.  
It follows that
\ifnottechreport{
$
\pwCertOf{\semB}(\query(\pdb')(\tup)) = \pwCertOf{\semB}\left(\sum_{o=1}^{n} \left(\sum_{j \in M_o} \prod_{i=1}^{n} k_{j_i} \right) \right)  = F
$.
  }
\iftechreport{
$$
\pwCertOf{\semB}(\query(\pdb')(\tup)) = \pwCertOf{\semB}\left(\sum_{o=1}^{n} \left(\sum_{j \in M_o} \prod_{i=1}^{n} k_{j_i} \right) \right)  = F
$$
}
which contradicts our assumption that $\pwCertOf{\semB}(\query(\pdb')(\tup)) = T$ and thus concludes the proof.
\end{proof}
\section{Implementation}
\label{sec:implementation}

We now discuss the implementation of a \abbrUADB as a query rewriting front-end built on top of a relational DBMS.
A $\uaK{\semK}$-relation with schema $\relschema(A_1, \ldots, A_n)$ annotated with a pairs of $\semK$-elements $\uae{d}{c}$ is encoded by a $\semK$-relation $\relschema'(A_1, \ldots, A_n, \cAttr)$ where the annotation of each tuple encodes $d$ and attribute $\cAttr$ stores $c$.
We specifically implement \abbrUADBs for bag semantics, as this is the model used by most DBMSes.
In contrast to $\semN$-relations where the multiplicity of a tuple is stored as its annotation, relational databases represent a tuple $\tup$  with multiplicity $n$ as $n$ copies of $t$. While in principle we could use attribute $\cAttr$ to store $c$ for each copy of $t$, alternatively we can use $\cAttr$ as a boolean marker and mark $c$ copies of $\tup$ as certain ($1$) and the remaining $d - c$ copies as uncertain ($0$)  as shown in the example in Section~\ref{sec:intro}.
We believe that this approach is easier to interpret and, thus, 
apply it here. 

Our frontend rewriting engine receives queries of the form $\query(\uadb)$ over an $\uaK{\semN}$-annotated database $\uadb$ with schema $\{\; \relschema_i(A_1, \ldots, A_n) \;\}$.
It rewrites such a query into an equivalent query $\rewrN{\query}(\db)$ over a classical bag-relational database $\db$ with schema $\{\; \relschema_i'(A_1, \ldots, A_n, C) \;\}$ where $C \in \{ 0, 1 \}$ denotes the uncertainty label.
The rewriting $\rewrN{\cdot}$ is defined through a set of rules given in Figure~\ref{fig:rewriteRules}. \tabularnewline
To support queries over a wide range of incomplete and probabilistic data models we allow the user to specify the data model of each input relation. Our rewriting engine uses SQL implementations of the labeling schemes and extraction of best guess worlds from Section~\ref{sec:create-labelings} to transform such an input relation into our encoding of $\uaK{\semN}$-relations.
\begin{figure}
{\footnotesize
\begin{eqnarray*}
  \rewrN{R}                              &=&
    \text{A Labeled \texttt{R} (see Section~\ref{sec:create-labelings})}\\
  \rewrN{\selection_\theta(Q)}           &=&
    \text{\lstinline{SELECT * FROM Q WHERE} }\theta\\
  \rewrN{\projection_{A_1\ldots A_n}(Q)} &=&
    \text{\lstinline{SELECT A1, ..., AN, C FROM Q}}\\
  \rewrN{Q_1 \bowtie_\theta Q_2}         &=&
    \text{\lstinline{SELECT Q1.*, Q2.*, Q1.C*Q2.C AS C}}\\
    &&\text{\lstinline{FROM Q1, Q2 WHERE} }\theta\\
  \rewrN{Q_1 \union Q_2}                 &=&
    \text{\lstinline{Q1 UNION ALL Q2}}
\end{eqnarray*}
}\\[-3mm]
\caption{Query rewrite rules}
\label{fig:rewriteRules}
\trimfigurespacing
\end{figure}
We implement our approach as a middleware over a database system through an extension of SQL. An input query is first parsed, translated into a relational algebra graph,  rewritten using $\rewrN{\cdot}$, and then converted back to SQL for execution.
\ifnottechreport{We formally prove the correctness of our rewriting and show SQL implementations of our labeling schemas from  Section~\ref{sec:create-labelings} in~\cite{FH18}.}
\iftechreport{\subsection{Relational Algebra Rewriting and Correctness}
\label{sec:corr-rewr}

To prove that the rewriting defined above is correct, we first formally define the function $\bagEnc$ implementing the  encoding of a $\uaK{\semN}$-database as an $\semN$-database and restate $\rewrN{}$ as relational algebra rewriting rules. Afterwards, we prove that this rewriting correctly encodes  $\uaK{\semN}$ query semantics.
In the following, we use $\{t \mapsto k\}$ to denote a singleton relation where tuple $t$ is annotated with $k$ and all other tuples are annotated with $0$. Recall that $\arity{R}$ denotes the arity (number of attibutes) of a relation.

\begin{Definition}[Multiset encoding]
$\bagEnc(R)$ is a function from $\uaK{\semN}$-relations to $\semN$-relations. Let $R$ be a $\uaK{\semN}$-relation with schema ${A_1,\ldots ,A_n}$. Let $R'$ be an $\semN$-relation with schema ${A_1,\dots ,A_n,\uAttr}$  that is the result of $\bagEnc(R)$ for some $R$. $\bagEnc$ and its inverse are defined as:
\begin{align*}
	\bagEnc(R) &=\bigcup_{t \in \aDom^{\arity{R}}} \{(t,1) \mapsto \cOf{R(t)} \} \\
	&\hspace{2cm}\cup \{(t,0) \mapsto \dOf{R(t)} - \cOf{R(t)} \} \\
	\bagEnc^{-1}(R') &=\bigcup_{t \in \aDom^{\arity{R}}} {t \mapsto (R'(t,1), R'(t,0)+R'(t,1))}
\end{align*}
\end{Definition}

We define $\bagEnc$ over databases as applying $\bagEnc$ to every relation in the database.
Note that even though we define the encoding for bag semantics here, it can be generalized to any $\uaK{\semK}$ where semiring $\semK$ has a monus~\cite{Geerts:2010bz} by replacing $-$ with $\monK$ (the monus operation).
Next, we define the relational algebra version of our rewriting $\rewrN{\cdot}$ that translates an input query into a query over the encoding produced by $\bagEnc$. 
Again, the rewriting is defined through a set of rules (one per relational algebra operator). The rules are shown in Figure~\ref{fig:rel-algebra-rewrN}. Here $\schemaOf(\query)$ denotes the schema of the result of query $\query$ and $e \to a$ used in generalized  projection expressions denotes projecting on  the result of evaluating expression $e$ and calling the resulting attribute $a$.

\begin{figure*}[t]
  \centering
  \begin{align*}
   \rewrN{R} &= R\\
	 \rewrN{\selection_\theta(Q)}&=\selection_\theta(\rewrN{Q}) \\
	 \rewrN{\projection_A(Q)}&=\projection_{A,\cAttr}(\rewrN{Q}) \\
	 \rewrN{Q_1 \join_\theta Q_2}&=\projection_{\schemaOf(Q_1 \join Q_2),min(Q_1.\cAttr,Q_2.\cAttr) \rightarrow \cAttr}(\rewrN{Q_1} \join_\theta \rewrN{Q_2}) \\
	 \rewrN{Q_1 \union Q_2} &= \rewrN{Q_1} \union \rewrN{Q_2}
\end{align*}

\caption{Relational algebra rewrite rules implementing $\rewrN{Q}$}
\label{fig:rel-algebra-rewrN}
\end{figure*}

\begin{Theorem}
  Let $\uadb$ be a $\uaK{\semN}$-database and $\query$ an $\raPlus$ query. The following holds:
  \begin{align*}
    \query(\uadb) = \bagEnc^{-1} (\rewrN{\query}(\bagEnc(\uadb)))
  \end{align*}
\end{Theorem}
\begin{proof}
  Straightforward induction over the structure of queries.
  \proofpara{Base case: $\query = R$} WLOG consider a tuple $\tup$ and let $R(\tup) = \uae{d}{c}$. We know that $$\bagEnc(R)(t,0) = \dOf{R(t)} - \cOf{R(t)} = d - c$$ and $$\bagEnc(R)(t,1) = \cOf{R(t)} = c$$
  Let $R'' = \bagEnc^{-1}(\rewrN{\query}(\bagEnc(R))) = \bagEnc^{-1}(\bagEnc(R))$. Then $R''(t) = \uae{R(t,0) + R(t,1)}{R(t,1)} = \uae{d - c + c}{c} =  \uae{d}{c}$.

  \proofpara{Induction Step}
  Assume that the claim holds for queries $\query_1$ and $\query_2$, we have to show that it also holds for applying an operator of $\raPlus$ to the result of these queries. We use ${\query_{\rewrN{}}}_i = \rewrN{\query_i}$ and  ${\rel}_{i}= \query_i(\uadb)$. 

  \proofpara{Selection $\selection_{\theta}(\query_1)$}
  Note that
  $$\rewrN{\selection_\theta(\query_1)} = \selection_\theta(\rewrN{\query_1})$$
  Consider a tuple $\tup$  with $\rel_1(\tup) = \uae{d}{c}$. Let ${\rel_1}'' = \bagEnc^{-1}(\selection_{\theta}(\bagEnc(\rel_1))$.
We have $\selection_{\theta}(\rel_1)(\tup) = \rel_1(\tup) \uaMult{\semN} \theta(\tup)$ and
$$\selection_{\theta}(\bagEnc(\rel_1)(\tup,0)) = (d - c) \cdot \theta(\tup,0)$$
$$\selection_{\theta}(\bagEnc(\rel_1)(\tup,1)) = c \cdot \theta(\tup,1)$$
Since the selection condition does not access attribute $U$, we have $$\theta(\tup,0) = \theta(\tup,1) = 1 \Leftrightarrow \theta(\tup) = \uae{0}{1}$$
Applying the definition of $\bagEnc^{-1}$, we get
$${R_1}''(\tup) = \uae{ (d-c) \cdot \theta(\tup,0) + c \cdot \theta(\tup,1)}{c \cdot \theta(\tup,1)}$$
We now distinguish two cases: either $\tup \models \theta$ and $\tup \not\models \theta$.
First consider the case where $\tup \models \theta$. Then, $\theta(\tup,0) = \theta(\tup,1) =1$ and we get

$${R_1}''(\tup) = \uae{ (d - c + c) \cdot 1}{c \cdot 1} = \uae{d}{c} = R_1(\tup) = R_1(\tup) \uaMult{\semN} \theta(\tup)$$

Now consider the case $\tup \not\models \theta$. Then, $\theta(\tup,0) = \theta(\tup,1) = 0$ and we get

$${R_1}''(\tup) = \uae{ (d - c + c) \cdot 0}{c \cdot 0} = \uae{0}{0} =  R_1(\tup) \uaMult{\semN} \theta(\tup)$$

\proofpara{Natural Join $\query_1 \join_\theta \query_2$}
Let $$\rel'' =  \bagEnc^{-1}(  \rewrN{Q_1 \join_\theta Q_2}(\bagEnc(\rel_1),\bagEnc(\rel_2))$$
and consider a tuple $\tup$ with $\tup \models \theta$ and let $\tup_1 = \tup[R_1]$, $\tup_2 = \tup[R_2]$, $R_i(\tup_i) = \uae{d_i}{c_i}$ for $i \in \{1,2\}$, and
$\query_{res} =  \rewrN{Q_1 \join_\theta Q_2}$.
\begin{align*}
\query_{res} = &\projection_{\schemaOf(Q_1 \join_\theta Q_2),min(Q_1.\cAttr,Q_2.\cAttr) \rightarrow \cAttr}(\query_{join})\\
  \query_{join} = &\rewrN{\query_1} \join_\theta \rewrN{\query_2}
\end{align*}

Based on the induction assumption we have $\rewrN{\query_i}(\tup_i,0) = d_i - c_i$ and $\rewrN{\query_i}(\tup_i,1) = c_i$. Mapping $\bagEnc$ creates two versions of $\tup_i$, thus, there are 4 ways of joining these versions:

\begin{align*}
\query_{join}(\tup, 0,0) &= \rewrN{\query_1}(\tup_1,0) \cdot \rewrN{\query_2}(\tup_2,0) \\&= (d_1 - c_1) \cdot (d_2 - c_2)\\[2mm]
\query_{join}(\tup, 0,1) &= \rewrN{\query_1}(\tup_1,0) \cdot \rewrN{\query_2}(\tup_2,1) \\&= (d_1 - c_1) \cdot c_2\\[2mm]
\query_{join}(\tup, 1,0) &= \rewrN{\query_1}(\tup_1,1) \cdot \rewrN{\query_2}(\tup_2,0) \\&= c_1 \cdot (d_2 - c_2)\\[2mm]
\query_{join}(\tup, 1,1) &= \rewrN{\query_1}(\tup_1,1) \cdot \rewrN{\query_2}(\tup_2,1) \\&= c_1 \cdot c_2\\
\end{align*}

The projection expression $min(Q_1.\cAttr, Q_2.\cAttr)$ maps the first three cases to $(\tup, 0)$ and the last case to $(\tup,1)$. Thus,
\begin{align*}
  \query_{res}(\tup,0) &= (d_1 - c_1) \cdot (d_2 - c_2) + (d_1 - c_1) \cdot c_2 + c_1 \cdot (d_2 - c_2)  \\&= d_1 \cdot d_2 - c_1 \cdot c_2\\[2mm]
  \query_{res}(\tup,1) &= c_1 \cdot c_2
\end{align*}
Finally, we get
\begin{align*}
\rel''(\tup) = \uae{d_1 \cdot d_2 - c_1 \cdot c_2 + c_1 \cdot c_2}{c_1 \cdot c_2} = [\query_1 \join_\theta \query_2](\tup)
\end{align*}

\proofpara{Projection $\projection_{A}(\query_1)$}
$\rewrN{\projection_{A}(\query_1)} = \projection_{A,\cAttr}(\rewrN{\query_1})$.
Recall the definition of projection: $[\projection_A(R_1)](\tup) = \sum_{s[A] = t} R_1(s)$. Consider a tuple $\tup$ and let $\{s_1, \ldots, s_n\}$ be the set of tuples with $s_i[A] = t$. Furthermore, let $s_i = \uae{d_i}{c_i}$ and
$${\rel_1}'' = \bagEnc^{-1}(\projection_{A,\cAttr}(\bagEnc(\rel_1))$$
Then, ${\rel_1}''(\tup) = \uae{k_{t,0} + k_{t,1}}{k_{t,1}}$ for
\begin{align*}
  k_{t,0} &=  {\sum_{(s,0)[A,\uAttr] = (t,0)} \bagEnc(R_1)(s,0)} = {\sum_{i=1}^{n} d_i} - {\sum_{i=1}^{n} c_i}\\
  k_{t,1} &=  {\sum_{(s,1)[A,\uAttr] = (t,1)} \bagEnc(R_1)(s,1)} = {\sum_{i=1}^{n} c_i}\\
\end{align*}
Thus, we get
\begin{align*}
  {\rel_1}''(\tup) &= \uae{k_{t,0} + k_{t,1}}{k_{t,1}}\\
  &= \left[{\sum_{i=1}^{n} d_i} - {\sum_{i=1}^{n} c_i} + {\sum_{i=1}^{n} c_i},{\sum_{i=1}^{n} c_i}\right]\\
  &= \sum_{s[A] = t} R_1(s)\\
  &= \projection_A(R_1)(\tup)
\end{align*}

\proofpara{Union $\query_1 \union \query_2$}
$\rewrN{\query_1 \union \query_2} = (\rewrN{\query_1} \union \rewrN{\query_2})$. Consider a tuple $\tup$  with $\rel_i(\tup) = \uae{d_i}{c_i}$. Let
$${\rel}'' = \bagEnc^{-1}(\bagEnc(\rel_1) \union \bagEnc(\rel_2))$$.

We have $[\rel_1 \union \rel_2](\tup) = \uae{d_1 + d_2}{c_1 + c_2}$ and
\begin{align*}
[\bagEnc(\rel_1) \cup \bagEnc(\rel_2)](\tup,0)) &= (d_1 - c_1) + (d_2 - c_2)  \\
[\bagEnc(\rel_1) \cup \bagEnc(\rel_2)](\tup,1)) &= c_1 + c_2
\end{align*}
Based on this we get
\begin{align*}
  {R_1}''(\tup) &= \uae{ (d_1 - c_1) + (d_2 - c_2) + (c_1 + c_2)}{c_1 + c_2}\\
  &= \uae{ d_1 + d_2}{c_1 + c_2}\\
               &= \uae{d_1}{c_1} \uaAdd{\semN} \uae{d_2}{c_2}=  [\rel_1 \union \rel_2](\tup)
\end{align*}

\end{proof}

\subsection{SQL Implementations of Labeling Schemes}
\label{sec:sql-impl-label}

We now show SQL implementations of our methods for extracting best guess worlds and labeling schemes from Section~\ref{sec:create-labelings}.

\mypar{\abbrTIs}
Consider a \abbrTI relation $R(A_1, \ldots, A_n)$ which is stored as a relation $R'(A_1, \ldots, A_n, P)$ where attribute $P$ stores the probabilities of tuples. Recall that we include all tuples $\tup$ where $\prob(\tup) \geq 0.5$ in the best guess world and our labeling scheme for \abbrTIs annotates tuples $\tup$ with $T$ (certain) if $\prob(t) = 1$. In SQL this is expressed as

\begin{lstlisting}[morekeywords={TI,IS,PROBABILITY,WITH}]
SELECT A1, ... An,
       CASE WHEN P = 1
            THEN 1
            ELSE 0
       END AS C
FROM R
WHERE P >= 0.5
\end{lstlisting}

We expect the user to specify the name of the attribute storing the probability for any relation that is marked to be a \abbrTI relation.
The example shown below illustrates how to mark a relation $R$ which stores probabilities in attribute $p$ as a \abbrTI relation.

\begin{lstlisting}[morekeywords={TI,IS,PROBABILITY,WITH}]
SELECT * FROM R IS TI WITH PROBABILITY (p)
\end{lstlisting}

\mypar{\abbrXDBs}
For an x-relation $R(A_1, \ldots, A_n)$ which is stored as a relation $R'(\xAttr, \altAttr, A_1, \ldots, A_n, \pAttr)$ where $\xAttr$ stores identifiers for x-tuples and $\altAttr$ store an identifier for alternatives that is unqiue within the scope of an x-tuple.
For each x-tuple $\xTup$ we pick the alternative with the highest probability if the total probability mass of the x-tuple is larger or equal to $0.5$. We only mark alternatives of x-tuples as certain if $\prob(\xTup) = 1$ and $\card{\xTup} = 1$. In the SQL implementation we make extensive use of analytical functions (SQL's \lstinline!OVER!-clause).

\begin{lstlisting}
SELECT A1, ..., An
       CASE WHEN P = 1
            THEN 1
            ELSE 0
       END AS C
FROM R
WHERE Aid = FIRST_VALUE(Aid) OVER w1
      AND 1 - (sum(P) OVER w2)
          >= max(P) OVER w2
WINDOW w1 AS (PARTITION BY Xid ORDER BY P DESC),
       w2 AS (PARTITION BY Xid)
\end{lstlisting}

When an input relation is identified as an x-relation, we require that the user specifies which attributes stores x-tuple identifies, alternative identifiers, and probabilities. For example, consider the SQL snipplet shown below.

\begin{lstlisting}[morekeywords={X,XID,ALTID,PROBABILITY,IS}]
SELECT *
FROM R IS X WITH XID (tid)
                 ALTID (aid)
                 PROBABILITY (p)
\end{lstlisting}

\mypar{\abbrCtables}
For a \abbrCtable $R(A_1, \ldots, A_n)$ which is stored as a relation $R'(A_1, \ldots, A_n, V_1, \ldots, V_n, \lcAttr)$ where $\lcAttr$ stores the local condition $\phi_\pdb(\tup)$ (as a string) and $V_i$ stores a variable name if $A_i = v$ for some variable $v$ and \lstinline!NULL! otherwise. The SQL implementation of the labeling schema and best guess world computation for \abbrCtables assumes the existence of a UDF \lstinline!isTautology! that implements the tautology check as described in Section~\ref{sec:create-labelings}.

\begin{lstlisting}[morekeywords={IS}]
SELECT A1, ..., An
       CASE WHEN isTautology(LC)
            THEN 1
            ELSE 0
       END AS C
FROM R
WHERE V1 IS NULL AND ... AND Vn IS NULL
\end{lstlisting}

To mark an input as a \abbrCtable the user has to specify which attributes store the $V_i$'s and local condition.

\begin{lstlisting}[morekeywords={CTABLE,LOCAL,CONDITION,VARIABLES,WITH}]
SELECT *
FROM R IS CTABLE WITH VARIABLES (V1, ..., Vn)
                      LOCAL CONDITION (lc)
\end{lstlisting}


}

\section{Related Work} \label{related}

\mypar{Incomplete and probabilistic data models}
Uncertainty was recognized as an important problem by the database community early-on. Codd~\cite{DBLP:journals/tods/Codd79} extended the relational model with null values to represent missing information and proposed to use 3-valued logic to evaluate queries over databases with null values.
Imielinski~\cite{DBLP:journals/jacm/ImielinskiL84} introduced \abbrVtables and \abbrCtables as representations of incompleteness. C-tables are closed under full relational algebra. 
Reiter~\cite{R86} proposed to model databases as logical theories, a model equivalent to \abbrVtables. 
\iftechreport{Abiteboul~\cite{AG85} defined update operations over incomplete databases. }
\iftechreport{Underlying all these models is the possible world semantics. }
Probabilistic data models quantify the uncertainty in incomplete databases by assigning probabilities to individual possible worlds. \abbrTIs~\cite{suciu2011probabilistic} are a prevalent model for probabilistic data where each tuple is associated with its marginal probability and tuples are assumed to be independent.
Green et al.~\cite{DBLP:journals/debu/GreenT06} studied probabilistic versions of \abbrCtables.  
\abbrVCtables generalize \abbrCtables~\cite{5447879,Yang:2015:LOA:2824032.2824055} by allowing symbolic expressions as values.
\ifnottechreport{Probabilistic query processing (PQP) has been studied for several decades (e.g., an important survey is~\cite{suciu2011probabilistic}).}


\iftechreport{
\mypar{Probabilistic Query Processing}
Probabilistic query processing (PQP) has been a field of research for several decades (e.g., an important survey is~\cite{suciu2011probabilistic}).
Computing the marginal probability of a query result tuple can be reduced to weighted model counting and, thus, is \sharpP in general~\cite{DS07b}. Most practical approaches for PQP are either limited to queries which can be answered in \ptime (so-called \textit{safe} queries) and/or compute approximate probabilities for query answers (e.g.,~\cite{OH10}).
Systems implementing PQP include Sprout~\cite{DBLP:conf/sigmod/FinkHOR11}, Trio~\cite{DBLP:conf/vldb/AgrawalBSHNSW06}, MCDB~\cite{jampani2008mcdb}, Mimir~\cite{DBLP:journals/corr/NandiYKGFLG16}, MYSTIQ~\cite{DBLP:conf/sigmod/BoulosDMMRS05}, and many others.
}

\mypar{Certain Answers}
%
Many approaches for answering queries over incomplete databases employ certain answer semantics~\cite{DBLP:journals/jacm/ImielinskiL84,AK91,L16a,GL16,GL17}.  The foundational work by Lipski~\cite{L79a} defined certain answers analogously to our approach, but using minima instead of GLBs.
Computing certain answers is \conpcomplete~\cite{AK91,DBLP:journals/jacm/ImielinskiL84} (data complexity) for first order queries, even for restricted data models such as Codd-tables. This hardness result even holds for conjunctive queries over more complex uncertain data models (e.g., OR-databases~\cite{DBLP:journals/jcss/ImielinskiMV95}).
Thus, it is not surprising that approaches for approximating the set of certain answers have been proposed.
Reiter~\cite{R86} proposed a \ptime algorithm that returns a subset of the certain answers (c-sound) for positive existential queries (and a limited form of universal queries). Guagliardo and Libkin~\cite{GL17,L16a,GL16} 
 propose a query semantics that preserves c-soundness for full relational algebra (first order queries) \revm{for Codd- and \abbrVtables.} Then, \cite{GL17}~defined certain and possible multiplicities for bag semantics, and presented initial thoughts on how to extend~\cite{GL16} for bag semantics. 
 \reva{Our approach works with a wider range of data models and models of uncertainty than \cite{GL16,GL17}, at the cost of being a slightly weaker approximation.  Furthermore, unlike this approach, \abbrUADBs are closed under query evaluation.}
 Sundarmurthy et al.~\cite{sundarmurthy_et_al:LIPIcs:2017:7061} introduced m-tables, which can represent not just uncertainty, but also model information about missing tuples, as well as terms c-soundness/-correctness. This approach works for both set and bag semantics.
 Consistent query answering~\cite{DBLP:series/synthesis/2011Bertossi,DBLP:conf/pods/ArenasBC99} (CQA) is computing certain answers to a query over the incomplete database defined by of all repairs for a database that violates a set of constraints. The complexity of variants of this problem has been studied extensively (e.g.,~\cite{DBLP:journals/ipl/KolaitisP12,DBLP:conf/pods/CaliLR03,DBLP:conf/pods/KoutrisW18}) and several combinations of classes of constraints and queries have been identified that permit first-order rewritings~\cite{FM05,GP17,DBLP:journals/tods/Wijsen12,DBLP:conf/pods/Wijsen10}.
 Geerts et al.~\cite{GP17} study first order under-approximations of certain answers in the context of CQA.

\mypar{Annotated Databases}
%
Green et al.~\cite{Green:2007:PS:1265530.1265535} introduced the semiring annotation framework that we utilize in this work. 
The connection between annotated databases, provenance, and uncertainty has been recognized early-on.
A particular type of semiring annotations, often called Lineage, has been used  for probabilistic query processing~(e.g., see~\cite{suciu2011probabilistic,WT08}).
Green et al.~\cite{Green:2007:PS:1265530.1265535} observed that set semantics incomplete databases can be expressed as $\semK$-relations by annotating each tuple with the set of worlds containing it. We define a more general type of incomplete databases based on $\semK$-relations which is defined for any l-semiring.
%
Kostylev et al.~\cite{DBLP:conf/icdt/KostylevB12} investigate how to deal with dependencies among annotations from multiple domains.
Similar to~\cite{DBLP:conf/icdt/KostylevB12}, we consider ``multi-dimensional'' annotations, 
but for a very different purpose: to extend incomplete databases beyond set semantics. 




\begin{figure}
  \centering
    \includegraphics[width=0.35\textwidth, trim=0.2cm 1.6cm 0cm 1.3cm, clip]{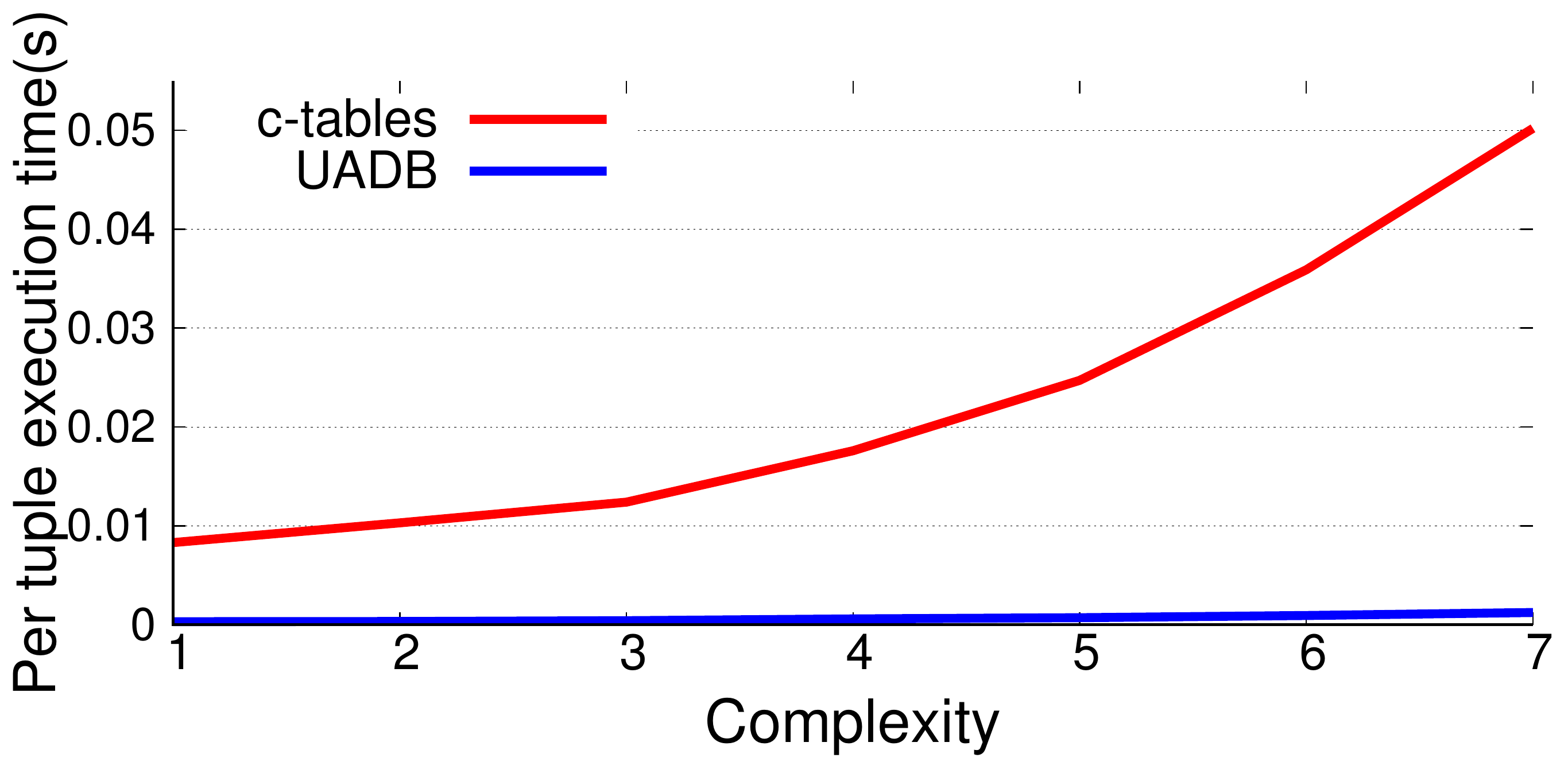}
    \vspace*{-3mm}
  \bfcaption{Certain answers over \abbrCtables}
  \label{fig:ctable}
  \trimfigurespacing
\end{figure}

\section{Experiments}
\label{sec:experiments}

We evaluate the performance of queries over \abbrUADBs implemented on a commercial DBMS\footnote{The DBMS is not identified due to license restrictions}.
We compare \textbf{\abbrUADB}s with
(1)~\textbf{Det}:~Deterministic \abbrBGQP,
(2)~\textbf{Libkin}:~An alternate under-approximation of certain answers~\cite{L16a,GL16},
(3)~\textbf{MayBMS}:~We use MayBMS to compute the full set of possible answers\footnote{Times listed for MayBMS do include only computing possible answers and not \revb{computing probabilities unless stated otherwise.}}, and
(4)~\textbf{MCDB}:~We use MCDB-style~\cite{jampani2008mcdb} database sampling (10 samples) to over-approximate the certain answers.
All experiments are run on a machine with 2$\times$6 core AMD Opteron 4238 CPUs, 128GB RAM, 4$\times$1TB 7.2K HDs (RAID 5). We report the average running time of 5 runs. We also evaluate 
false negative (i.e., a misclassified certain answer) rates for \abbrUADBs and both false negative and false positive rates for other systems. Furthermore, we demonstrate that \abbrBGQP and \abbrUADBs produce answers that are of higher utility (more similar to a ground truth result) than certain answers.

\subsection{Performance Comparison} 
We first use PDBench~\cite{antova2008fast}, a modified TPC-H data generator~\cite{tpch} that introduces uncertainty by generating random possible values for randomly selected cells (attributes).
The generator produces a columnar encoding \revb{optimized for MayBMS}, with 
tables as pairs of tuple identifiers and attribute values.
Ambiguity arises from having multiple values for the same tuple identifier.
We directly run MayBMS queries \revb{(omitting probability computations)} on these columnar tables. For MCDB, we simulate the tuple bundle query using 10 samples. We also apply Libkin by constructing a database instance with nulls from the PDBench tables and applying queries  generated by the optimized rewriting described in~\cite{GL16}. We run deterministic queries and queries generated by our approach on one possible world that is selected by randomly choosing a value for each uncertain cell. For our approach, we treat the input as an \abbrXDB and mark tuples with at least one uncertain cell as uncertain. The three PDBench queries roughly correspond to TPC-H queries Q3, Q6 and Q7.

\begin{figure*}[!ht]
  \begin{minipage}{1.0\linewidth}
\centering
\begin{subfigure}[b]{.31\textwidth}
	\centering
  \includegraphics[width=0.9\textwidth, trim=0.5cm 3cm 2cm 4cm, clip]{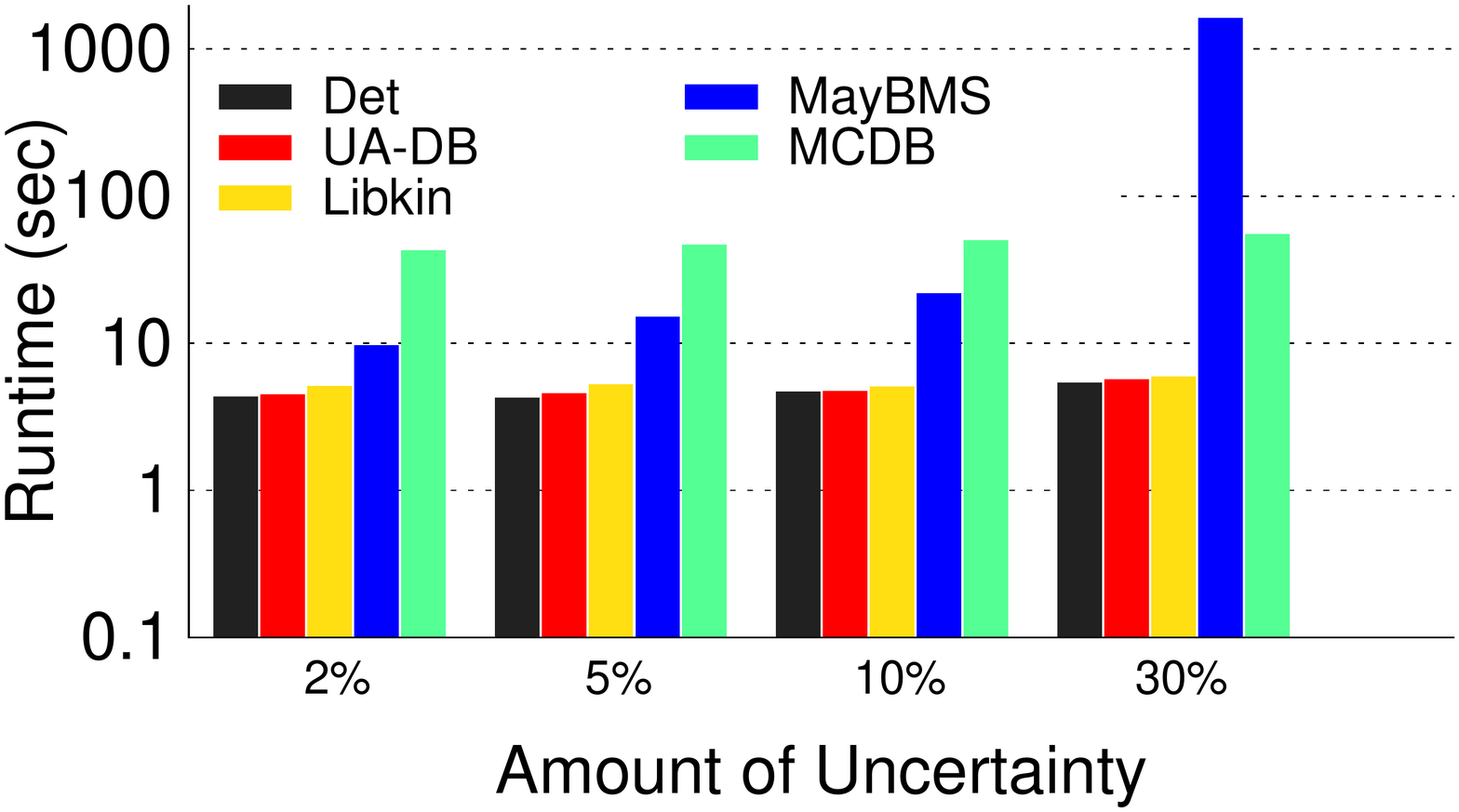}
  \bfcaption{PDBench - Q1}
  \label{fig:q1}
\end{subfigure}
\begin{subfigure}[b]{.31\textwidth}
	\centering
  \includegraphics[width=0.9\textwidth, trim=0.5cm 3cm 2cm 4cm, clip]{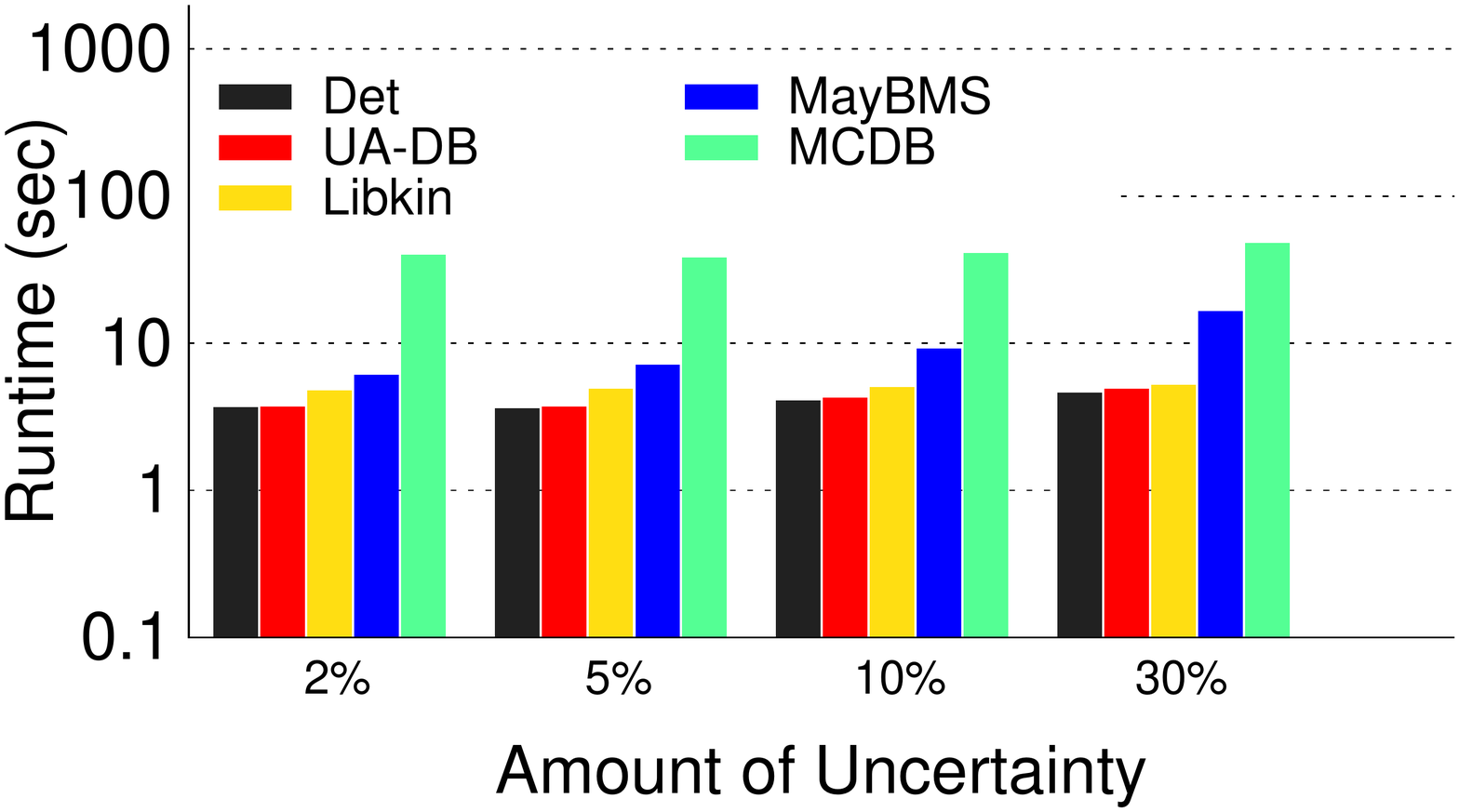}
  \bfcaption{PDBench - Q2}
  \label{fig:q2}
\end{subfigure}
\begin{subfigure}[b]{.31\textwidth}
	\centering
  \includegraphics[width=0.9\textwidth, trim=0.5cm 3cm 2cm 4cm, clip]{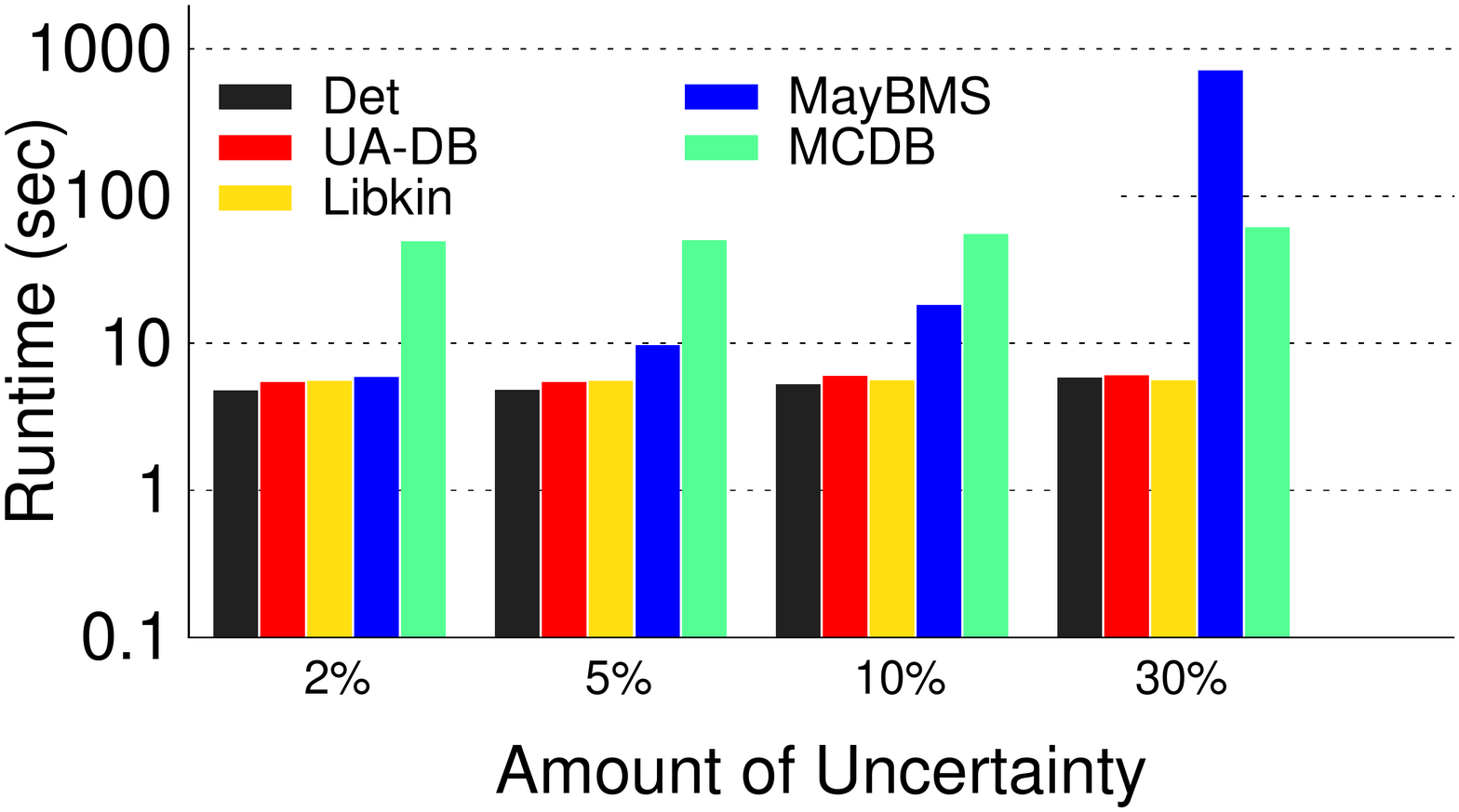}
  \bfcaption{PDBench - Q3}
  \label{fig:q3}
\end{subfigure}
\vspace*{-4mm}
\bfcaption{Performance of PDBench queries -  varying the amount of uncertainty for scale factor 1}
\label{fig:runtime_uncert}
\end{minipage}


\begin{minipage}{1.0\linewidth}

\begin{minipage}[b]{0.65\textwidth}\centering
{\small
  \begin{tabular}{@{}rrrrcrrr@{}} 
& \multicolumn{3}{c}{\textbf{\abbrUADB}} & \phantom{abc}& \multicolumn{3}{c}{\textbf{MayBMS}}\\
\cmidrule{2-4} \cmidrule{6-8}
& $Q1$ & $Q2$ & $Q3$ && $Q1$ & $Q2$ & $Q3$\\ \midrule
$2\%$ & 14,260 & 152,583 & 9,016 && 113,966 & 210,996 & 15,108\\
$5\%$ & 34,041 & 152,432 & 8,619 && 501,114 & 327,052 & 32,438\\
$10\%$ & 61,800 & 152,389 & 8,794 && 2,392,916 & 618,199 & 97,454\\
$30\%$ & 130,581 & 152,885 & 7,994 && 134,054,635 & 3,941,554 & 4,351,782\\
\bottomrule
\end{tabular}
}
\vspace{-3mm}
\bfcaption{Query result sizes (\#rows)}
\label{table:ressize}
\end{minipage}%
\begin{minipage}[b]{0.35\textwidth}\centering
  {\small
\begin{tabular}{@{}rrrr@{}}\toprule
& $Q1$ & $Q2$ & $Q3$ \\ \midrule
$2\%$ & 0 (0\%) & 143,618 (94\%) & 7,861 (87\%)\\
$5\%$ & 1 (0\%) & 130,594 (86\%) & 6,023 (70\%)\\
$10\%$ & 4 (0\%) & 111,120 (73\%) & 3,979 (45\%)\\
$30\%$ & 1 (0\%) & 52,724 (34\%) & 586 (7\%)\\
\bottomrule
\end{tabular}
}
\vspace{-1mm}
\bfcaption{Result certain answer \%}
\label{table:resp}
\end{minipage}
\end{minipage}

\begin{minipage}{1.0\linewidth}
\centering
\begin{subfigure}[b]{.31\textwidth}
	\centering
  \includegraphics[width=0.9\textwidth, trim=0.5cm 3cm 2cm 3cm, clip]{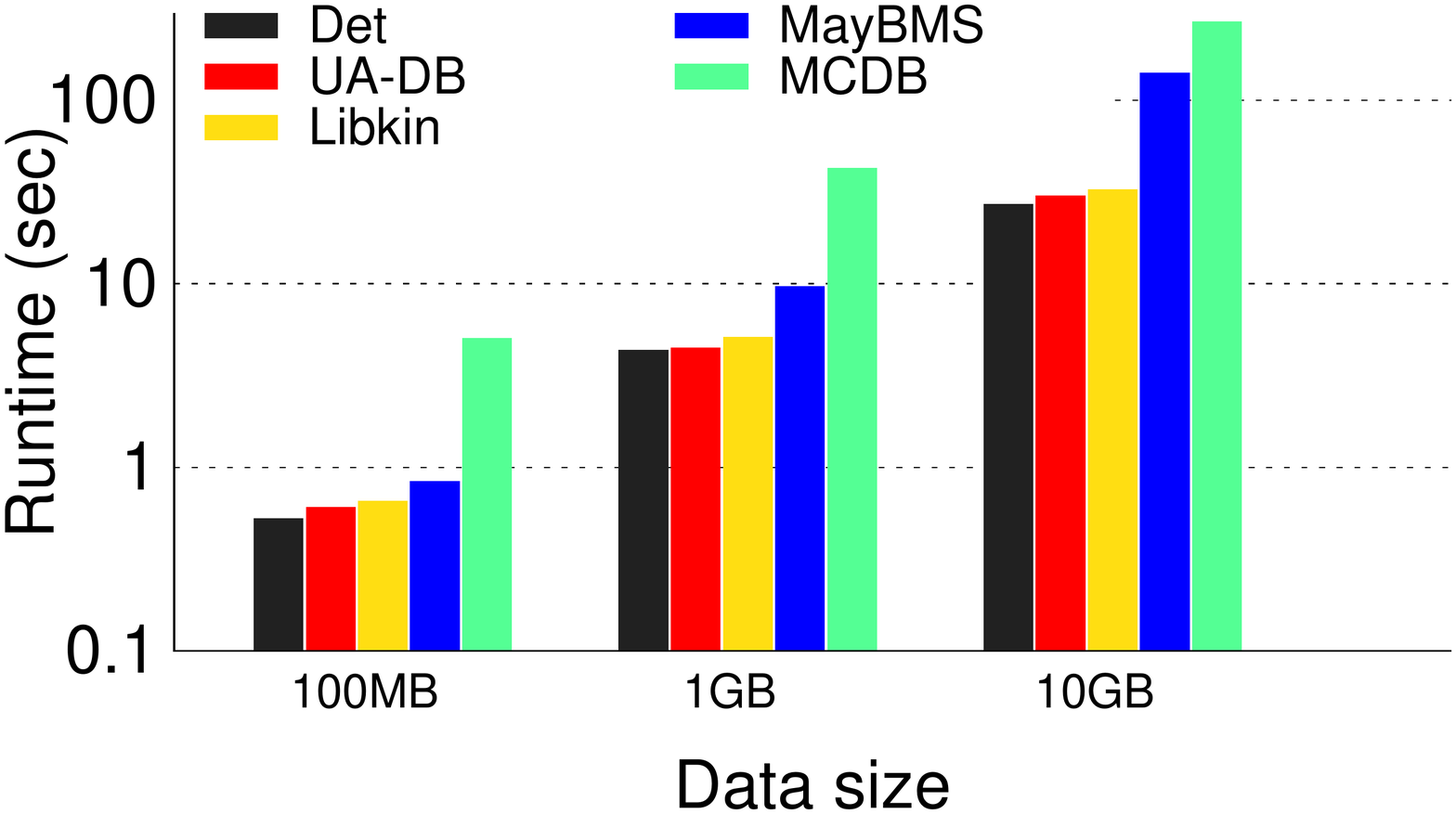}
  \bfcaption{PDBench - Q1}
  \label{fig:q1d}
\end{subfigure}
\begin{subfigure}[b]{.31\textwidth}
	\centering
  \includegraphics[width=0.9\textwidth, trim=0.5cm 3cm 2cm 3cm, clip]{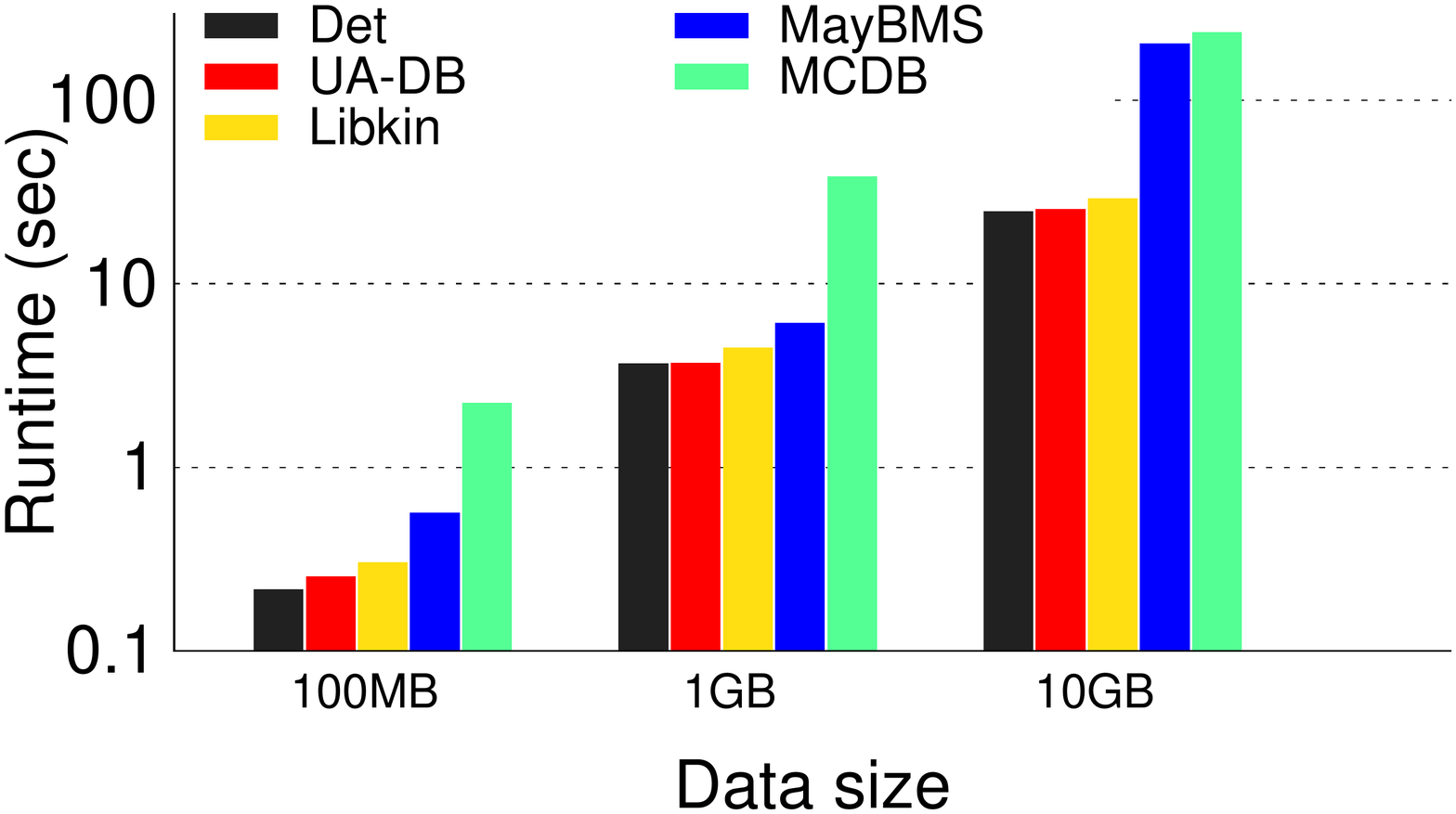}
  \bfcaption{PDBench - Q2}
  \label{fig:q2d}
\end{subfigure}
\begin{subfigure}[b]{.31\textwidth}
	\centering
  \includegraphics[width=0.9\textwidth, trim=0.5cm 3cm 2cm 3cm, clip]{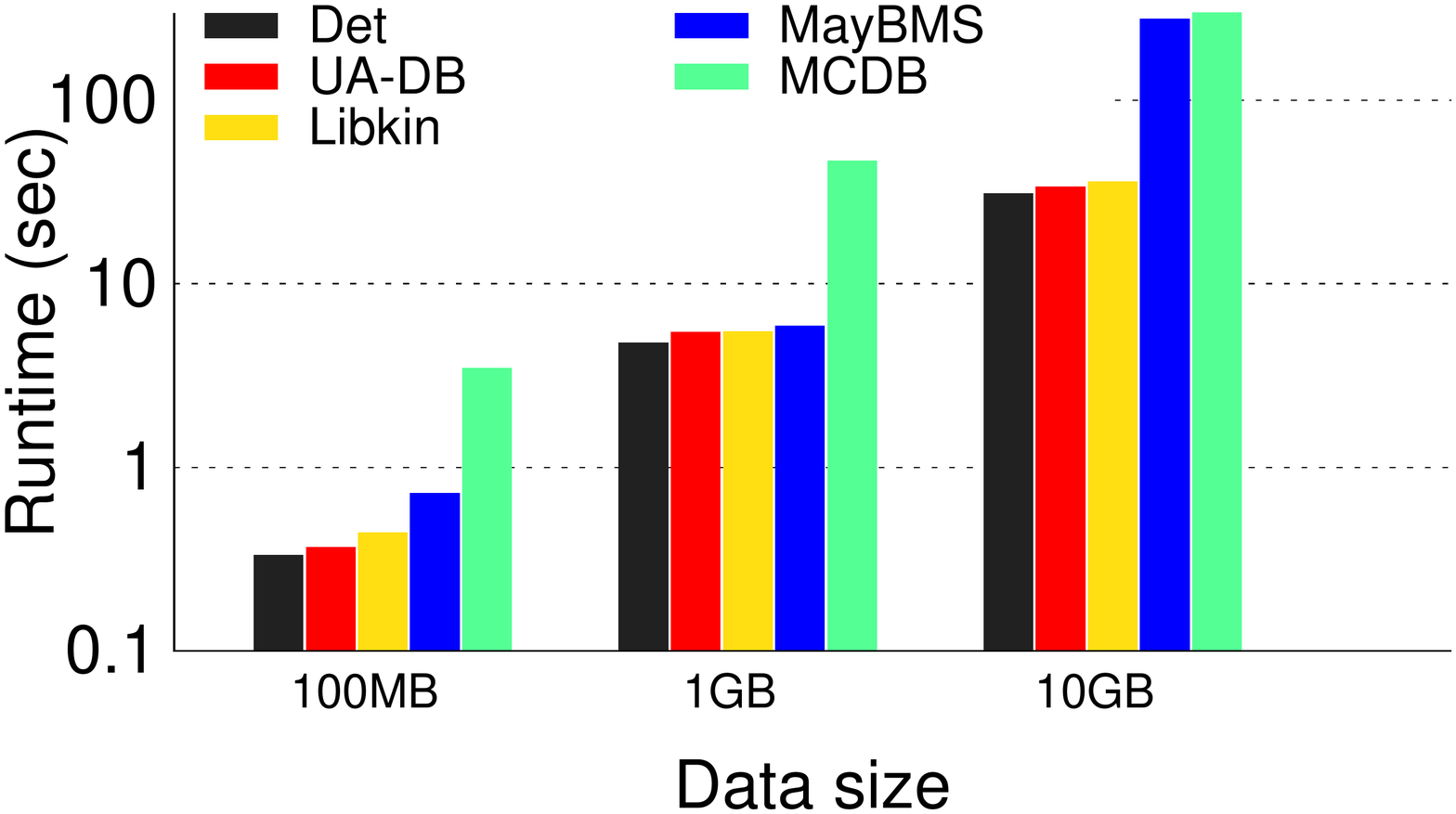}
  \bfcaption{PDBench - Q3}
  \label{fig:q3d}
\end{subfigure}
\vspace*{-3mm}
\bfcaption{Performance of PDBench queries -  varying database size for 2\% uncertainty}
\label{fig:runtime_data}
\end{minipage}
\vspace{2mm}
\input{exp_figures_2.tex}

\end{figure*}


\mypar{Amount of uncertainty}
Using a scale factor 1 database ($\sim$1GB of data per possible world), we evaluate scalability with respect to amount of uncertainty.
Using PDBench, we vary the percentage of uncertain cells in the range 2\%, 5\%, 10\% and 30\%.
Each uncertain cell has up to 8 possible values.
Figure~\ref{fig:runtime_uncert} shows the runtime results for the three PDBench queries.
As expected, runtimes for \abbrUADBs and Libkin are close to deterministic query processing.
The slight overheads arise from propagating uncertainty annotations and dealing with nulls, respectively. Furthermore, \abbrUADBs have to output additional tuples that belong to the \termBGW, but are not certain.
Libkin slightly outperforms \abbrUADBs{} for query Q3 at levels of uncertainty above 10\%, since the query's join only returns certain tuples and, thus, there is no overhead for dealing with nulls.
For queries Q1 and Q2, \abbrUADBs{} slightly outperform Libkin as the overhead of dealing with nulls outweighs the overhead for returning a larger result and propagating uncertainty annotations.
MCDB effectively needs to evaluate queries once for each sample, and so runs more than 10 times slower than deterministic query processing.
MayBMS has a reasonable, but still noticeable overhead at lower levels of uncertainty.
As uncertainty increases, the query output size in MayBMS increases roughly cubically for Q1 and Q3, and it begins to perform \revb{several orders of magnitude (the plots use log scale)} slower than \abbrUADBs.
MayBMS performs better for the simple selection query Q2.
To better understand the performance of MayBMS, we show result sizes (number of tuples) for each query varying amounts of uncertainty in  Figure~\ref{table:ressize}. Our approach 
produces the same number of results as deterministic processing.
Conversely, MayBMS returns the full set of possible answers and, thus the result size increases dramatically as uncertainty increases. 
We also show the percentage of certain answers for each query per input uncertainty level in Figure~\ref{table:resp}.
The unexpected increase of result size for Q1 over \abbrUADBs is caused by a shift in the correlation between attributes \texttt{o\_orderkey} and \texttt{l\_shipdate} that affects the number of tuples passing Q1's selection condition resulting from PDBench choosing values for uncertain cells independently.




\mypar{Dataset size}
To evaluate scalability, we use datasets with scale factors (SF) 0.1 (100MB), 1 (1GB) and 10 (10GB) and fix the uncertainty percentage (2\%).
The results 
are shown in Figure~\ref{fig:runtime_data}. Again \abbrUADBs and Libkin exhibit performance similar to deterministic queries as well as certain answers and MCDB is again roughly 10 times slower (the sample size is 10). MayBMS's relative overhead over deterministic processing increases with data set size. For instance, for Q1 the overhead is $\sim$ 60\% for SF 0.1 and $\sim$500\% for SF 10.

\mypar{Certain Answers over \abbrCtables}
As an example of a more complex incomplete data model, we evaluate the performance of \abbrUADBs{} against computing certain answers over \abbrCtables.
We create a synthetic table with 8 attributes. For each tuple we randomly chose half of its attributes to be variables and the other half to be floating point constants.
We construct random  queries by assembling a scaling number of randomly chosen self-joins, projections, or selections.
We count query execution time using \abbrUADBs{}.
The exact certain tuples of the \abbrCtables{} result are computed by instrumenting the query to calculate a local condition for every result tuple and running the Z3 constraint solver (\url{https://github.com/Z3Prover/z3}) over the resulting boolean expression.
An answer is certain iff its local condition is a tautology.
\iftechreport{
Each local condition's complexity of depends on how tuples are combined by the query.
Joins combine tuple conditions by conjunction, while projections and unions combine matching result tuples by disjunction.
Selection extends the local condition on rows where the selection predicate accesses a variable-valued attribute.
Each selection operator further increases complexity for each conjunction, disjunction or arithmetic operation.
}
Figure~\ref{fig:ctable} shows the average runtime per  result tuple for both c-tables and \abbrUADBs averaged over all randomly generated queries.
The x-axis is the number of operators (i.e., selection, projection or join) in the source query.
Overhead for \abbrCtables increases super-linearly in query complexity from about $27\times$ to over $40\times$.

\ifnottechreport{
\begin{figure}
\centering
{
\footnotesize
\begin{tabular}{rrrrrr}\toprule
Dataset & Rows & Cols & $U_{Attr}$ & $U_{Row}$\\ \midrule
Building Violations & 1.3M & 35 & 0.82\% & 12.8\%\\
Shootings in Buffalo & 2.9K & 21 & 0.24\% & 2.1\%\\
Business Licenses & 63K & 25 & 1.39\% & 14.0\%\\
Chicago Crime & 6.6M & 17 & 0.21\% & 0.9\%\\
Contracts & 94K & 13 & 1.50\% & 19.2\%\\
Food Inspections & 169K & 16 & 0.34\% & 4.6\%\\
Graffiti Removal & 985K & 15 & 0.09\% & 0.8\%\\
Building Permits & 198K & 19 & 0.42\% & 5.3\%\\
Public Library Survy & 9.2K & 99 & 1.19\% & 14.2\%\\
\bottomrule
\end{tabular}
}
\vspace*{-4mm}
\bfcaption{Real World Datasets}
\label{table:realdata}
\trimfigurespacing
\vspace*{1mm}
\end{figure}
}

\iftechreport{
\begin{figure*}
{\small
  \begin{tabular}{rrrrrp{10cm}}
\toprule
Dataset & Rows & Cols & $U_{Attr}$ & $U_{Row}$ & URL\\ \midrule
Building Violations & 1.3M & 35 & 0.82\% & 12.8\% & {\url{https://data.cityofchicago.org/Buildings/Building-Violations/22u3-xenr}}\\
Shootings in Buffalo & 2.9K & 21 & 0.24\% & 2.1\% & {\url{http://projects.buffalonews.com/charts/shootings/index.html}}\\
Business Licenses & 63K & 25 & 1.39\% & 14.0\% & {\url{https://data.cityofchicago.org/Community-Economic-Development/Business-Licenses-Current-Active/uupf-x98q}}\\
Chicago Crime & 6.6M & 17 & 0.21\% & 0.9\% & {\url{https://data.cityofchicago.org/Public-Safety/Crimes-2001-to-present/ijzp-q8t2}}\\
Contracts & 94K & 13 & 1.50\% & 19.2\% & {\url{https://data.cityofchicago.org/Administration-Finance/Contracts/rsxa-ify5}}\\
Food Inspections & 169K & 16 & 0.34\% & 4.6\% & {\url{https://data.cityofchicago.org/Health-Human-Services/Food-Inspections/4ijn-s7e5}}\\
Graffiti Removal & 985K & 15 & 0.09\% & 0.8\% & {\url{https://data.cityofchicago.org/Service-Requests/311-Service-Requests-Graffiti-Removal/hec5-y4x5}}\\
Building Permits & 198K & 19 & 0.42\% & 5.3\% & {\url{https://https://www.kaggle.com/aparnashastry/building-permit-applications-data/data}}\\
Public Library Survy & 9.2K & 99 & 1.19\% & 14.2\% & {\url{https://www.imls.gov/research-evaluation/data-collection/public-libraries-survey/explore-pls-data/pls-data}}\\
\bottomrule
\end{tabular}
}
\vspace*{-2mm}
\bfcaption{Real World Datasets}
\label{table:realdata}
\end{figure*}
}

\subsection{Real world datasets}
We use multiple real world datasets \ifnottechreport{(\Cref{sec:datasets})}
from a wide variety of domains to evaluate how our approach performs for real world data.
We use SparkML to impute missing values in the datasets, treating alternative imputations as a source of uncertainty.
The resulting dataset, represented as an \abbrXDB, was converted to a \abbrUADB using $\doTULX$ (\ref{sec:create-labelings}), which marks all tuples with at least one uncertain attribute as uncertain.
Figure~\ref{table:realdata} shows basic statistics for the cleaned datasets\ifnottechreport{:}\iftechreport{ and URLs for the original datasets: } the \#rows, \#attributes, the percentage of attribute values that are uncertain  ($U_{attr}$), and the percentage of rows marked as uncertain by our c-complete labeling scheme ($U_{row}$). 


\begin{figure}
{
\footnotesize
  \begin{tabular}{rrrrrr}\toprule
& $Q1$ & $Q2$ & $Q3$ & $Q4$ & $Q5$\\ \midrule
Overhead & 2.28\% & 1.81\% & 1.32\% & 2.88\% & 3.51\%\\
Error Rate & 0.55\% & 0.37\% & 0\% & 0.92\% & 0.29\%\\
\bottomrule
\end{tabular}
}
\vspace*{-2mm}
\bfcaption{Real Query Results}
\label{table:realquery}
\trimfigurespacing
\end{figure}

\begin{figure*}
  \begin{minipage}{1.0\linewidth}
\centering
\begin{subfigure}[b]{.31\textwidth}
  \centering
  \includegraphics[width=\textwidth, trim=0.5cm 1.8cm 0cm 4cm, clip]{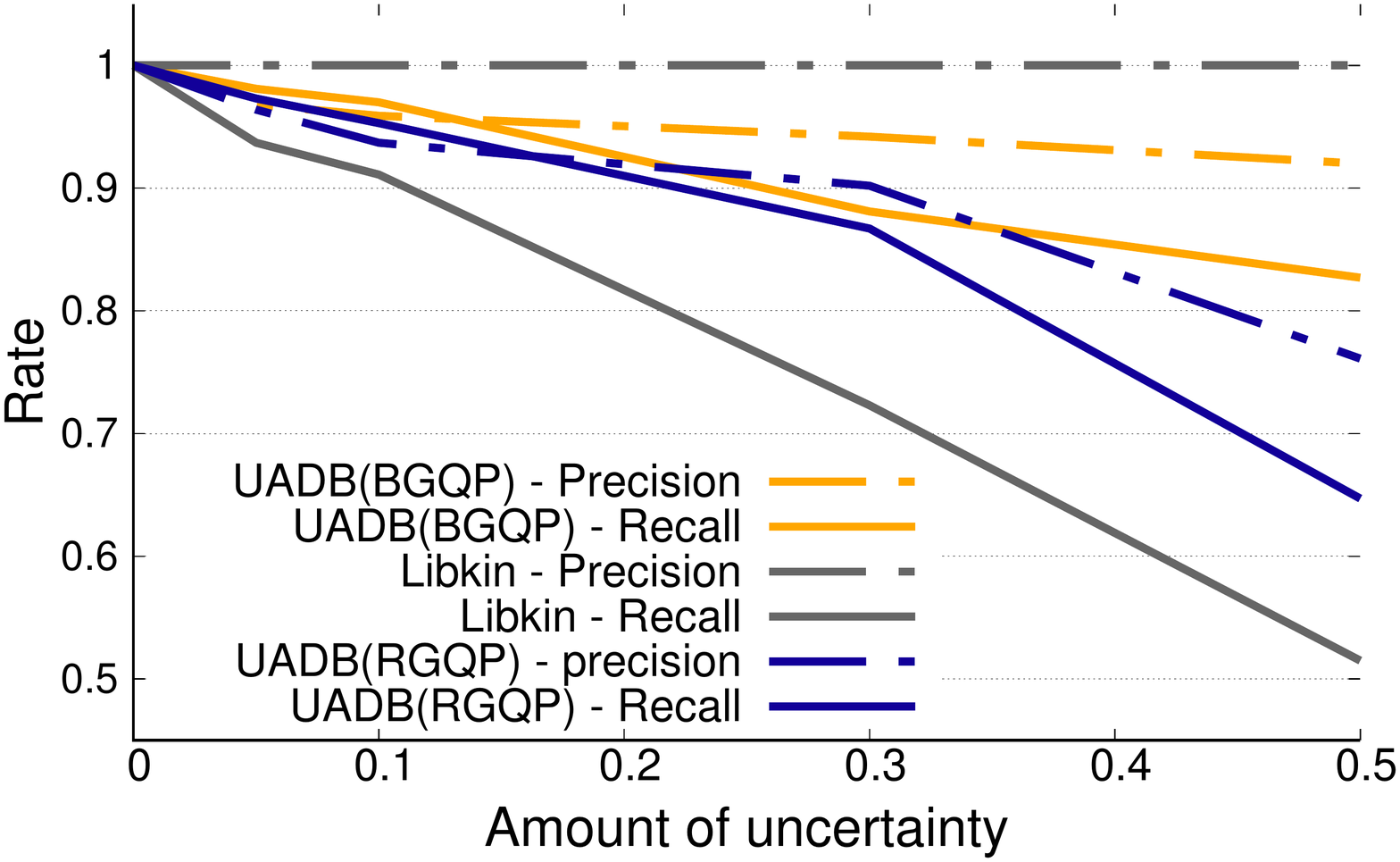}
  \bfcaption{Income Survey}
  \label{fig:u1}
\end{subfigure}
\begin{subfigure}[b]{.31\textwidth}
  \centering
  \includegraphics[width=\textwidth, trim=0.5cm 1.8cm 0cm 4cm, clip]{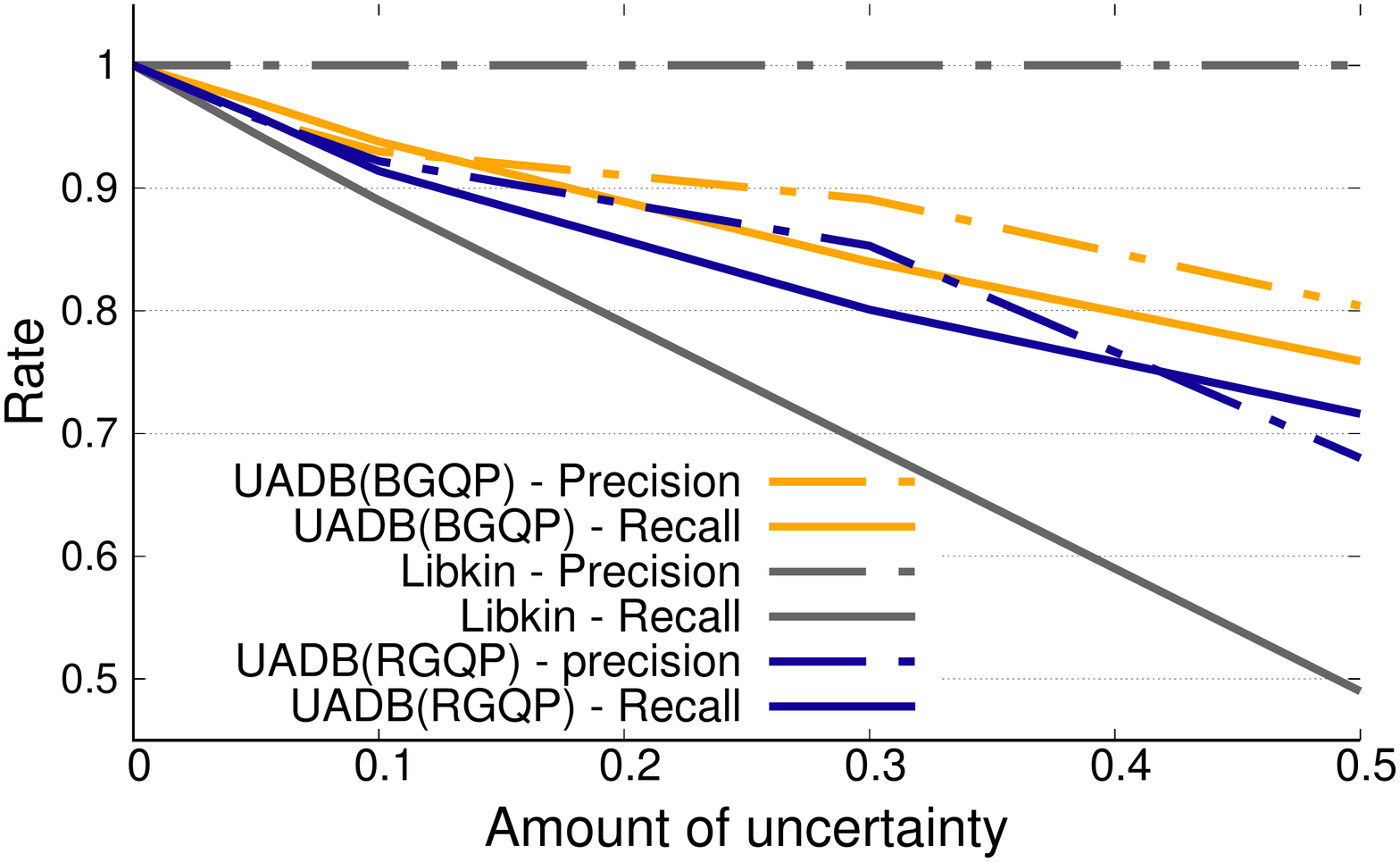}
  \bfcaption{Buffalo News}
  \label{fig:u29}
\end{subfigure}
\begin{subfigure}[b]{.31\textwidth}
  \centering
  \includegraphics[width=\textwidth, trim=0.5cm 1.8cm 0cm 4cm, clip]{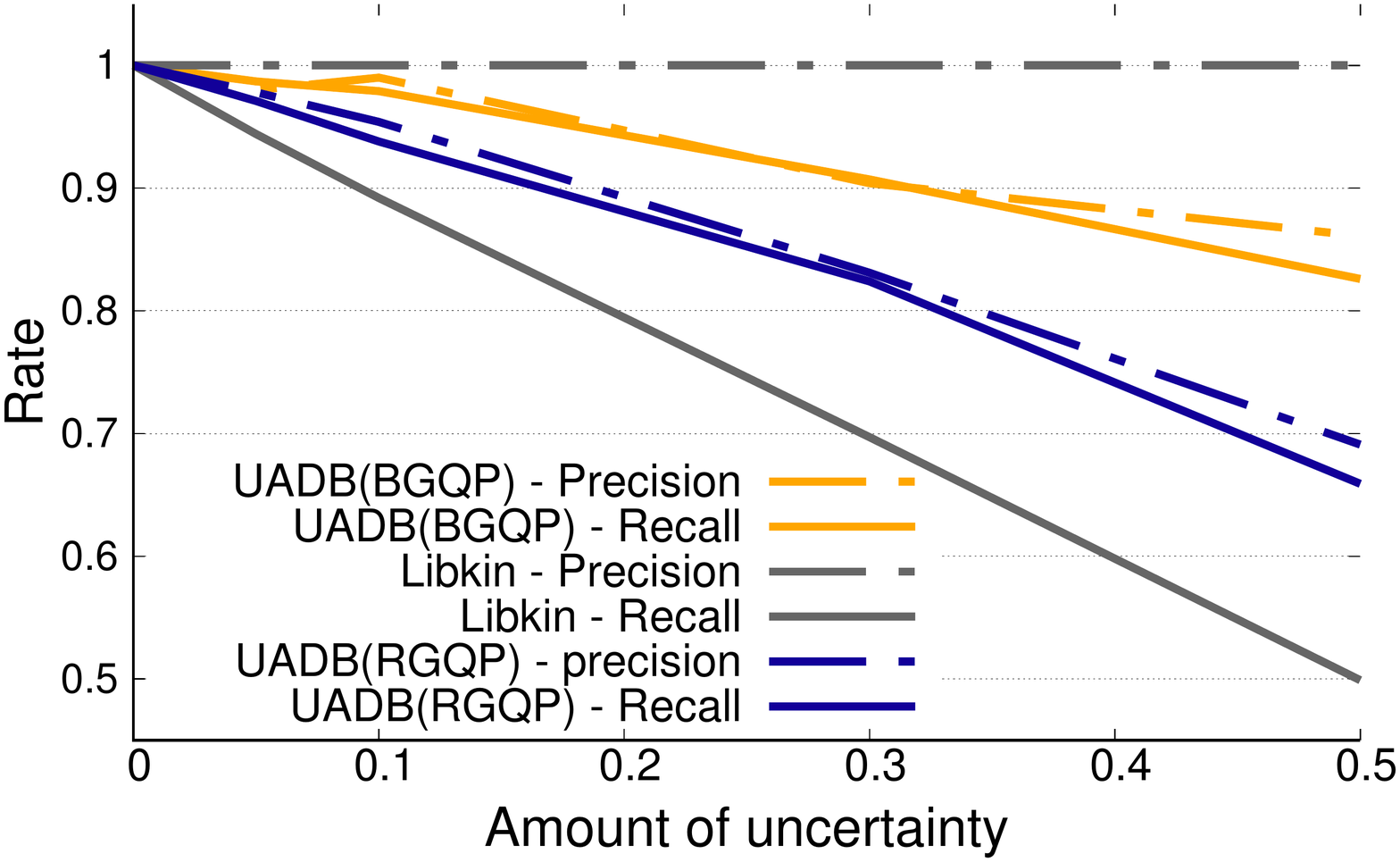}
  \bfcaption{Business License}
  \label{fig:u3}
\end{subfigure}
\vspace*{-4mm}
\bfcaption{Utility -  varying the amount of uncertainty}
\label{fig:utility}
\end{minipage}
\trimfigurespacing
\end{figure*}

\begin{figure}
  \resizebox{1\linewidth}{!}{
  \begin{minipage}{1.25\linewidth}
    {\footnotesize
\renewcommand{\arraystretch}{0.9}
  \begin{tabular}{r|rrrrrr}\toprule
& & \textbf{UADB} & \textbf{MB-02} & \textbf{MB-05} & \textbf{MB-10} & \textbf{MB-20} \\ \midrule
\multirow{2}{*}{$Q_{P1}$} & time (ms) & 3.1 & 4.0 (4.1) & 22.7 (22.3) & 308.5 (305.6) & 4.8k (4.7k) \\
					& error & 0\% & 0\% (0\%) & 0\% (0\%) & 0\% (0\%) & 0\% (0\%)\\
\midrule
\multirow{2}{*}{$Q_{P2}$} & time (ms) & 4.4 & 6.8 (6.8) & 28.4 (28.5) & 374.5 (367.0) & 8.8k (7.0k) \\
					& error & 1.6\% & 0\% (0\%) & 0\%( 0\%) & 0\% (0.5\%) & 0.5\% (1.1\%) \\
\midrule
    \multirow{2}{*}{$Q_{P3}$} & time (ms) & 7.6 & 54.0 (20.3) &  17.0k (10.8k) & 289.7k (118.6k) & 3.5m (1.1m) \\
					& error & 3.0\% & 0\% (0.1\%) & 0.1\% (0.1\%) & 0.2\% (0.3\%) & 0.6\% (1.1\%)\\

\bottomrule
\end{tabular}
}
\end{minipage}
}
\vspace*{-4mm}
\bfcaption{Probabilistic database}
\label{table:probdatabase}
\trimfigurespacing
\end{figure}

\mypar{Incompleteness}
To measure the false negative rate (fraction of answers that are misclassified as uncertain) of our approach, we use queries that project on a randomly chosen set of attributes.
The rationale for this is that based on Theorem~\ref{theo:x-db-c-completeness}, projecting an uncertain tuple onto a subset of its certain attributes (no x-key) causes the tuple to be a certain answer.
This is the primary situation in which \abbrUADBs mis-classify results, so this experiment represents a worst case scenario for \abbrUADBs.
We evaluate queries which project on a randomly chosen set of attributes and measure the false negative rate (\textit{FNR}). 
Figure~\ref{fig:incomp_viol} to~ \ref{fig:incomp_pls} show the distribution of the FNR (min, 25-percentile, median, 75-percentile, max) for queries with a fixed number of projection attributes. 
As expected, the FNR decreases as the number of projection attributes grows, but 
is low in general (less than 20\% in the worst case for the worst case dataset). For most datasets, the median FNR is below 5\% when at least half of the attributes are involved in the projection.
Note that selection and join do not produce any ``new'' false negative results (see proof of Theorem~\ref{theo:x-db-c-completeness}). 
This shows that 
for real world datasets with correlated errors, the FNR is typically low. 

\iftechreport{\subsection{Real Queries}}
\ifnottechreport{\mypar{Real Queries}}
We next evaluate the effectiveness of our approach on five queries over the real world datasets\ifnottechreport{ (we present the SQL code and descriptions of these queries in~\cite{FH18}).}
\iftechreport{(the SQL code and descriptions of these queries are shown below).  }
Most of our real world datasets are from open data portals that associate analyses (e.g., visualizations) with datasets.
Test queries are reverse engineered from these analyses.
We measure the performance overhead and false negative rate of \abbrUADBs. Performance overhead is measured as the slowdown relative to deterministic query processing. 
As Figure~\ref{table:realquery} shows, our approach introduces a slight (less than $4\%$) overhead for these queries.
The worst case ($4\%$) is Q5, which involves a join operator.
All other queries, which contain only selections and projections have under $3\%$ overhead.
In each case, we saw a $1\%$ false negative rate or lower.
Notably, Q3 returns no misclassified results due to its small result size.
\BG{Can we do Libkin for these?}

\mypar{Probabilistic databases}
\reva{
We next compare the performance and accuracy of \abbrUADBs against MayBMS.
For this experiment, we use a BI-DB (an x-DB with probabilities), varying the number of alternatives for each block and use three queries $Q_{P1}$, $Q_{P2}$ and $Q_{P3}$ of varying complexity described in~\cite{FH18}.
For MayBMS, we treat tuples with probability $p\geq1$ as certain. MayBMS may report prob. $>1$ due to rounding/approximation errors. 
Figure~\ref{table:probdatabase} shows both runtime and error rate for both systems, with 2, 5, 10, or 20 alternatives. For MayBMS we show the result for exact probability computation and for approximation using the scheme from~\cite{OH10} with an error bound of 0.3 (shown in parentheses).
Note that query processing in a \abbrUADB is independent of the number of possible worlds. Only a single alternative is used for each block.
We observe that MayBMS's results include both false positives and false negatives.
Because results are computed by summing floating point numbers, even MayBMS' exact probability computations
exhibits a small amount of rounding error that is more noticeable for larger number of alternatives (e.g., MB-20).
Although approximating probabilities can improve performance especially for complex queries, as the number of possible alternatives increases, MayBMS is still orders of magnitude slower than \abbrUADBs. $Q_{P3}$ includes a self-join which further slows MayBMS down due to the increase in possible worlds and expression complexity.
}

\begin{figure}[t]
  \centering
  \includegraphics[width=0.95\linewidth, trim=0.5cm 6cm 0cm 5cm]{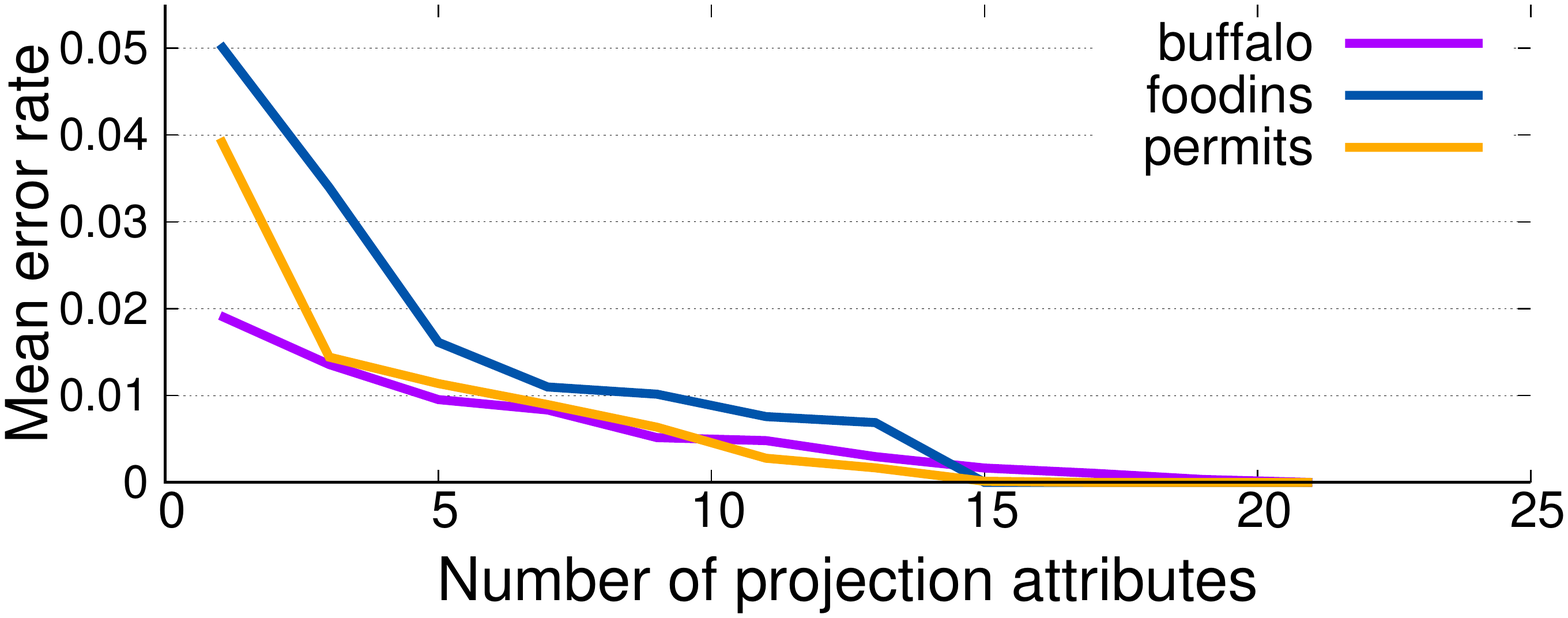}
  \caption{\revm{Bag semantics - mislabelings}}
  \label{fig:bag-semantics---misc}
  \trimfigurespacing
\end{figure}
\begin{figure}[t]
  \centering
  \includegraphics[width=1\linewidth,trim={0.5cm 5cm 0cm 3.5cm}]{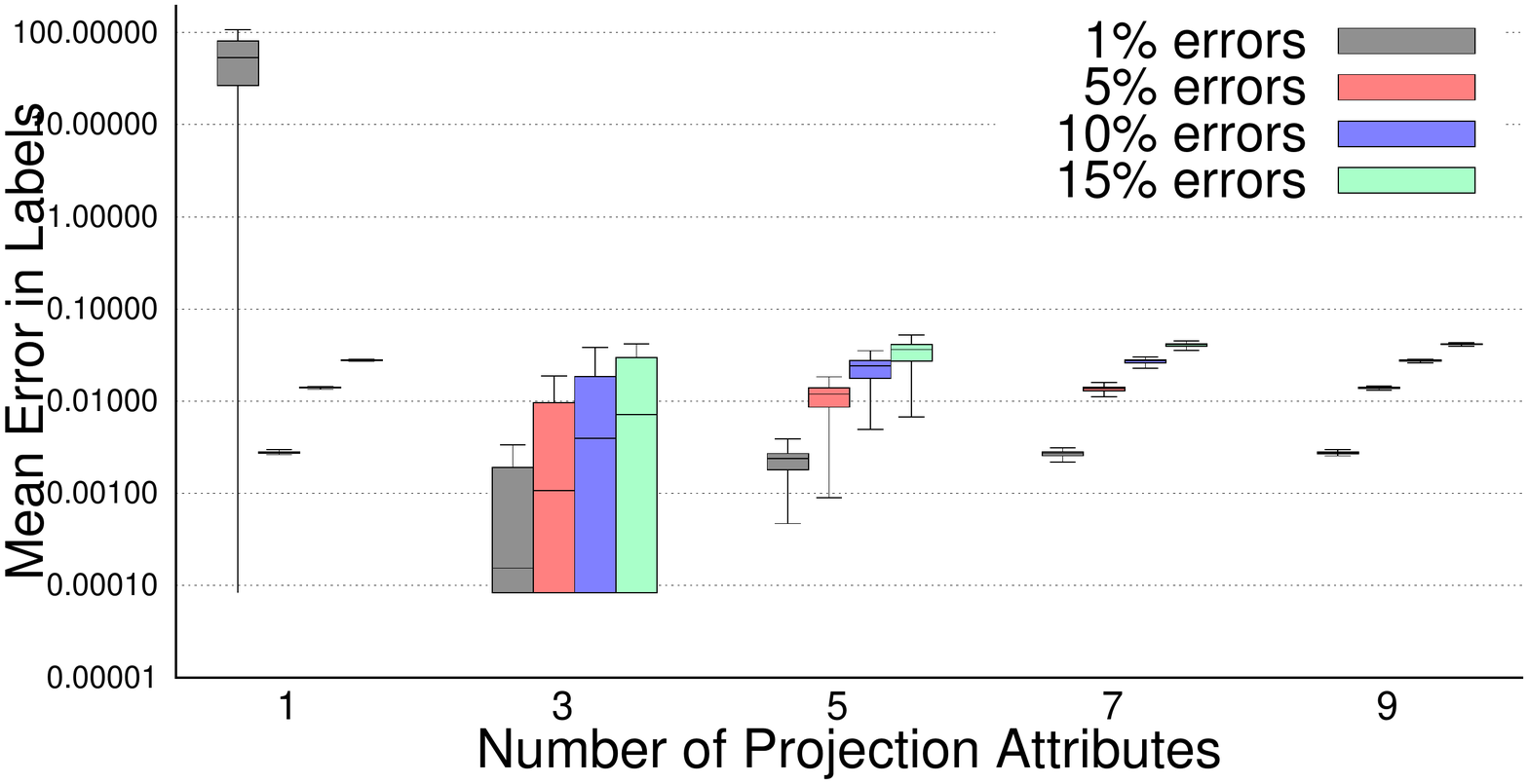}
  \caption{\revm{Access control semiring - mislabelings}}
  \label{fig:access-control-misclassified}
  \trimfigurespacing
\end{figure}
\mypar{Beyond Set Semantics} \revm{
  In this experiment we evaluate the FNR of our approach using bag semantics (semiring $\semN$) and the access control semiring $\semAC$~\cite{Green:2007:PS:1265530.1265535}.
  For the bag semantics experiment we evaluate projections under bag semantics over some of the real world datasets from Figure~\ref{table:realdata}. The results for this experiment are shown in
Figure~\ref{fig:bag-semantics---misc}. Observe that the FNR is similar to the set semantics case.
The access control semiring  annotates each tuple with an access control level (one of \textbf{0} - ``nobody can access the data', \textbf{T} is ``top secret'', \textbf{S} is ``secret'', \textbf{C} is ``confidential'', and \textbf{P} is ``public'') to determine what clearance-level is necessary to view the tuple. Addition (multiplication) is max (min) according to the following order over the elements $\mathbf{0} < \mathbf{T} < \mathbf{S} < \mathbf{C} < \mathbf{P}$.
For this experiment, we emulate a scenario where private information in a dataset is heuristically detected and secured with an $\semAC$ annotation.
Using 5 real world datasets from Figure~\ref{table:realdata}, we randomly assigned access control labels to each tuple in the dataset and then created multiple labelings with 1\%, 2\%, 5\%,
and 10\% of misclassified tuples. We evaluated random projection queries over these datasets and measured the amount of misclassified query results weighted by the distance between the certain annotation and the labeling, e.g., the distance of $\mathbf{C}$ and $\mathbf{T}$ is $\frac{2}{5} = 0.4$. In Figure~\ref{fig:access-control-misclassified}, we vary the number of projection attributes and show the distribution of the amount of misclassified query results over 9 randomly selected projection queries for 5 datasets. The FNR increases when the input error rate is increased, but is quite low in most cases.
}

\iftechreport{
\subsection{Query Descriptions}
\mypar{Q1}
This query is expressed over the Chicago crime dataset.
The query returns all crime ids and case numbers for all thefts, domestic batteries, and criminal damages. Here, attribute \texttt{IUCR} (Illinois Uniform Crime Reporting code) is a system for specifying crime types.

\begin{lstlisting}
SELECT id, case_number,
       CASE IUCR
         WHEN 0820 then 'Theft'
         WHEN 0486 then 'Domestic Battery'
         WHEN 1320 then 'Criminal Damage'
       END AS crime_type
FROM Q
WHERE IUCR=0820 OR IUCR=0486 OR IUCR=1320
\end{lstlisting}

\mypar{Q2}
Find all crime ids, case numbers, longitudes and latitudes of crimes within a retangular area containing Chicago WaterTower.

\begin{lstlisting}
SELECT id, case_number, Longitude, Latitude
FROM crime
WHERE Longitude BETWEEN -87.674 AND -87.619
      AND Latitude BETWEEN 41.892 AND 41.903
\end{lstlisting}

\mypar{Q3}
This is a query over the graffiti dataset. Q3 returns all street addresses and zip codes for graffiti removal requests that are currently open.

\begin{lstlisting}
SELECT Street_Address, ZIP_Code, status
FROM graffiti
WHERE status='Open'
\end{lstlisting}

\mypar{Q4}
Find all dates, addresses and zip codes of food inspections of restaurants that passed, but were identified as ``high risk''. \\

\begin{lstlisting}
SELECT Inspection_Date, address, zip
FROM foodinspections
WHERE results = 'Pass w/ Conditions'
      AND risk = 'Risk 1 (High)'
\end{lstlisting}

\mypar{Q5}
For each crime id, case numbers, and IUCRs of crimes, find all status, service request numbers and community areas from graffiti removal requests where both take place in district 8 and the graffiti removal request's location is within 100 coordinate units of the crime's location.\\

\begin{lstlisting}
SELECT c.ID,
       c.Case_Number,
       c.IUCR,
       g.status,
       g.Service_Request_Number,
       g.Community_Area
FROM
    (SELECT * FROM graffiti
     WHERE police_district = 8) g,
    (SELECT * FROM crime
     WHERE district = '008') c
WHERE c.X_Coordinate < g.X_Coordinate + 100
  AND c.X_Coordinate > g.X_Coordinate - 100
  AND c.Y_Coordinate < g.Y Coordinate + 100
  AND c.Y_Coordinate > g.Y_Coordinate - 100
\end{lstlisting}

\mypar{MayBMS-$Q_{P1}$}
Find probability for a randomly chosen tuple.\\
\begin{lstlisting}
SELECT conf()
FROM buffalo
WHERE index=1;
\end{lstlisting}

\mypar{MayBMS-$Q_{P2}$}
Find probability for shooting in each district for a random range of incidence.\\
\begin{lstlisting}
SELECT *
FROM
	(SELECT "District_shooting",
			index,
			conf()
	 FROM bp20
	 GROUP BY "District_shooting",index) x
WHERE index<2000
  AND index>650
  AND "District_shooting"='BD';
\end{lstlisting}

\mypar{MayBMS-$Q_{P3}$}
Find probabilities for all incidences that happened in the same district with same type of shooting for a random incident.
\begin{lstlisting}
SELECT xind, yind,p
FROM
	(SELECT y.index AS yind,
			x.index AS xind,
			x."District_shooting" AS xds,
			y."District_shooting" AS yds,
			x."Type_shooting" AS xts,
			y."Type_shooting" AS yts,
			conf() AS p
	 FROM bp20 y, bp20 x
	 GROUP BY y.index,
	 		  x.index,
			  y."District_shooting",
			  x."District_shooting",
			  x."Type_shooting",
			  y."Type_shooting") z
WHERE xds=yds
  AND xts=yts
  AND xind=692;
\end{lstlisting}


}

\subsection{Utility of Query Answers}\label{sec:exp-utility}
\reva{We claim that \abbrBGQP and, thus also \abbrUADBs, have better utility than certain answers, as additional, useful possible answers are included in the result.
  The next experiment supports this claim quantitatively by contrasting the under-approximation of Libkin with \abbrUADBs{} evaluating two methods for extracting a best-guess world.}
To start, we create an incomplete database for which we have the ground truth (i.e., a ``correct'' possible world).
This world (denoted as $D_{ground}$) is created by processing a source dataset to remove all rows with nulls.
We next use $D_{ground}$ to create an incomplete database $\pdb$ by replacing a random set of attribute values with nulls, varying the fraction of attributes replaced from 0\% (deterministic input), to 50\%.
\reva{Then, we derive a best guess world $D_{clean}$ from $\pdb$ by either using a standard missing value imputation algorithm (we refer to this method as \textbf{BGQP}) or randomly pick a replacement value (random-guess query processing or \textbf{RGQP}).}
We evaluate queries over $\pdb$ and $D_{clean}$ using Libkin and \abbrUADBs{} respectively, and compare the result with the ground truth $D_{ground}$.
\Cref{fig:utility} shows both precision (fraction of results in $D_{ground}$) and recall (fraction of $D_{ground}$ in the results) as we vary the level of uncertainty.
Libkin's method always under-approximates, guaranteeing 100\% precision.
However, recall is much lower than for \abbrUADBs{} and drops rapidly when the amount of uncertainty is increased.
In contrast, the precision and recall achieved by \abbrUADBs{} remains between 80-90\% \reva{for BGQP}, even when half of all attribute values are uncertain.
This confirms our conjecture that certain answers are often more dissimilar to the actual answers than answers obtained over a \termBGW.
\reva{Compared with BGQP, RGQP  produces less accurate and complete results. However, its precision is ~70\% or higher and its recall is still much higher than Libkin.}

\BG{OLD: From the Synthetic data experiments as shown in figure~\ref{fig:runtime_uncert} and figure~\ref{fig:runtime_data} we know that certain answer approaches can have performance close to deterministic query processing. However possible answers are also useful in some context. We now compare the utility of certain answers with \abbrUADBs by comparing precision and recall of the query results. In the experiment we want to simulate a use-case workflow for the \abbrUADBs while keep the ground to truth for the measurement. We use three real world datasets including Income Survey, Buffalo News and Business License. We first introduce various amount of errors to the original full datasets to created a dirty instance of them. Then we repair those dirty datasets using data cleaning tools to create a cleaned instance for each dirty datasets. We run selection queries over each cleaned datasets to measure precision in terms of the percentage of result that is similar to ground truth result and recall in terms of the percentage of ground truth result that is correctly returned by querying over cleaned data. Figure~\ref{fig:utility} shows the precision and recall versus amount of uncertainty introduced to the dirty datasets. We can see that uncertain answers always have 100\% precision since certain answers will not include any possible results. This means certain answers will have bad performance in terms of recall as lots of results will be missing when comparing with ground truth. \abbrUADB is based on one possible world which preserves not only certain answers but also possible answers through queries. So although \abbrUADB has less precision comparing to certain answers, our approach achieved a good precision and recall balance. Also, depend on the data cleaning method, we maintained close to 80\% precision and recall when uncertainty level dropped down to 50\%.
}




\section{Conclusions and Future Work}\label{sec:conclusion}
We propose \abbrUADBs as a novel and efficient way to represent uncertainty as bounds on certain answers. Being based on $\semK$-relations, our approach applies
to the incomplete version of any data model that can be encoded as $\semK$-relations including set and bag semantics.
\abbrUADBs are backward compatible with many uncertain data models such as \termTIs{}, \abbrXDBs{} and \abbrCtables{}.
In future work, we plan to extend our approach with attribute level annotations to encode certainty at finer granularity and to support larger classes of queries, e.g., queries involving negation and aggregation.
Furthermore, it would be interesting to  study uncertain versions of semirings beyond sets and bags in more depth and explore new use cases such as inconsistent query answering and querying the result of data exchange.




\bibliographystyle{ACM-Reference-Format}
\bibliography{uadb.bib,oliver.bib}

\appendix

\newcommand{\proofSeparator}{
\vspace{-3mm}
\textcolor[gray]{0.7}{\hrule}
\vspace{-1mm}
}

\section{Proofs}
\label{sec:proofs}

\begin{proof}[Proof of Lemma~\ref{lem:poss-world-is-homomorphism}]
   Proven by substitution of definitions. 
%
  \begin{align*}
    \pwWorld{i}(\pwZero) &= \pwZero[i] = \zeroK
    &\pwWorld{i}(\pwOne) &= \pwOne[i] = \oneK
  \end{align*}\\[-8mm]
  \begin{align*}
    \pwWorld{i}(\pwe{k_1} \pwAdd \pwe{k_2}) &= (\pwe{k_1} \pwAdd \pwe{k_2})[i] = \pwe{k_1}[i] \addK \pwe{k_2}[i] \\ &= \pwWorld{i}(\pwe{k_1}) \addK \pwWorld{i}(\pwe{k_2})\\
    \pwWorld{i}(\pwe{k_1} \pwMult \pwe{k_2}) &= (\pwe{k_1} \pwMult \pwe{k_2})[i] = \pwe{k_1}[i] \multK \pwe{k_2}[i] \\ &= \pwWorld{i}(\pwe{k_1}) \multK \pwWorld{i}(\pwe{k_2}) \qedhere
  \end{align*}

\end{proof}

\proofSeparator

\begin{proof}[Proof of Theorem~\ref{theo:UA-preserve-approx}]
  We first prove that the possible world $\db = \pwWorld{i}(\pdb)$ for some $i$ encoded by $\uadb$ is preserved by queries. We have to show that for any query $\query$ we have $h_{\dOfName}(\query(\uadb)) = \pwWorld{i}(\query(\pdb))$. Since a \abbrUADB is the direct product of two semirings, $h_{\dOfName}$ is a homomorphism. Also by construction we have $h_{\dOfName}(\uadb) = \db$. Using these facts and Lemma~\ref{lem:poss-world-is-homomorphism} we get:
  \begin{align*}
    h_{\dOfName}(\query(\uadb))
    = &\query(h_{\dOfName}(\uadb))
    = \query(\db)\\
    = &\query(\pwWorld{i}(\pdb))
    = \pwWorld{i}(\query(\pdb))
  \end{align*}
  For the same argument as above, $h_{\cOfName}$ is a homomorphism, so $\query(h_{\cOfName}(\uadb)) = h_{\cOfName}(\query(\uadb))$. Since according to Theorem~\ref{theo:UA-are-c-sound}  queries over labelings preserve the under\hyp{}approximation of certain annotations this implies the theorem.
\end{proof}

\proofSeparator
\begin{proof}[Proof of Lemma~\ref{lem:order-factors-through-operations}]
  \proofpara{$\addK$}
    Based on the definition of $\ordK$, if $k \ordK k'$ then there exists $k''$ such that $k \addK k'' = k'$. Thus, $k_3 = k_1 \addK {k_1}'$ and $k_4 = k_2 \addK {k_2}'$ for some ${k_1}'$ and ${k_2}'$. Also, $(k_1 \addK k_2) \ordK (k_1 \addK k_2) \addK k''$ for any $k''$ and we get:
    \begin{align*}
      k_1 \addK k_2 \ordK (k_1 \addK k_2) \addK ({k_1}' \addK {k_2}')
      = k_3 \addK k_4
    \end{align*}
  \proofpara{$\multK$}
  The proof for multiplication $\multK$ is similar.

\noindent\resizebox{1\linewidth}{!}{
    \begin{minipage}{1.1\linewidth}
    \begin{align*}
      &(k_1 \multK k_2)\\
      \ordK &(k_1 \multK k_2) \addK (k_1 \multK {k_2}') \addK ({k_1}' \multK k_2) \addK ({k_1}' \multK {k_2}')\\
      = &(k_1 \addK {k_1}') \multK (k_2 \addK {k_2}') = k_3 \multK k_4 \qedhere
    \end{align*}
  \end{minipage}
  }
\end{proof}

\proofSeparator

\begin{proof}[Proof of Lemma~\ref{lem:super-all}]
  Recall that $\pwAdd$ and $\pwMult$ are defined element-wise and that $\pwCertain(\pwe{k}) = \GlbK(\pwe{k})$. Furthermore, $k_1 \ordK k_2$ iff $\exists k': k_1 \addK k' = k_2$. Consider an arbitrary $\pwe{k_1}, \pwe{k_2} \in \pwkDom$.  Let $k_{glb_1} = \GlbK(\pwe{k_1})$ and $k_{glb_2} = \GlbK(\pwe{k_2})$. Based on the definition of $\GlbK$ this implies that for any $i$, $k_{glb_1} \ordK \pwe{k_1}[i]$ which in turn implies that $\pwe{k_1}[i] = k_{glb_1} \addK k'$ for some $k'$. Analog, we can find a $k''$ such that $\pwe{k_2}[i] = k_{glb_2} \addK k''$.


  \proofpara{Superadditivity}
Let $k_{glb} = \GlbK(\pwe{k_1} \pwAdd \pwe{k_2})$. We are going to prove that $k_{glb_1} \addK k_{glb_2}$ is a lower bound for $(\pwe{k_1} \pwAdd \pwe{k_2})$, i.e., that $\forall i \in \pwDom: k_{glb_1} \addK k_{glb_2} \ordK (\pwe{k_1} \pwAdd \pwe{k_2})[i]$. Since, $k_{glb}$ is the greatest lower bound this implies that $k_{glb_1} \addK k_{glb_2} \ordK k_{glb}$. Consider an arbitrary $i \in \pwDom$. Based on the discussion above we have:

\noindent\resizebox{1\linewidth}{!}{
  \begin{minipage}{1.05\linewidth}
\begin{align*}
&(\pwe{k_1} \pwAdd \pwe{k_2})[i] = \pwe{k_1}[i] \addK \pwe{k_2}[i] = k_{glb_1} \addK k' \addK k_{glb_2} \addK k''\\
= &(k_{glb_1}  \addK k_{glb_2}) \addK k' \addK k'' \geqK k_{glb_1} \addK k_{glb_2}
\end{align*}
\end{minipage}
}
Thus, $k_{glb_1} \addK k_{glb_2}$ is a lower bound and since $k_{glb_1} = \pwCertain(\pwe{k_1})$ and $k_{glb_2} = \pwCertain(\pwe{k_2})$  it follows that $\pwCertain$ is superadditive:
\begin{align*}
\pwCertain(\pwe{k_1}) \addK \pwCertain(\pwe{k_2}) \ordK \pwCertain(\pwe{k_1} \pwAdd \pwe{k_2})
\end{align*}


\proofpara{Supermultiplicativity}
We use an analogous argument to prove supermultiplicativity. Let $k_{glb} = \pwCertain(\pwe{k_1} \multK \pwe{k_2})$. We will prove that $k_{glb_1} \multK k_{glb_2}$ is a lower bound for $(\pwe{k_1} \pwMult \pwe{k_2})$ which implies supermultiplicativity. Consider $i \in \pwDom$:

\noindent\resizebox{1\linewidth}{!}{
  \begin{minipage}{1.15\linewidth}
\begin{align*}
&(\pwe{k_1} \pwMult \pwe{k_2})[i]
= (k_{glb_1} \addK k') \multK (k_{glb_2} \addK k'')\\
= &(k_{glb_1} \multK k_{glb_2}) \addK (k_{glb_1} \multK k'') \addK (k' \multK k_{glb_2})  \addK (k' \multK k'')\\
\geqK &(k_{glb_1} \multK k_{glb_2}) \qedhere
\end{align*}
\end{minipage}
}
\end{proof}

\proofSeparator

\ifbool{ShowDetailedProofs}{
\begin{proof}[Proof of Lemma~\ref{lem:k-rel-preserves-soundness-over-c-correct}]
  Consider an $\raPlus$ query $\query$ and $\pdb$ a $\pwK$-database.
  To prove preservation of c-soundness, we have to show that the result of $\query(\TUL)$ is a c-sound labeling for $\query(\pdb)$, i.e., that for any tuple $\tup$ we have $\query(\TUL)(\tup) \ordK \pwCertain(\query(\pdb), \tup)$.
  Recall that $\raPlus$ queries over $\pwK$-relations and queries over $\semK$-labelings  are defined using the semiring addition and multiplication operations. Hence, the claim \[\query(\TUL)(\tup) \ordK \pwCertain(\query(\pdb), \tup)\] follows immediately from the superadditivity and supermultiplicativity of $\pwCertain$ (Lemma~\ref{lem:super-all}) and the fact that $\TUL$ is a c-correct labeling. 
\end{proof}
}

\begin{proof}[Proof of Theorem~\ref{theo:UA-are-c-sound}]
    Since $\TUL$ is a c-sound labeling,  for any tuple $\tup$ we have $\TUL(\tup) \ordK \pwCertain(\pdb, \tup)$. We have to prove that for any 
    $\tup$ we have $\query(\TUL)(\tup) \ordK \pwCertain(\query(\pdb), \tup)$.
%
%
%
%
  For that we show that for any $k_1, k_2 \in \semK$ and $\pwe{k_3}, \pwe{k_4} \in \pwkDom$ such that $k_1 \ordK \pwCertain(\pwe{k_3})$ and $k_2 \ordK \pwCertain(\pwe{k_4})$, we have $(k_1 \addK k_2) \ordK \pwCertain(\pwe{k_3} \pwAdd \pwe{k_4})$ and $k_1 \multK k_2 \ordK \pwCertain(\pwe{k_3} \pwMult \pwe{k_4})$.

   \begin{align*}
     k_1 \addK k_2
    \ordK &\pwCertain(\pwe{k_3}) \addK \pwCertain(\pwe{k_4}) \tag{by  Lemma~\ref{lem:order-factors-through-operations}}\\
    \ordK &\pwCertain(\pwe{k_3} \pwAdd \pwe{k_4}) \tag{by Lemma~\ref{lem:super-all}}\\[1mm]
    k_1 \multK k_2
    \ordK &\pwCertain(\pwe{k_3}) \multK \pwCertain(\pwe{k_4}) \tag{by Lemma~\ref{lem:order-factors-through-operations}}\\
    \ordK &\pwCertain(\pwe{k_3} \pwMult \pwe{k_4})  \tag{by Lemma~\ref{lem:super-all}}
  \end{align*}

  Since by assumption the input labeling is c-sound, we have $\TUL(\tup) \ordK \pwCertain(\pdb, \tup)$ for any tuple $\tup$. Thus, based on the property we have just proven and the fact the $\semK$-relational query semantics is defined based on the operations of semirings only, this implies that for any tuple $\tup$: $\query(\TUL)(\tup) \ordK \pwCertain(\query(\pdb), \tup)$.
  Thus,  $\query(\TUL)$ is a c-sound labeling for $\query(\pdb)$.
\end{proof}

\proofSeparator

\begin{proof}[Proof of Theorem~\ref{theo:h-homo-UA-c-table}]
  Let $\TUL = \doTULC(\pdb)$. 	A tuple $\tup$ is labeled as certain iff $\phi_{\pdb}(\tup)$ is in CNF and $\models \phi_{\pdb}(t)$, which means the expression $\phi_{\pdb}$ is a tautology. By definition of \abbrCtables, a tuple $\tup$ exists in a possible world if $\phi_{\pdb}(\tup)$ evaluates to true in that possible world. Thus, 
$\tup$ must exist in all possible worlds if $\phi_{\pdb}(\tup)$ is a tautology and $\TUL$ is c-sound.
\end{proof}

\proofSeparator

\begin{proof}[Proof of Theorem~\ref{theo:h-homo-UA-TUL-correct}]
  Trivially holds, since a tuple is certain iff it is not optional and has only one alternative. Even though multiple x-tuples may share an alternative, the independence of x-tuples guarantees that this does not lead to additional certain tuples.  
\end{proof}




\section{Datasets}\label{sec:datasets}

\textbf{Building Violations:} Building violations issued by Chicago's Department of Buildings from 2006 to the present.  \footnote{\url{https://data.cityofchicago.org/Buildings/Building-Violations/22u3-xenr}}
\textbf{Shootings in Buffalo:}  Shootings in Buffalo during the year 2016.  \footnote{\url{http://projects.buffalonews.com/charts/shootings/index.html}}
\textbf{Business Licenses:}  Current and active business licenses issued by the Department of Business Affairs and Consumer Protection.  \footnote{\url{https://data.cityofchicago.org/Community-Economic-Development/Business-Licenses-Current-Active/uupf-x98q}}
\textbf{Chicago Crime:}  Reported incidents of crime occurring from 2001 to present.  \footnote{\url{https://data.cityofchicago.org/Public-Safety/Crimes-2001-to-present/ijzp-q8t2}}
\textbf{Contracts:}  Contracts and modifications awarded by the City of Chicago since 1993.  \footnote{\url{https://data.cityofchicago.org/Administration-Finance/Contracts/rsxa-ify5}}
\textbf{Food Inspections:}  Inspections of restaurants and other food establishments in Chicago from January 1, 2010 to the present.  \footnote{\url{https://data.cityofchicago.org/Health-Human-Services/Food-Inspections/4ijn-s7e5}}
\textbf{Graffiti Removal:}  All graffiti removal requests, open and closed, since January 1, 2011.  \footnote{\url{https://data.cityofchicago.org/Service-Requests/311-Service-Requests-Graffiti-Removal/hec5-y4x5}}
\textbf{Building Permits:}  All types of structural permits in San Francisco from Jan 1, 2013-Feb 25th 2018.  \footnote{\url{https://www.kaggle.com/aparnashastry/building-permit-applications-data/data}}
\textbf{Public Library Survey:}  Over 25 years worth of research publications about the Public Libraries Survey.  \footnote{\url{https://www.imls.gov/research-evaluation/data-collection/public-libraries-survey/explore-pls-data/pls-data}}
\textbf{NHANES:} Family level information on income sources, monthly income, and family cash assets.  \footnote{\url{https://wwwn.cdc.gov/Nchs/Nhanes/2013-2014/INQ_H.htm}}




\end{document}